\documentclass[11pt]{article}
\usepackage[T1]{fontenc}
\usepackage{lmodern}
\usepackage[margin=1in]{geometry}
\usepackage{amsmath,amssymb,amsthm,dsfont}
\usepackage{authblk}
\usepackage{aliascnt}
\usepackage{xcolor}
\usepackage{booktabs}
\usepackage{capt-of}
\usepackage{needspace}
\usepackage{placeins}
\usepackage{caption}
\usepackage{tikz}
\usetikzlibrary{arrows.meta,decorations.pathreplacing}
\usepackage[colorlinks=true,allcolors=blue!60!black]{hyperref}
\usepackage[capitalise,noabbrev]{cleveref}
\hypersetup{pdftitle={Black Hole Radiation Decoding in the Haar Random Oracle Model},pdfauthor={Ezekiel Cochran and Atul Mantri},bookmarksnumbered=true}

\newtheorem{theorem}{Theorem}[section]
\newaliascnt{lemma}{theorem}
\newtheorem{lemma}[lemma]{Lemma}
\aliascntresetthe{lemma}
\newaliascnt{corollary}{theorem}
\newtheorem{corollary}[corollary]{Corollary}
\aliascntresetthe{corollary}
\newtheorem*{informaltheorem}{Main result (informal)}

\newcommand{\reg}[1]{\mathsf{#1}}
\newcommand{\ket}[1]{\lvert #1\rangle}
\newcommand{\bra}[1]{\langle #1\rvert}
\newcommand{\proj}[1]{\lvert #1\rangle\!\langle #1\rvert}
\newcommand{\braket}[2]{\langle #1\vert #2\rangle}
\newcommand{\ketbra}[2]{\lvert #1\rangle\!\langle #2\rvert}
\newcommand{\Tr}{\operatorname{Tr}}
\newcommand{\id}{\mathds{1}}
\newcommand{\Img}{\operatorname{Im}}
\newcommand{\Dom}{\operatorname{Dom}}
\newcommand{\Sym}{\operatorname{Sym}}

\title{Black Hole Radiation Decoding in the Haar Random Oracle Model}
\author{Ezekiel Cochran\thanks{ecochran@vt.edu}\ }
\author{Atul Mantri\thanks{atulmantri@vt.edu}}
\affil{Department of Computer Science, Virginia Tech, USA 24061}
\date{}

\begin{document}
\maketitle

\begin{abstract}
We prove optimal query bounds for recovering a single qubit from black-hole radiation in the Haar random oracle model,
when the remaining black hole contains at most one sixteenth of the system's qubits.
Recovery with any constant Haar-averaged advantage over trivial decoding requires queries proportional to the Hilbert-space dimension of the remaining black hole.
The lower bound is unconditional for decoders chosen independently of the sampled unitary,
allows arbitrary computation between queries, and access to $U, U^\dagger, U^*, U^\mathsf T$ along with the controlled variants.
An existing decoder using only forward and inverse queries attains a matching bound.
The proof uses a path recording oracle to compare real and maximally mixed states.
As applications, we obtain an efficiently preparable, statistically far, computationally indistinguishable (EFI) pair
and quantum commitments relative to a public Haar oracle,
as well as prove a tight linear rank lower bound for Uhlmann transformation on a fixed-target family.
\end{abstract}

\clearpage
\setcounter{tocdepth}{2}
\tableofcontents
\setcounter{tocdepth}{2}
\clearpage

\section{Introduction}
\label{sec:introduction}

Hayden and Preskill~\cite{HP07} studied an old black hole after roughly half its initial entropy has been emitted. Under rapidly mixing unitary dynamics, they showed that information carried by a quantum system sent into the black hole becomes information-theoretically recoverable after a short interval of further radiation, provided the observer has access to both the earlier and subsequent radiation.

\cite{AMPS13} described a procedure acting on the old radiation
which recovers entanglement with the next single qubit of radiation:
If an observer could perform this computation quickly enough, she could then jump into the black hole
and discover the partner mode of that qubit, ostensibly violating the monogamy of entanglement.

There have been many interesting ideas addressing this apparent paradox,
including foundational ideas in many categories of modern physics.
Forefront among these is the computational complexity argument of \cite{HH13},
who showed that while the firewall-creating computation is information-theoretically possible with knowledge of the black hole dynamics,
computational complexity evidence
suggests that it cannot be carried out before the black hole in question evaporates,\footnote{
An efficient decoder for arbitrary correctable dynamics would imply $\textsf{SZK}\subseteq\textsf{BQP}$; assuming otherwise gives worst-case hardness, rather than a statement about every physical black hole.
}
providing a possible computational obstruction to the proposed experiment..
Further, Aaronson~\cite[Section~6.5.1, Theorem~6.5.3]{Aar16} derived hardness from injective quantum-secure one-way functions,
and Brakerski~\cite{Bra23} showed an existential equivalence between hard-to-decode black hole states and EFI.

We study this question in the Haar-random oracle model (HROM) introduced by Chen and Movassagh~\cite{CM24}, which allows forward and inverse queries.
A reference qubit $\reg B$ is initially maximally entangled with one input qubit, while the other $n-1$ input qubits are in $\ket0$.
A Haar-random unitary $U$ acts on these $n$ qubits, and its output is split into radiation $\reg R$ and a remaining black hole $\reg H$.
The observer receives $\reg R$ and must produce a qubit $\reg A$ entangled with $\reg B$; the observer cannot act on $\reg B$ or $\reg H$.
The observer can query the same unitary used to prepare the radiation, as shown in \cref{fig:decoding-model};
we allow arbitrary computation between queries and study how many queries are needed.

\begin{figure}[htbp]
\centering
\begin{tikzpicture}[x=1cm,y=1cm,font=\small,gate/.style={draw,fill=white,line width=.6pt,inner sep=3pt},wire/.style={line width=.6pt}]
  \draw[wire] (0,2.05) node[left] {$\ket0$} -- (9.3,2.05) node[right] {$\reg B$};
  \draw[wire] (0,1.15) node[left] {$\ket0$} -- (2.3,1.15);
  \draw[wire] (0,.2) node[left] {$\ket{0^{n-1}}$} -- (2.3,.2);
  \node[gate,minimum width=.42cm,minimum height=.42cm] at (.7,2.05) {$\mathrm H$};
  \draw[wire] (1.45,2.05) -- (1.45,1.15);
  \fill (1.45,2.05) circle (1.6pt);
  \draw[wire,fill=white] (1.45,1.15) circle (.11);
  \draw[wire] (1.45,1.04) -- (1.45,1.26);
  \draw[wire] (1.34,1.15) -- (1.56,1.15);
  \node[gate,minimum width=.8cm,minimum height=1.65cm] at (2.7,.675) {$U$};
  \draw[wire] (3.1,1.15) -- (5.3,1.15) node[midway,above=3pt] {$\reg R$};
  \draw[wire,-{Latex[length=1.6mm]}] (3.1,.2) -- (3.3,.2) -- (3.3,-1.5) -- (9.3,-1.5) node[right] {$\reg H$} -- (9.3,-1.87);
  \draw[wire] (4.1,-.65) node[left] {$\ket0$} -- (5.3,-.65) node[midway,above=3pt,font=\footnotesize] {ancillae};
  \node[gate,minimum width=2.45cm,minimum height=2.2cm,align=center] at (6.525,.25) {$\mathcal D^{U^?}$\\[3pt]\footnotesize $q$ queries};
  \draw[wire] (7.75,1.15) -- (9.3,1.15) node[right] {$\reg A$};
  \draw[wire,-{Latex[length=1.6mm]}] (7.75,-.65) -- (8.5,-.65) -- (8.5,-1.02);
  \draw[decorate,decoration={brace,amplitude=4pt},line width=.6pt] (9.95,2.05) -- (9.95,1.15) node[midway,right=6pt,align=center,font=\footnotesize] {Bell-state\\fidelity};
\end{tikzpicture}
\caption{
The decoding task.
    A Hadamard gate and a CNOT prepare the Bell pair; $U$ acts on its second qubit and $n-1$ qubits initialized to $\ket{0}$,
producing $\reg{RH}$. The decoder receives $\reg R$, but has no access to $\reg B,\reg H$.
It makes $q$ queries on its own registers, which may include $\reg R$:
the preparation query is excluded from $q$.
The downward arrows mark registers that are traced out, namely the decoder's non-output bits and $\reg H$.
}
\label{fig:decoding-model}
\end{figure}
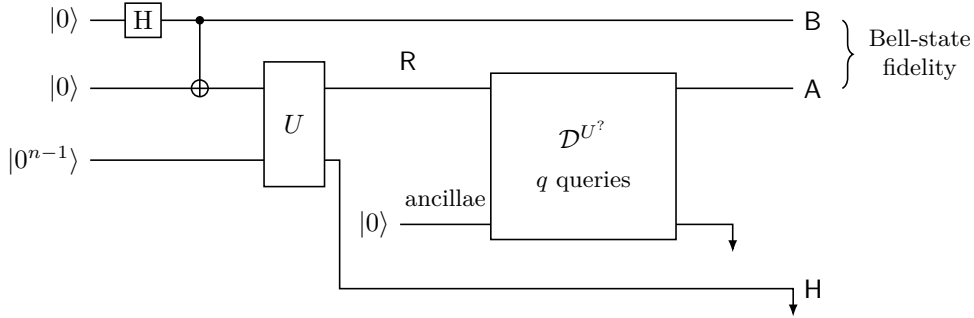

This differs slightly from the \cite{HH13} formulation, where $U$ acts on the all-$\ket{0}$ initial state and $\reg{B}$ is a third output register.
There, when $|\reg R|$ is sufficiently larger than $|\reg{BH}|$,
Page's theorem implies for a typical $U$ that the reduced state on the system $\reg B \reg H$ is nearly maximally mixed and recovery from $\reg R$ is information-theoretically possible \cite{HH13}.
This remains true here: decoupling implies that a recovery map with high fidelity exists ~\cite{DBWR14}.
A decoder with access to both $\reg R$ and $\reg H$, along with inverse access, could simply apply $U^\dagger$ to the original output, recovering the original input qubit with one query.
With access only to $\reg R$ and not $\reg H$, this trivial strategy fails.
We allow queries to $U,U^\dagger,U^*,U^{\mathsf T}$, including their controlled versions, and show that recovery still requires exponentially many queries in the regimes stated below.
The circuit and its initial auxiliary state must be chosen independently of the sampled $U$.

\subsection{Results}
\label{sec:main-results}

We measure recovery by $\mathbb E_{U \sim \mu_n}\bra{\Phi^+} \rho_\reg{BA} \ket{\Phi^+}$, the squared Bell-state fidelity of $\reg{BA}$ averaged over $U$.
Outputting a qubit independent of $\reg R$ gives fidelity $1/4$.
Our lower bounds apply to every constant improvement over this value.

\Needspace{11\baselineskip}
\begin{informaltheorem}
Let $|\reg R|\geq|\reg H|$ and fix $\delta>0$.
Every decoder whose circuit and initial auxiliary state are independent of $U$, and whose Haar-averaged Bell-state fidelity is at least $1/4+\delta$, requires
\[
  \begin{array}{ll}
    \Omega(2^{|\reg H|}) & \text{queries with forward-only access},\\[3pt]
    \Omega\!\left(\min\{2^{|\reg H|},2^{n/16}\}\right) & \text{queries with all four types of access}.
  \end{array}
\]
Both bounds allow controlled queries and arbitrary computation between queries.
The second also allows a query to coherently select among $U,U^\dagger,U^*,U^{\mathsf T}$ and the identity.
\end{informaltheorem}

The two lower bounds scale the same when $|\reg H|\leq n/16$.
Outside this regime, the $2^{n/16}$ cutoff comes from the approximation to the Haar oracle used in the proof.
The formal bounds and their implications are in \cref{thm:all-haar,cor:forward-haar,cor:query-lower-bounds}.
The same lower bounds hold if a fixed decoder succeeds on any constant fraction of Haar oracles, by \cref{cor:typical-decoding}.

We demonstrate the decoding bounds by showing bounds for the easier distinguishing task,
in which an algorithm receives $\reg{BR}$ and access to $U$ and must distinguish the joint state in the real experiment from the maximally mixed state on $\reg{BR}$.
A successful decoder implies a successful distinguisher by acting on $\reg R$ and measuring $\reg{BA}$ onto the $\ket{\Phi^+}$ Bell state.
The test accepts with the decoder's fidelity in the real experiment and with probability $1/4$ in the mixed\footnote{
In similar arguments this is often called the ``ideal'' or ``hybrid'' experiment, especially when it models the desirable behavior of a cryptographic primitive.}
experiment.

We discuss a matching upper bound in \cref{sec:decodability-upper}, where
for $|\reg R|\geq|\reg H|+1$ we specialize the Uhlmann algorithm of Utsumi et al.~\cite[Theorem~13]{UNWT25} with an oracle-independent cutoff to obtain an $O(2^{|\reg H|})$-query decoder using forward and inverse queries, with average fidelity at least $1-2\cdot2^{(|\reg H|-|\reg R|)/2}-\epsilon$ for any fixed $0<\epsilon<1/2$.
For $|\reg H|\leq n/16$ and $|\reg H|\to\infty$, it achieves every fixed target fidelity in $(1/4,1)$.
Hence our lower bound with all four query types is tight in this regime.

We also discuss several applications in \cref{sec:applications}, including a linear rank lower bound for the Uhlmann transformation task
and the security of using the Haar-random oracle to generate EFI pairs and quantum commitments.
An Uhlmann transformation aligns two purifications by acting only on the accessible subsystem.
For inaccessible marginals $\rho,\tau$, the relevant rank is $r=\operatorname{rank}(\sqrt\rho\sqrt\tau)$.
The maximum root overlap is $\|\sqrt\rho\sqrt\tau\|_1$; its square is the maximum squared overlap.
With $|\reg R|=15|\reg H|$, our source and a fixed target give a family requiring $\Omega(r)$ preparation queries,
matching the known $O(r)$ upper bound at fixed positive error.
The lower bound allows squared-fidelity error up to $1/8$ on Haar average and all four query types including coherent selection,
as long as the algorithm is chosen independently of $U$.

Further, the joint state on $\reg{BR}$ and the maximally mixed state form an EFI pair in our public-oracle model:
they are efficiently preparable, statistically far apart, and computationally indistinguishable against algorithms with oracle-independent circuits and advice.
We use the oracle-relative version of the EFI definition~\cite[Definition~3.1]{BCQ23}, averaging the distinguishing gap over the public Haar oracle.
We bound the expected absolute acceptance gap for every such polynomial-query algorithm with all four query types.
The standard EFI construction then also gives a statistically binding bit commitment~\cite[Lemma~4.7]{BCQ23}.
The Haar encoding itself also commits to an arbitrary qubit: the sender sends $\reg R$ and later reveals $\reg H$.
Our comparison proves computational hiding~\cite[Definition~4.3]{GJMZ23}, and decoupling proves statistical swap binding~\cite[Definition~4.2]{GJMZ23}.

\subsubsection{Organization}
\Cref{sec:preliminaries} defines the access model and sets up the experiments,
and \cref{sec:query-bounds} states the query bounds, along with the reduction from decoding to distinguishing.
\Cref{sec:trace-distance-proof-sketch} proves the trace-norm bound between the marginals in the path recording experiments, which is our main technical contribution.
\Cref{sec:applications} presents the decoder, rank lower bound, and cryptographic consequences.
The appendices give the full commutator calculations and application proofs.

\subsection{Related work}
\label{sec:related-work}
\paragraph{Strong random unitaries.}
Schuster et al.~\cite[Theorems~1 and~2]{SMLBH25} define strong unitary designs from Haar in experiments with a bounded number of queries to $U$, $U^\dagger$, $U^*$, and $U^{\mathsf T}$, and strong pseudorandom unitaries (PRUs) by indistinguishability for efficient experiments with these query types. 
At fixed design order and approximation error, their Theorem~1 gives depth $O(\log n)$ with ancillas and $O(\log^2 n)$ without ancillas.
Under their cryptographic assumptions, Theorem~2 gives strong PRUs of depth $O(\log n)$ with ancillas and $\operatorname{poly}(\log n)$ without ancillas. 

Their lower bounds~\cite[Proposition~1]{SMLBH25} concern circuit depth rather than query count.
In Section~5 they observe that a strong design encodes as well as Haar in the Hayden--Preskill experiment because the Yoshida--Kitaev decoder uses $U^*$, and note that this transfer applies only when decoding is efficient from query access to the four variants.

Our bounds show that decoding from the radiation alone is query-hard for Haar under all four query types.
Transfer to a strong design requires its order and access model to cover the complete experiment: a $q$-query decoder and challenge preparation use $q+1$ calls, or $4q+1$ after reducing selected queries to controlled queries.
An output trace-norm approximation error of $\eta$ changes the average Bell-test probability by at most $\eta/2$.
For a strong PRU, indistinguishability transfers computational hardness to efficient decoders within its security budget, without the secret circuit description.

In Section~6 they ask whether strong PRUs enable new cryptographic applications and suggest that ancilla-free constructions of commitments would be stronger evidence that natural dynamics have cryptographic properties. Our direct qubit commitment requires no secret key. The sender applies the public oracle $U$ to the message qubit and $n-1$ qubits initialized to $\ket0$, without additional registers for encoding.
The real EFI state is prepared by applying $U$ to one half of a Bell pair and $n-1$ qubits initialized to $\ket0$; the other half is the reference register $\reg B$. The bit-commitment protocol keeps a purification register for opening; for commitment to $0$, this register is padded so that both opening registers have the same size.

\paragraph{Radiation decoding.}
Kim, Tang, and Preskill~\cite[Definition~6.1]{KTP20} formulate the hypothesis that low-rank radiation is computationally indistinguishable from a maximally mixed state.
Our bounds give a quantitative analogue in the Haar-random oracle model.
The efficient Hayden--Preskill decoder of Yoshida and Kitaev~\cite{YK17} has access to an entangled partner of the initial black hole.
Here the other input qubits start in a known pure state, and no such partner is supplied to the decoder.

Yang and Engelhardt~\cite[Theorem~1]{YE23} study the hardness of learning Haar-random dynamics before a fresh input is supplied.
Their algorithm finishes querying before it receives the input, whereas our decoder receives the radiation before its queries.
Poremba, Ragavan, and Vaikuntanathan~\cite[Theorem~8.4]{PRV24} bound joint recovery of a classical message by two separated observers, each making one forward or inverse query.

\paragraph{Haar oracles and query access.}
Zhandry~\cite{Zha25} argues that unitary-oracle models should allow controlled access to the unitary, its inverse, conjugate, and transpose.
Tang and Wright~\cite[Theorem~1.1]{TW25} remove controlled queries for problems that do not depend on a global phase in the oracle.
Tang, Wright, and Zhandry~\cite[Theorem~1.3]{TWZ25} show that conjugate queries can give an exponential advantage over forward and inverse queries.
Their Theorem~1.7 also gives a commitment secure with forward and inverse access that one conjugate query can break.
Bartusek and Goldin~\cite{BG26} construct reusable unclonable encryption in the Haar-random oracle model with all four query types and prove a unitary reprogramming lemma.
These results motivate proving our decoding and commitment bounds in the stronger access model.

\paragraph{Recording oracles.}
Zhandry~\cite{Zha19} introduced compressed oracles for quantum queries to random functions.
Carolan~\cite{Car25} extends this approach to random permutations, and Carolan and Majenz~\cite{CM26} improve its soundness bound.
For Haar unitaries, Ma and Huang~\cite[Theorems~5 and~8]{MH24} give forward and bidirectional path-recording approximations, and Schuster et al.~\cite[Theorem~10]{SMLBH25} extend them to all four query types.
Grinko and Yoshida~\cite{GY25} give exact compressed oracles for compact-group representations using Clebsch--Gordan transforms, including all four query types.
Foxman et al.~\cite{FLMNW26} give an exact path-recording construction for closed subgroups of the unitary group.
Our proof uses the explicit approximate maps of Ma--Huang and Schuster et al.; the insertion and deletion estimates are proved for those maps.

\paragraph{Uhlmann transformation.}
Bostanci et al.~\cite[Theorem~10.3]{BEMPQY23} relate radiation decoding to Uhlmann transformation.
Utsumi et al.~\cite[Theorem~13 and Corollary~5]{UNWT25} give an $O(r)$ preparation-query upper bound at fixed fidelity error and a cube-root lower bound in the rank of the inaccessible marginals, at sufficiently small error.
They already apply their algorithm to Uhlmann decoders in Theorem~33 and to Haar-based quantum state merging in Theorem~34.
Our upper bound specializes the algorithm with one rank cutoff that works for every $U$.

Lee and Jo~\cite[Theorem~1.1 and Corollary~1.2]{LJ26} prove nearly quadratic sample lower bounds for estimating fidelity to a known rank-$r$ reference.
At sufficiently small fixed error, their lower bound for a maximally mixed reference also holds with purification copies.
Wang~\cite[Theorems~1.7 and~5.18]{Wan26} proves an $\Omega(r/\log^2r)$ query lower bound for fidelity estimation at sufficiently small fixed error.
The transformation-to-estimation reduction in Utsumi et al.~\cite[Proof of Corollary~5]{UNWT25} gives the same lower bound for Uhlmann transformation with forward and inverse preparation queries.
\Cref{app:fidelity-comparison} gives the reduction and checks its hypotheses.
Liu and Wang~\cite[Theorem~3.3 and Corollary~3.4]{LW26} obtain optimal fidelity estimation when one of the two states is pure, without knowing which one.
Our rank lower bound concerns two mixed inaccessible marginals.

Our lower bound is linear, allows squared-fidelity error up to $1/8$, and includes all four query types on a fixed-target Haar family.
The extension of the query-simulation theorem to all four types is left open in~\cite[Footnote~3]{TWZ25}.

\paragraph{Cryptographic consequences.}
Brakerski~\cite[Theorems~4.1 and~5.2]{Bra23} gives reductions between EFI pairs and radiation ensembles satisfying hardness and strong decodability conditions.
Brakerski, Canetti, and Qian~\cite[Lemma~4.7]{BCQ23} construct quantum bit commitments from EFI pairs.
Bostanci et al.~\cite[Theorem~8.10]{BEMPQY23} characterize commitments with strong statistical hiding and strong computational binding through average-case hard Uhlmann instances with almost equal marginals.
Our Uhlmann family has nearly equal inaccessible marginals on $\reg{BH}$ for most Haar oracles, together with the stated Haar-averaged query lower bound for transforming their purifications.
The separate real and mixed challenges on $\reg{BR}$ form a statistically far EFI pair.
This pair gives the computationally hiding, statistically binding bit commitment proved in \cref{cor:bit-commitment}.

\paragraph{Common Haar states and public unitaries.}
Chen, Coladangelo, and Sattath~\cite{CCS24} and Ananth, Gulati, and Lin~\cite{AGL24} study cryptography when parties receive copies of common Haar states.
They obtain pseudorandom states with bounded-copy security and quantum bit commitments in those models.
Goldin and Zhandry~\cite{GZ25} give conditions for lifting common-Haar-state separations to unitary-oracle separations.
Constructions of pseudorandom states and unitaries in quantum and Haar oracle models also appear in~\cite{Kre21,ABGL24,ABGL25,HY24}.
In particular, the forward-and-inverse constructions~\cite{ABGL25,HY24}, combined with the commitment and EFI implications~\cite{GJMZ23,BCQ23}, already give EFI pairs and commitments in that access model.
Here the adversary can query the public unitary used to prepare its challenge, with all four query types.

\paragraph{Entropy and entanglement.}
Bansal et al.~\cite[Theorem~6 and Appendix~M]{BMBKH24} give EFI pairs by tracing out part of a secret-key pseudorandom state ensemble.
Aaronson et al.~\cite{ABF24} construct indistinguishable pure-state ensembles with different entanglement entropies.
Bouland et al.~\cite{BFG24} obtain public-key pseudoentanglement under the Learning with Errors assumption.
Ryzov et al.~\cite{RGLE26} relate operational definitions of pseudoentanglement to EFI pairs.
Our entropy observation concerns mixed $\reg{BR}$ marginals and follows directly from their rank.

\section{Preliminaries}
\label{sec:preliminaries}

\subsection{Notation}

We write $[m]=\{0,\ldots,m-1\}$.
Registers are denoted by sans-serif letters.
For a register $\reg X$, its size $|\reg X|$ is its number of qubits.
We use $\|v\|_2$ for vector norm, $\|X\|$ for operator norm, and $\|X\|_1=\Tr\sqrt{X^\dagger X}$ for trace norm.
Trace distance is $\frac12\|\rho-\sigma\|_1$.
We write $\mathrm F(\rho,\sigma)=\|\sqrt\rho\sqrt\sigma\|_1$ for root fidelity.
For a pure target $\ket\psi$, squared fidelity is $\bra\psi\rho\ket\psi$.
Quantum channels contract trace distance and cannot decrease fidelity~\cite[Corollary~3.40 and Theorem~3.27]{Wat18}.
All logarithms are base two.

\subsection{The two experiments}
\label{sec:experiments}

Let $n\geq1$, $n=|\reg R|+|\reg H|$, and $x_{0,c}=c\,0^{n-1}$ for $c\in\{0,1\}$.
We identify an $n$-qubit output label with a pair $(r,h)\in[2^{|\reg R|}]\times[2^{|\reg H|}]$.
The oracle $U$ is an $n$-qubit unitary sampled from Haar measure.
The real experiment prepares
\begin{equation}
  \ket{\Psi^{\mathrm{real},U}_0}=\frac1{\sqrt2}\sum_{c\in\{0,1\}}\ket c_{\reg B}\bigl(U\ket{x_{0,c}}\bigr)_{\reg{RH}}. \label{eq:initial-real-state-u}
\end{equation}
The resulting challenge on $\reg{BR}$ is
\[
  \rho_{\reg{BR}}^U=\Tr_{\reg{H}}\proj{\Psi^{\mathrm{real},U}_0}.
\]
The distinguishing algorithm receives $\reg{BR}$ and has no access to $\reg H$.
The mixed experiment supplies $\id_{\reg{BR}}/2^{1+|\reg R|}$ instead, with the same oracle $U$.
The algorithm's query count $q$ excludes the one forward query used to prepare the real challenge.
Let $\rho_q^{\mathrm{real},U^?}$ and $\rho_q^{\mathrm{mixed},U^?}$ be its states immediately after the last query, averaged over $U$; for $q=0$, these are the supplied input states, including the auxiliary state.
Every bound also holds after the algorithm's final oracle-independent channel, including its complete computation when $q=0$, by contractivity.

A uniform bit $b$, sampled independently of $U$, selects the real experiment for $b=0$ and the mixed experiment for $b=1$.
A distinguishing algorithm outputs $a$ and has advantage $\Pr[a=b]-1/2$.
A decoder acts only on $\reg R$ and auxiliary registers and outputs a qubit $\reg A$.
Its squared Bell-state fidelity for a fixed oracle and its Haar average are
\[
  \begin{aligned}
    f_{\mathcal D}(U)&=\bra{\Phi^+}(\operatorname{id}_{\reg{B}}\otimes\mathcal D^U)(\rho_{\reg{BR}}^U)\ket{\Phi^+},\\
    f&=\mathbb E_U f_{\mathcal D}(U),\qquad \ket{\Phi^+}=\frac{\ket{00}+\ket{11}}{\sqrt2}.
  \end{aligned}
\]
Thus fidelity in the decoding task is the acceptance probability of the Bell-state test.
Outputting a fixed qubit independently of the radiation gives $f=1/4$.

\subsection{Oracle algorithms}
\label{sec:oracle-algorithms}

Conjugation and transposition are in the computational basis.
The allowed query types are
\[
  U^?\qquad\text{for } ?\in\{1,\dagger,*,\mathsf T\},
\]
with $U^1=U$.
Forward-only access permits $U$; bidirectional access permits $U,U^\dagger$; four-mode access permits every $U^?$.
A controlled query applies one type conditional on a control register; a selected query coherently chooses among the four types and the identity.
\Cref{cor:selected-queries} states its constant-factor cost.

The algorithm has an $n$-qubit query register $\reg Q$ and a workspace $\reg W$.
In the distinguishing experiment, $\reg W$ includes $\reg{BR}$ and the auxiliary registers.
Its initial auxiliary state is supplied in tensor product with the challenge and is independent of $U$ and $b$.
Every operation between queries is independent of the sampled $U$.
The bounds permit nonuniform circuits, arbitrary auxiliary states, and unbounded computation, subject to this independence condition.

For the proof, we purify the auxiliary state with a register independent of $U$ and $b$ and include it in $\reg W$.
Giving the algorithm this register can only help it.
By also storing measurement outcomes and discarded registers in the workspace, we may take all operations between queries to be isometries.
For a fixed sequence of query types, the enlarged algorithm has the form
\[
  A_{q+1}U_qA_q\cdots U_1A_1,\qquad U_i\in\{U,U^\dagger,U^*,U^{\mathsf T}\}.
\]
All $A_i$ are oracle-independent isometries on the algorithm's registers.
Storing measurement outcomes in a control register includes classical choices in the selected-query model; terminated branches select the identity.

\Needspace{11\baselineskip}
\section{Query lower bounds}
\label{sec:query-bounds}

The two experiments are difficult to distinguish with few queries.
The following state bounds quantify this for a fixed sequence of uncontrolled query types.
Their proofs are given in \cref{sec:trace-distance-proof-sketch}.

\begin{theorem}[Four query types]
\label{thm:all-haar}
Let $|\reg R|\geq|\reg H|$ and $q\geq0$.
For every algorithm in the two experiments of \cref{sec:experiments}, with circuit and initial auxiliary state independent of $U$ and a fixed sequence of $q$ uncontrolled queries to $U^?$,
\begin{equation}
  \|\rho_q^{\mathrm{real},U^?}-\rho_q^{\mathrm{mixed},U^?}\|_1\leq\min\left\{2,\ (25+24\sqrt2)\frac q{2^{|\reg H|}}+\frac{36(q+1)^2}{2^{n/8}}\right\}. \label{eq:all-haar-bound}
\end{equation}
\end{theorem}

Forward-only access gives a sharper comparison, depending only on $q/2^{|\reg H|}$.
Write $\rho_q^{j,U}$ for the same Haar-averaged output when every query is forward, for $j\in\{\mathrm{real},\mathrm{mixed}\}$.

\begin{corollary}[Forward-only bound]
\label{cor:forward-haar}
Under the same independence assumptions as \cref{thm:all-haar}, for $|\reg R|\geq|\reg H|$ and any $q\geq0$, forward-only queries satisfy
\begin{equation}
  \|\rho_q^{\mathrm{real},U}-\rho_q^{\mathrm{mixed},U}\|_1\leq\min\left\{2,\ (12+3\sqrt2)\frac q{2^{|\reg H|}}\right\}. \label{eq:forward-haar-bound}
\end{equation}
\end{corollary}

Both statements compare the complete experiments, in which the challenge and later queries use the same sampled oracle.
We now include controlled and selected queries and derive the distinguishing and decoding lower bounds.

\subsection{Controlled and selected queries}
\label{sec:controlled}

\begin{corollary}[Controlled queries]
\label{cor:controlled-queries}
The Haar bounds in \cref{eq:all-haar-bound,eq:forward-haar-bound} also hold for controlled queries, with the same query count.
\end{corollary}
\begin{proof}[Proof Sketch]
Tang--Wright~\cite[Theorem~1.1 and Claim~2.6]{TW25} replace each controlled query by one uncontrolled query of the same type.
The replacement circuit is independent of $U$ and its input.
For every fixed $U$ and input, including reference registers, its output is the original output averaged over one common random global phase of $U$.
The real challenge is unchanged under $U\mapsto e^{i\theta}U$, since its preparation acquires only a global phase; the mixed challenge is also unchanged.
Thus the same replacement applies to either challenge, including the oracle-dependent real input.
Haar measure is phase invariant, so both averaged output states and the query count are preserved.
The uncontrolled bounds therefore apply with the same query count.
\Cref{app:controlled-proof} gives the counter construction and the phase-invariance calculation.
\end{proof}

\begin{corollary}[Selected queries]
\label{cor:selected-queries}
For $q$ queries coherently selecting among $\id,U,U^\dagger,U^*,U^{\mathsf T}$, \cref{eq:all-haar-bound} holds with $q$ replaced by $4q$.
This includes adaptive choices of query type.
\end{corollary}
\begin{proof}
Let $P_0,\ldots,P_4$ be fixed pairwise orthogonal control projectors with sum $\id$.
Set $U_1=U,U_2=U^\dagger,U_3=U^*,U_4=U^{\mathsf T}$.
Then
\[
  \id\otimes P_0+\sum_{j=1}^4U_j\otimes P_j=\prod_{j=1}^4\bigl(\id\otimes(\id-P_j)+U_j\otimes P_j\bigr).
\]
On the range of $P_j$ for $1\leq j\leq4$, only the $j$th factor is nontrivial; on the range of $P_0$, all factors are the identity.
For each fixed projector $P_j$, we can compute its binary control into an auxiliary qubit, make the controlled query, and uncompute the qubit using oracle-independent gates.
Thus four controlled queries implement the selected query.
The claim follows from \cref{cor:controlled-queries}.
Measurement outcomes can be stored in the control, and terminated branches select the identity.
\end{proof}

\Needspace{10\baselineskip}
\subsection{Distinguishing and decoding}

\begin{corollary}[Distinguishing and decoding]
\label{cor:query-lower-bounds}
For $|\reg R|\geq|\reg H|$, a fixed positive Haar-averaged distinguishing advantage requires
\[
  \begin{aligned}
    q&=\Omega(2^{|\reg H|})&&\text{with forward-only access},\\
    q&=\Omega\!\left(\min\{2^{|\reg H|},2^{n/16}\}\right)&&\text{with all four query types}.
  \end{aligned}
\]
Both bounds allow controlled queries; the second also allows selected inverse, conjugate, and transpose queries.
The same bounds hold for a fixed positive decoding fidelity advantage over $1/4$.
\end{corollary}
\begin{proof}
Let $0\leq M\leq\id$ be the effect for guessing the real experiment and write $\Delta=\rho_q^{\mathrm{real},U^?}-\rho_q^{\mathrm{mixed},U^?}$.
Since $\Tr\Delta=0$, its positive and negative parts each have trace $\|\Delta\|_1/2$.
Consequently,
\[
  \Pr[a=b]-\frac12=\frac12\Tr[M\Delta]\leq\frac14\|\Delta\|_1.
\]
The state bounds give the query lower bounds for any fixed advantage, using $q+1\leq2q$ for $q\geq1$.
At $q=0$, Haar twirling gives equal averaged inputs.
\Cref{app:query-lower-bounds} gives explicit lower-bound constants.

For decoding, run the decoder and measure $\reg{BA}$ with the Bell projector.
In the real experiment this accepts with average probability $f$.
For each fixed $U$, the mixed output is $\id_{\reg{B}}/2\otimes\sigma_{\reg{A}}^U$ for a state $\sigma_{\reg{A}}^U$ of trace one.
Since $\Tr_{\reg{B}}\proj{\Phi^+}=\id_{\reg{A}}/2$, its acceptance probability is
\[
  \Tr\left[\proj{\Phi^+}\left(\frac{\id_{\reg{B}}}{2}\otimes\sigma_{\reg{A}}^U\right)\right]=\frac14\Tr\sigma_{\reg{A}}^U=\frac14.
\]
The resulting distinguishing advantage is $(f-1/4)/2$, with the same queries.
\end{proof}

\begin{corollary}[Decoding on a constant fraction of oracles]
\label{cor:typical-decoding}
Fix $p,\delta>0$ and let $|\reg R|\geq|\reg H|$.
The asymptotic decoding lower bounds in \cref{cor:query-lower-bounds} also hold for any $q$-query decoder, with circuit and initial auxiliary state independent of $U$, satisfying
\[
  \Pr_U[f_{\mathcal D}(U)\geq1/4+\delta]\geq p.
\]
The lower-bound constants may depend on $p,\delta$.
\end{corollary}
\begin{proof}[Proof sketch]
Using fresh challenges and auxiliary states, a distinguisher estimates whether the decoder's Bell-test probability is above or below $1/4$.
It then tests the original challenge and reverses the outcome if the estimated sign is negative.
A constant number of trials, depending only on $p,\delta$, gives a Haar-averaged acceptance gap at least $p\delta/2$.
The total query count is $O(q+1)$, so \cref{cor:query-lower-bounds} applies.
\Cref{app:typical-decoding} proves the sign estimate and checks conditional independence.
\end{proof}

\section{Proof sketch of the query bounds}
In this section we outline the proof of \cref{thm:all-haar,cor:forward-haar},
with some explicit calculations pushed to the appendix.
\label{sec:trace-distance-proof-sketch}
A distinguishing algorithm is supplied a state on $\reg{BR}$ as a challenge,
either prepared as in \cref{fig:decoding-model}, or in the mixed state $\id_\reg{BR}/2^{1 + |\reg R|}$.
The real and mixed challenges have the same Haar average before any query,
but the marginals can become distinguishable as queries are made.
In the real and mixed experiments respectively, write $\rho_q^{\mathrm{real},U^?},\rho_q^{\mathrm{mixed},U^?}$ for the algorithm's averaged states after $q$ queries to $U^?$, where $?\in\{1,\dagger,*,\mathsf T\}$.
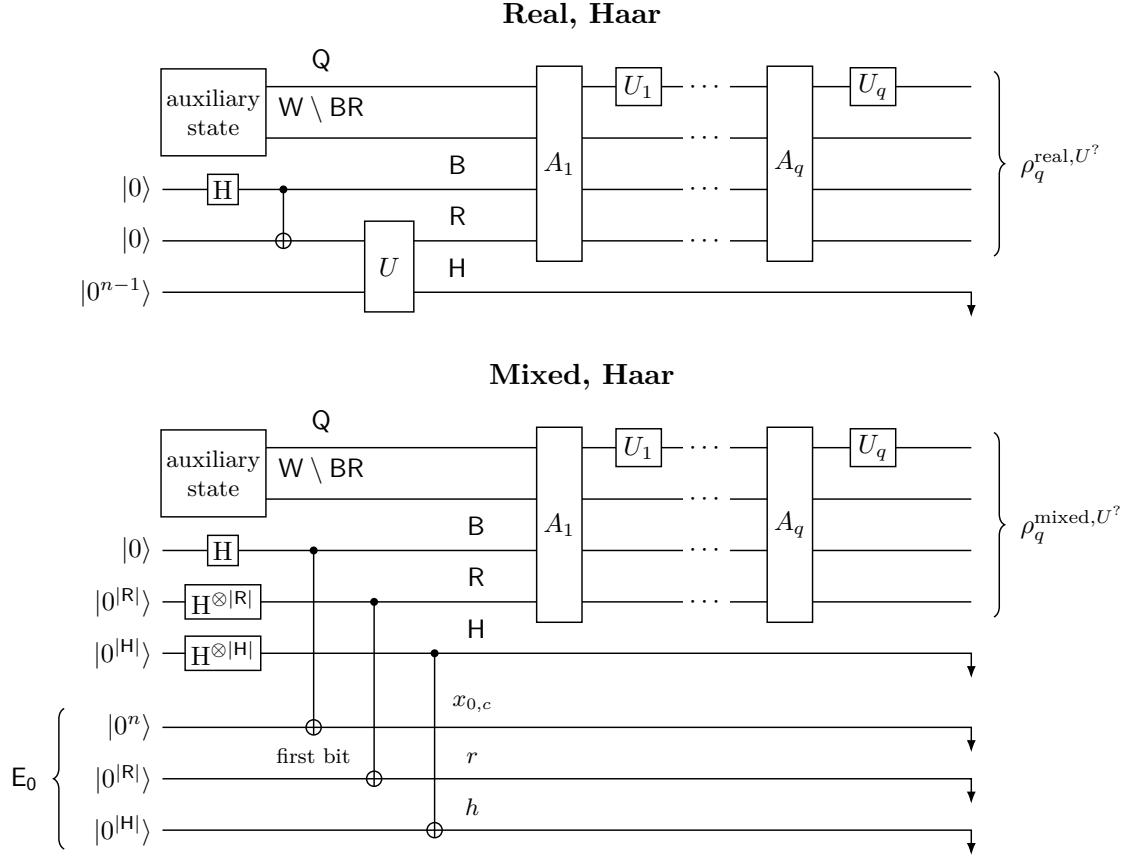
\begin{figure}[!htbp]
\centering
\textbf{Real, Haar}\par\smallskip
\begin{tikzpicture}[x=1cm,y=1cm,font=\small,gate/.style={draw,fill=white,line width=.55pt,inner sep=2pt},wire/.style={line width=.55pt},hidden/.style={draw=black,line width=.55pt}]
  \path (-2.2,0) -- (13.3,0);
  \foreach \y in {3.25,2.57} {\draw[wire] (1.35,\y) -- (10.7,\y);}
  \node[gate,minimum width=1.35cm,minimum height=1.15cm,align=center,font=\footnotesize] at (.675,2.91) {auxiliary\\state};
  \node[above=2pt] at (2.1,3.25) {$\reg Q$};
  \node[above=2pt] at (2.1,2.57) {$\reg W\setminus\reg{BR}$};
  \draw[wire] (0,1.89) node[left] {$\ket0$} -- (10.7,1.89);
  \draw[wire] (0,1.21) node[left] {$\ket0$} -- (10.7,1.21);
  \draw[wire] (0,.53) node[left] {$\ket{0^{n-1}}$} -- (2.675,.53);
  \draw[hidden,-{Latex[length=1.6mm]}] (3.325,.53) -- (10.7,.53) -- (10.7,.2);
  \node[gate,minimum width=.4cm,minimum height=.4cm] at (.8,1.89) {$\mathrm H$};
  \draw[wire] (1.6,1.89) -- (1.6,1.21);
  \fill (1.6,1.89) circle (1.5pt);
  \draw[wire,fill=white] (1.6,1.21) circle (.10);
  \draw[wire] (1.6,1.11) -- (1.6,1.31);
  \draw[wire] (1.5,1.21) -- (1.7,1.21);
  \node[gate,minimum width=.65cm,minimum height=1.2cm] at (3,.87) {$U$};
  \node[above=2pt] at (3.9,1.89) {$\reg B$};
  \node[above=2pt] at (3.9,1.21) {$\reg R$};
  \node[above=2pt,text=black] at (3.9,.53) {$\reg H$};
  \node[gate,minimum width=.6cm,minimum height=2.58cm] at (5.25,2.23) {$A_1$};
  \node[gate,minimum width=.6cm,minimum height=.5cm] at (6.3,3.25) {$U_1$};
  \foreach \y in {3.25,2.57,1.89,1.21} {\node[fill=white,inner sep=2pt] at (7.2,\y) {$\cdots$};}
  \node[gate,minimum width=.6cm,minimum height=2.58cm] at (8.3,2.23) {$A_q$};
  \node[gate,minimum width=.6cm,minimum height=.5cm] at (9.4,3.25) {$U_q$};
  \draw[decorate,decoration={brace,amplitude=4pt},line width=.55pt] (11,3.45) -- (11,1.01) node[midway,right=6pt] {$\rho_q^{\mathrm{real},U^?}$};
\end{tikzpicture}
\par\medskip
\textbf{Mixed, Haar}\par\smallskip
\begin{tikzpicture}[x=1cm,y=1cm,font=\small,gate/.style={draw,fill=white,line width=.55pt,inner sep=2pt},wire/.style={line width=.55pt},hidden/.style={draw=black,line width=.55pt}]
  \path (-2.2,0) -- (13.3,0);
  \foreach \y in {3.25,2.57} {\draw[wire] (1.35,\y) -- (10.7,\y);}
  \node[gate,minimum width=1.35cm,minimum height=1.15cm,align=center,font=\footnotesize] at (.675,2.91) {auxiliary\\state};
  \node[above=2pt] at (2.1,3.25) {$\reg Q$};
  \node[above=2pt] at (2.1,2.57) {$\reg W\setminus\reg{BR}$};
  \draw[wire] (0,1.89) node[left] {$\ket0$} -- (10.7,1.89);
  \draw[wire] (0,1.21) node[left] {$\ket{0^{|\reg R|}}$} -- (10.7,1.21);
  \foreach \y in {.53,-.45,-1.13,-1.81} {\draw[hidden,-{Latex[length=1.6mm]}] (0,\y) -- (10.7,\y) -- (10.7,\y-.33);}
  \node[left] at (0,.53) {$\ket{0^{|\reg H|}}$};
  \node[left] at (0,-.45) {$\ket{0^n}$};
  \node[left] at (0,-1.13) {$\ket{0^{|\reg R|}}$};
  \node[left] at (0,-1.81) {$\ket{0^{|\reg H|}}$};
  \node[gate,minimum width=.4cm,minimum height=.4cm] at (.8,1.89) {$\mathrm H$};
  \node[gate,minimum width=1cm,minimum height=.4cm] at (.8,1.21) {$\mathrm H^{\otimes|\reg R|}$};
  \node[gate,minimum width=1cm,minimum height=.4cm] at (.8,.53) {$\mathrm H^{\otimes|\reg H|}$};
  \draw[wire] (2,1.89) -- (2,-.45);
  \fill (2,1.89) circle (1.5pt);
  \draw[wire,fill=white] (2,-.45) circle (.10);
  \draw[wire] (2,-.55) -- (2,-.35);
  \draw[wire] (1.9,-.45) -- (2.1,-.45);
  \node[below=4pt,font=\scriptsize] at (2,-.45) {first bit};
  \draw[wire] (2.8,1.21) -- (2.8,-1.13);
  \fill (2.8,1.21) circle (1.5pt);
  \draw[wire,fill=white] (2.8,-1.13) circle (.10);
  \draw[wire] (2.8,-1.23) -- (2.8,-1.03);
  \draw[wire] (2.7,-1.13) -- (2.9,-1.13);
  \draw[wire] (3.6,.53) -- (3.6,-1.81);
  \fill (3.6,.53) circle (1.5pt);
  \draw[wire,fill=white] (3.6,-1.81) circle (.10);
  \draw[wire] (3.6,-1.91) -- (3.6,-1.71);
  \draw[wire] (3.5,-1.81) -- (3.7,-1.81);
  \node[above=2pt] at (4.15,1.89) {$\reg B$};
  \node[above=2pt] at (4.15,1.21) {$\reg R$};
  \node[above=2pt,text=black] at (4.15,.53) {$\reg H$};
  \node[above=2pt,font=\footnotesize,text=black] at (4.1,-.45) {$x_{0,c}$};
  \node[above=2pt,font=\footnotesize,text=black] at (4.1,-1.13) {$r$};
  \node[above=2pt,font=\footnotesize,text=black] at (4.1,-1.81) {$h$};
  \draw[decorate,decoration={brace,amplitude=4pt,mirror},line width=.55pt,draw=black] (-1.3,-.2) -- (-1.3,-2.06) node[midway,left=6pt,text=black] {$\reg{E_0}$};
  \node[gate,minimum width=.6cm,minimum height=2.58cm] at (5.25,2.23) {$A_1$};
  \node[gate,minimum width=.6cm,minimum height=.5cm] at (6.3,3.25) {$U_1$};
  \foreach \y in {3.25,2.57,1.89,1.21} {\node[fill=white,inner sep=2pt] at (7.2,\y) {$\cdots$};}
  \node[gate,minimum width=.6cm,minimum height=2.58cm] at (8.3,2.23) {$A_q$};
  \node[gate,minimum width=.6cm,minimum height=.5cm] at (9.4,3.25) {$U_q$};
  \draw[decorate,decoration={brace,amplitude=4pt},line width=.55pt] (11,3.45) -- (11,1.01) node[midway,right=6pt] {$\rho_q^{\mathrm{mixed},U^?}$};
\end{tikzpicture}
\caption{The Haar experiments. The real preparation makes one forward call to $U$; the mixed preparation makes none. The algorithm then makes $q$ queries, with $U_i\in\{U,U^\dagger,U^*,U^{\mathsf T}\}$. Each gate $A_i$ acts on $\reg{QW}$; the dots omit intermediate gates and queries. In the mixed panel, the copies from $\reg R,\reg H$ are bitwise CNOTs, while the copy from $\reg B$ writes only the first bit of $x_{0,c}$. Downward arrows trace out the hidden registers. The output labels include the average over the same $U$ used throughout.}
\label{fig:haar-circuits}
\end{figure}
\FloatBarrier

The path-recording constructions of Ma and Huang~\cite{MH24} and Schuster et al.~\cite{SMLBH25} replace Haar queries by explicit maps
$V, \bar V$ so that $U, U^\dagger, U^*, U^\mathsf T$ are replaced by $V, V^\dagger, \bar V, \bar V^\dagger$ respectively.
$\bar V$ is the recording map for conjugate access, rather than the entrywise conjugate of $V$.
These maps store a ``path record'' of queried input/output pairs in inaccessible left and right query registers $\reg{E_L},\reg{E_R}$,
which store relation states $\ket L_\reg{E_L} \ket R_\reg{E_R}$ where each of $L, R$ is a multiset of pairs in $\left[2^n\right]^2$.
The size $|L|$ counts multiplicities, and $\operatorname{num}(L,e)$ is the multiplicity of the pair $e$.
The notation $L\setminus\{e\}$ removes one copy of $e$.
The union $L\cup R$ adds multiplicities, so $|L\cup R|=|L|+|R|$.
Write
\[
  \Img(L)=\{y:(x,y)\in L\text{ for some }x\},\qquad \Dom(L)=\{x:(x,y)\in L\text{ for some }y\}.
\]
The record basis consists of the normalized symmetric relation states of Ma and Huang~\cite[Section~2.1]{MH24},
with $\ket\emptyset$ denoting the empty relation.
For each fixed size, the states indexed by these multisets form an orthonormal basis of the symmetric subspace:
\cref{app:elementary} gives the explicit description.
A forward recording query can create a pair in $\reg{E_L}$ or delete a pair from $\reg{E_R}$; an inverse query reverses these roles.
The maps $A,B$ describe creation with a fresh output or input, and their adjoints describe deletion.
The barred maps exchange the two records, as summarized below.
Our $A,B$ are the maps denoted $V^L,V^R$ in Ma--Huang~\cite[Definition~25]{MH24}; $\bar A,\bar B$ correspond to the barred maps in Schuster et al.~\cite[Definition~25]{SMLBH25}.
For $|L\cup R|<2^n$, the two parts of the bidirectional oracle of Ma--Huang~\cite[Definition~25]{MH24} are
\begin{align}
  A\ket x_\reg Q\ket L_\reg{E_L}\ket R_\reg{E_R}&=\frac1{\sqrt{2^n-|\Img(L\cup R)|}}\sum_{y\notin\Img(L\cup R)}\ket y_\reg Q\ket{L\cup\{(x,y)\}}_\reg{E_L}\ket R_\reg{E_R},\nonumber\\
  B\ket y_\reg Q\ket L_\reg{E_L}\ket R_\reg{E_R}&=\frac1{\sqrt{2^n-|\Dom(L\cup R)|}}\sum_{x\notin\Dom(L\cup R)}\ket x_\reg Q\ket L_\reg{E_L}\ket{R\cup\{(x,y)\}}_\reg{E_R}. \nonumber
\end{align}
The corresponding creation maps for conjugate access are those of Schuster et al.~\cite[Definition~25]{SMLBH25}.
They have the same fresh-coordinate rule, with the destination records exchanged:
\begin{equation*}
  \bar A=\operatorname{SWAP}_{\reg{E_LE_R}}A\operatorname{SWAP}_{\reg{E_LE_R}},\qquad \bar B=\operatorname{SWAP}_{\reg{E_LE_R}}B\operatorname{SWAP}_{\reg{E_LE_R}}.
\end{equation*}
Each map is zero outside its defining domain. The path recording oracle is then
\begin{align}
  V&=A(\id-BB^\dagger)+(\id-AA^\dagger)B^\dagger,\nonumber\\
  \bar V&=\bar A(\id-\bar B\bar B^\dagger)+(\id-\bar A\bar A^\dagger)\bar B^\dagger, \nonumber
\end{align}
so that $V, V^\dagger, \bar V, \bar V^\dagger$ replace $U, U^\dagger, U^*, U^\mathsf T$, respectively.
The projection factors keep the $A$ and $B$ terms separate, so that $V$ and $\bar V$ are both partial isometries.
Each increases $|L \cup R|$ by at most one.

\cite{SMLBH25} show that this replacement results in approximation error of at most $O \left( \frac{q^2}{2^{n/8}} \right)$
in trace distance between the output state using Haar queries and the one using path recording queries.
For forward only queries, this is instead $O \left( \frac{q^2}{2^n} \right)$~\cite{MH24}.
Their approximation theorems apply to a complete circuit starting with empty records, including the one call that prepares the real challenge.
Our technical contribution is the comparison between the two path-recording states:
the chain is shown in \cref{fig:four-experiments}.

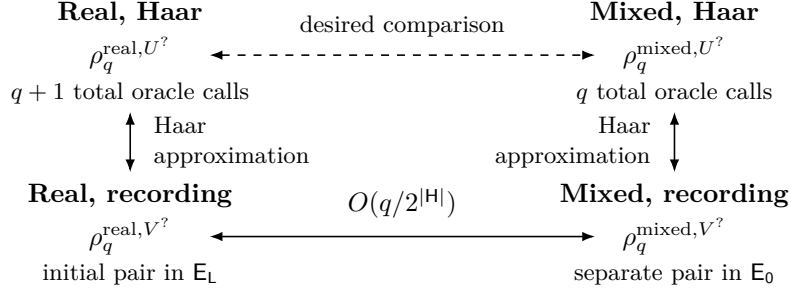
\begin{figure}[htbp]
\centering
\begin{tikzpicture}[x=1cm,y=1cm,font=\small,comparison/.style={line width=.6pt,{Latex[length=1.7mm]}-{Latex[length=1.7mm]}}]
  \node at (-3.6,.3) {\textbf{Real, Haar}};
  \node at (3.6,.3) {\textbf{Mixed, Haar}};
  \node at (-3.6,-.25) {$\rho_q^{\mathrm{real},U^?}$};
  \node at (3.6,-.25) {$\rho_q^{\mathrm{mixed},U^?}$};
  \node[font=\footnotesize] at (-3.6,-.75) {$q+1$ total oracle calls};
  \node[font=\footnotesize] at (3.6,-.75) {$q$ total oracle calls};
  \node at (-3.6,-2.1) {\textbf{Real, recording}};
  \node at (3.6,-2.1) {\textbf{Mixed, recording}};
  \node at (-3.6,-2.65) {$\rho_q^{\mathrm{real},V^?}$};
  \node at (3.6,-2.65) {$\rho_q^{\mathrm{mixed},V^?}$};
  \node[font=\footnotesize] at (-3.6,-3.15) {initial pair in $\reg{E_L}$};
  \node[font=\footnotesize] at (3.6,-3.15) {separate pair in $\reg{E_0}$};
  \draw[comparison,dashed] (-2.55,-.25) -- (2.55,-.25) node[midway,above=3pt,font=\footnotesize] {desired comparison};
  \draw[comparison] (-2.55,-2.65) -- (2.55,-2.65) node[midway,above=3pt] {$O(q/2^{|\reg H|})$};
  \draw[comparison] (-3.6,-1) -- (-3.6,-1.8) node[midway,right=5pt,align=left,font=\footnotesize] {Haar\\approximation};
  \draw[comparison] (3.6,-1) -- (3.6,-1.8) node[midway,left=5pt,align=right,font=\footnotesize] {Haar\\approximation};
\end{tikzpicture}
\caption{The four experiments for a fixed sequence of uncontrolled queries. The upper states are averaged over the same oracle used for challenge preparation and queries. The vertical comparisons use the path recording Haar approximations; our calculation proves the bottom comparison. In all cases $\rho$ denotes the reduced state on the algorithm's registers. \Cref{fig:haar-circuits,fig:recording-circuits} show all four circuits.}
\label{fig:four-experiments}
\end{figure}

\subsection{Connecting the path recording experiments}
Replacing the Haar-random oracle in \cref{fig:haar-circuits} with the path recording oracle, we obtain the circuits in \cref{fig:recording-circuits}.
\par\medskip
\noindent\begin{minipage}{\linewidth}
\centering
\textbf{Real, recording}
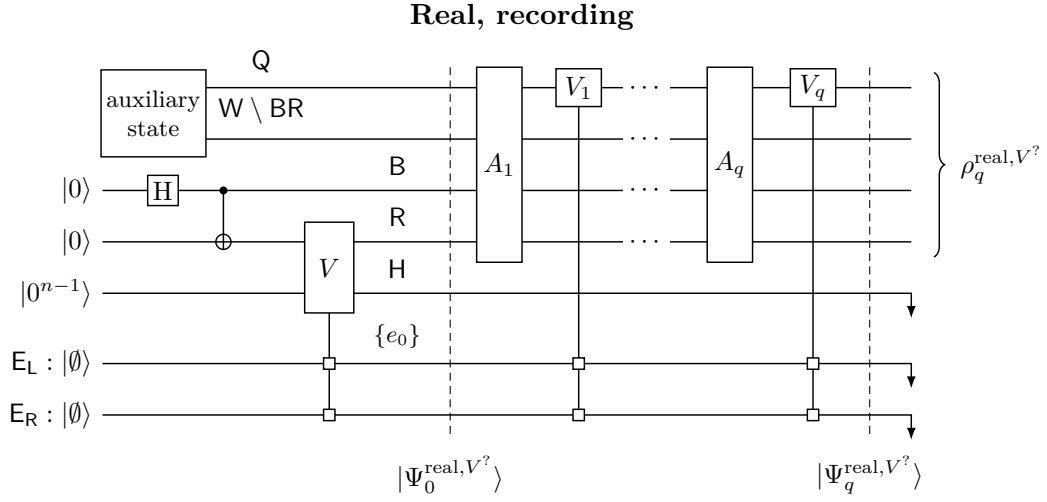
\par\smallskip
\begin{tikzpicture}[x=1cm,y=1cm,font=\small,gate/.style={draw,fill=white,line width=.55pt,inner sep=2pt},wire/.style={line width=.55pt},hidden/.style={draw=black,line width=.55pt},record/.style={draw=black,fill=white,line width=.55pt,minimum width=.15cm,minimum height=.15cm,inner sep=0pt}]
  \path (-2.2,0) -- (13.3,0);
  \foreach \y in {3.25,2.57} {\draw[wire] (1.35,\y) -- (10.7,\y);}
  \node[gate,minimum width=1.35cm,minimum height=1.15cm,align=center,font=\footnotesize] at (.675,2.91) {auxiliary\\state};
  \node[above=2pt] at (2.1,3.25) {$\reg Q$};
  \node[above=2pt] at (2.1,2.57) {$\reg W\setminus\reg{BR}$};
  \draw[wire] (0,1.89) node[left] {$\ket0$} -- (10.7,1.89);
  \draw[wire] (0,1.21) node[left] {$\ket0$} -- (10.7,1.21);
  \draw[wire] (0,.53) node[left] {$\ket{0^{n-1}}$} -- (2.675,.53);
  \draw[hidden,-{Latex[length=1.6mm]}] (3.325,.53) -- (10.7,.53) -- (10.7,.2);
  \draw[hidden,-{Latex[length=1.6mm]}] (0,-.4) node[left,text=black] {$\reg{E_L}:\ket\emptyset$} -- (10.7,-.4) -- (10.7,-.73);
  \draw[hidden,-{Latex[length=1.6mm]}] (0,-1.08) node[left,text=black] {$\reg{E_R}:\ket\emptyset$} -- (10.7,-1.08) -- (10.7,-1.41);
  \node[gate,minimum width=.4cm,minimum height=.4cm] at (.8,1.89) {$\mathrm H$};
  \draw[wire] (1.6,1.89) -- (1.6,1.21);
  \fill (1.6,1.89) circle (1.5pt);
  \draw[wire,fill=white] (1.6,1.21) circle (.10);
  \draw[wire] (1.6,1.11) -- (1.6,1.31);
  \draw[wire] (1.5,1.21) -- (1.7,1.21);
  \draw[hidden] (3,.87) -- (3,-1.08);
  \foreach \y in {-.4,-1.08} {\node[record] at (3,\y) {};}
  \node[gate,minimum width=.65cm,minimum height=1.2cm] at (3,.87) {$V$};
  \node[above=2pt] at (3.9,1.89) {$\reg B$};
  \node[above=2pt] at (3.9,1.21) {$\reg R$};
  \node[above=2pt,text=black] at (3.9,.53) {$\reg H$};
  \node[above=2pt,font=\footnotesize,text=black] at (3.9,-.4) {$\{e_0\}$};
  \draw[densely dashed,draw=black,line width=.4pt] (4.6,3.52) -- (4.6,-1.35) node[below=4pt,text=black] {$\ket{\Psi_0^{\mathrm{real},V^?}}$};
  \node[gate,minimum width=.6cm,minimum height=2.58cm] at (5.25,2.23) {$A_1$};
  \foreach \x in {6.3,9.4} {\draw[hidden] (\x,3.25) -- (\x,-1.08); \foreach \y in {-.4,-1.08} {\node[record] at (\x,\y) {};}}
  \node[gate,minimum width=.6cm,minimum height=.5cm] at (6.3,3.25) {$V_1$};
  \foreach \y in {3.25,2.57,1.89,1.21} {\node[fill=white,inner sep=2pt] at (7.2,\y) {$\cdots$};}
  \node[gate,minimum width=.6cm,minimum height=2.58cm] at (8.3,2.23) {$A_q$};
  \node[gate,minimum width=.6cm,minimum height=.5cm] at (9.4,3.25) {$V_q$};
  \draw[densely dashed,draw=black,line width=.4pt] (10.15,3.52) -- (10.15,-1.35) node[below=4pt,text=black] {$\ket{\Psi_q^{\mathrm{real},V^?}}$};
  \draw[decorate,decoration={brace,amplitude=4pt},line width=.55pt] (11,3.45) -- (11,1.01) node[midway,right=6pt] {$\rho_q^{\mathrm{real},V^?}$};
\end{tikzpicture}
\end{minipage}
\par\medskip
\noindent\begin{minipage}{\linewidth}
\centering
\textbf{Mixed, recording}\par\smallskip
\begin{tikzpicture}[x=1cm,y=1cm,font=\small,gate/.style={draw,fill=white,line width=.55pt,inner sep=2pt},wire/.style={line width=.55pt},hidden/.style={draw=black,line width=.55pt},record/.style={draw=black,fill=white,line width=.55pt,minimum width=.15cm,minimum height=.15cm,inner sep=0pt}]
  \path (-2.2,0) -- (13.3,0);
  \foreach \y in {3.25,2.57} {\draw[wire] (1.35,\y) -- (10.7,\y);}
  \node[gate,minimum width=1.35cm,minimum height=1.15cm,align=center,font=\footnotesize] at (.675,2.91) {auxiliary\\state};
  \node[above=2pt] at (2.1,3.25) {$\reg Q$};
  \node[above=2pt] at (2.1,2.57) {$\reg W\setminus\reg{BR}$};
  \draw[wire] (0,1.89) node[left] {$\ket0$} -- (10.7,1.89);
  \draw[wire] (0,1.21) node[left] {$\ket{0^{|\reg R|}}$} -- (10.7,1.21);
  \foreach \y in {.53,-.45,-1.13,-1.81,-2.63,-3.31} {\draw[hidden,-{Latex[length=1.6mm]}] (0,\y) -- (10.7,\y) -- (10.7,\y-.33);}
  \node[left] at (0,.53) {$\ket{0^{|\reg H|}}$};
  \node[left] at (0,-.45) {$\ket{0^n}$};
  \node[left] at (0,-1.13) {$\ket{0^{|\reg R|}}$};
  \node[left] at (0,-1.81) {$\ket{0^{|\reg H|}}$};
  \node[left] at (0,-2.63) {$\reg{E_L}:\ket\emptyset$};
  \node[left] at (0,-3.31) {$\reg{E_R}:\ket\emptyset$};
  \node[gate,minimum width=.4cm,minimum height=.4cm] at (.8,1.89) {$\mathrm H$};
  \node[gate,minimum width=1cm,minimum height=.4cm] at (.8,1.21) {$\mathrm H^{\otimes|\reg R|}$};
  \node[gate,minimum width=1cm,minimum height=.4cm] at (.8,.53) {$\mathrm H^{\otimes|\reg H|}$};
  \draw[wire] (2,1.89) -- (2,-.45);
  \fill (2,1.89) circle (1.5pt);
  \draw[wire,fill=white] (2,-.45) circle (.10);
  \draw[wire] (2,-.55) -- (2,-.35);
  \draw[wire] (1.9,-.45) -- (2.1,-.45);
  \node[below=4pt,font=\scriptsize] at (2,-.45) {first bit};
  \draw[wire] (2.8,1.21) -- (2.8,-1.13);
  \fill (2.8,1.21) circle (1.5pt);
  \draw[wire,fill=white] (2.8,-1.13) circle (.10);
  \draw[wire] (2.8,-1.23) -- (2.8,-1.03);
  \draw[wire] (2.7,-1.13) -- (2.9,-1.13);
  \draw[wire] (3.6,.53) -- (3.6,-1.81);
  \fill (3.6,.53) circle (1.5pt);
  \draw[wire,fill=white] (3.6,-1.81) circle (.10);
  \draw[wire] (3.6,-1.91) -- (3.6,-1.71);
  \draw[wire] (3.5,-1.81) -- (3.7,-1.81);
  \node[above=2pt] at (4.15,1.89) {$\reg B$};
  \node[above=2pt] at (4.15,1.21) {$\reg R$};
  \node[above=2pt,text=black] at (4.15,.53) {$\reg H$};
  \node[above=2pt,font=\footnotesize,text=black] at (4.1,-.45) {$x_{0,c}$};
  \node[above=2pt,font=\footnotesize,text=black] at (4.1,-1.13) {$r$};
  \node[above=2pt,font=\footnotesize,text=black] at (4.1,-1.81) {$h$};
  \draw[decorate,decoration={brace,amplitude=4pt,mirror},line width=.55pt,draw=black] (-1.3,-.2) -- (-1.3,-2.06) node[midway,left=6pt,text=black] {$\reg{E_0}$};
  \draw[densely dashed,draw=black,line width=.4pt] (4.6,3.52) -- (4.6,-3.56) node[below=4pt,text=black] {$\ket{\Psi_0^{\mathrm{mixed},V^?}}$};
  \node[gate,minimum width=.6cm,minimum height=2.58cm] at (5.25,2.23) {$A_1$};
  \foreach \x in {6.3,9.4} {\draw[hidden] (\x,3.25) -- (\x,-3.31); \foreach \y in {-2.63,-3.31} {\node[record] at (\x,\y) {};}}
  \node[gate,minimum width=.6cm,minimum height=.5cm] at (6.3,3.25) {$V_1$};
  \foreach \y in {3.25,2.57,1.89,1.21} {\node[fill=white,inner sep=2pt] at (7.2,\y) {$\cdots$};}
  \node[gate,minimum width=.6cm,minimum height=2.58cm] at (8.3,2.23) {$A_q$};
  \node[gate,minimum width=.6cm,minimum height=.5cm] at (9.4,3.25) {$V_q$};
  \draw[densely dashed,draw=black,line width=.4pt] (10.15,3.52) -- (10.15,-3.56) node[below=4pt,text=black] {$\ket{\Psi_q^{\mathrm{mixed},V^?}}$};
  \draw[decorate,decoration={brace,amplitude=4pt},line width=.55pt] (11,3.45) -- (11,1.01) node[midway,right=6pt] {$\rho_q^{\mathrm{mixed},V^?}$};
\end{tikzpicture}

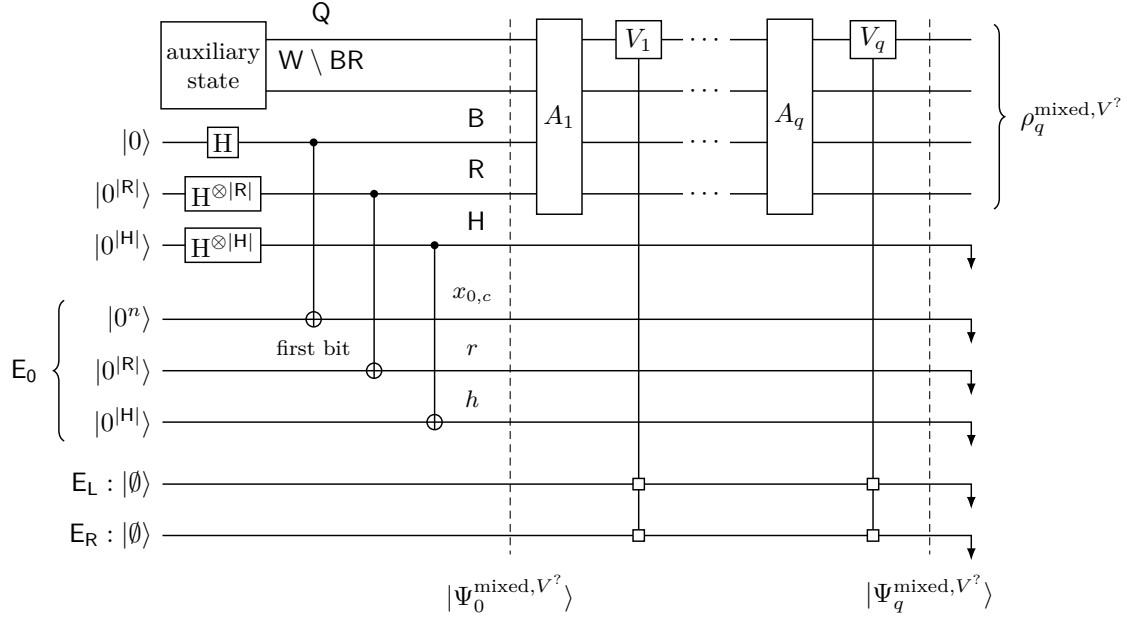
\captionof{figure}{The recording experiments. Each $V_i\in\{V,V^\dagger,\bar V,\bar V^\dagger\}$ acts on $\reg Q\reg{E_L}\reg{E_R}$, as indicated by the small squares connected to its gate. The algorithm gates leave all external registers untouched. The real preparation places $e_0$ in $\reg{E_L}$; the mixed preparation leaves $e_0$ in $\reg{E_0}$ and both oracle records empty. Dashed cuts mark the initial and final vectors. Downward arrows trace out hidden registers to give the displayed marginals, which may be subnormalized. The map $J$ compares the purifications in the proof; it is not a circuit gate.}
\label{fig:recording-circuits}
\end{minipage}
\par\medskip
For the challenge pair, we use $y_0=(r,h)$ and $e_0=(x_{0,c},y_0)$.
In general $c\in\{0,1\}$, $r\in[2^{|\reg R|}]$, $h\in[2^{|\reg H|}]$, and $x,y\in[2^n]$.
The initial path-recording purifications are
\begin{align}
  \ket{\Psi^{\mathrm{real},V^?}_0}&=\frac1{\sqrt{2^{n+1}}}\sum_{c,r,h}\ket{c,r}_{\reg{BR}}\ket h_{\reg{H}}\ket{\{e_0\}}_{\reg{E_L}}\ket\emptyset_{\reg{E_R}}, \label{eq:initial-real-purification}\\
  \ket{\Psi^{\mathrm{mixed},V^?}_0}&=\frac1{\sqrt{2^{n+1}}}\sum_{c,r,h}\ket{c,r}_{\reg{BR}}\ket h_{\reg{H}}\ket{e_0}_{\reg{E_0}}\ket\emptyset_{\reg{E_L}}\ket\emptyset_{\reg{E_R}}. \label{eq:initial-mixed-purification}
\end{align}
The first state is obtained by one $V$ query with empty records: $B^\dagger$ annihilates that input, so $V$ agrees with $A$ there.
Tracing out $\reg{HE_0}$ in the second state gives the mixed challenge $\id_{\reg{BR}}/2^{1+|\reg R|}$.
These differ only by the location of the challenge pair $e_0$, and to relate them
we use an insertion map $J$ following the merge-map method of Ananth et al.~\cite[Section~2.1]{ABGL25},
acting only on the external registers $\reg{E_0E_L}$:
\begin{equation}
  J\ket e_{\reg{E_0}}\ket L_{\reg{E_L}}\ket R_{\reg{E_R}}=\sqrt{\operatorname{num}(L,e)+1}\ket{L\cup\{e\}}_{\reg{E_L}}\ket R_{\reg{E_R}}. \label{eq:insertion-map-behavior}
\end{equation}
Note that it is not an isometry:
different choices of $\ket{e_0}_\reg{E_0} \ket L_\reg{E_L}$ can yield the same $\ket{L \cup \{e_0\}}_\reg{E_L}$ after insertion,
so insertion need not preserve inner products.
It aligns the initial states exactly: $J\ket{\Psi^{\mathrm{mixed},V^?}_0} = \ket{\Psi^{\mathrm{real},V^?}_0}$.
For $|L| = \ell$, symmetrizing the inserted challenge with the $\ell$ other positions gives $J = \sqrt{\ell + 1} \Pi_\text{sym}$,
which for the normalized relation states already symmetric on $\reg{E_L}$ gives the coefficient in \cref{eq:insertion-map-behavior}.
Since the symmetrization projection squares to itself, this gives
\begin{equation}
  J^\dagger J =\id+\Sigma,\qquad \Sigma=\sum_{i=1}^{\ell}\operatorname{SWAP}_{0i}, \label{eq:jdj-is-id-plus-sigma}
\end{equation}
where position $0$ is $\reg{E_0}$ and positions $1,\ldots,\ell$ form $\reg{E_L}$.
This is again spelled out in \cref{app:elementary}.
After $q$ queries, the two purifications are
\[
  \ket{\Psi_q^{j,V^?}}=V_qA_q\cdots V_1A_1\ket{\Psi_0^{j,V^?}},\qquad j\in\{\mathrm{real},\mathrm{mixed}\},
\]
with marginals on $\reg{QW}$ denoted by $\rho_q^{j,V^?}$.
Each recording query acts on $\reg Q\reg{E_L}\reg{E_R}$ while the algorithm gates act on $\reg{QW}$.
In the mixed experiment, $\reg{E_0}$ therefore remains isolated throughout the experiment.

Write $\mathcal A$ for the recording-algorithm operator, distinguishing it from primitive creation $A$, and let $\ket\psi$ be the mixed initial vector in \cref{eq:initial-mixed-purification}.
Insertion consumes $\reg{E_0}$, so the real and mixed marginals use different partial traces:
\[
  \begin{aligned}
    \rho_q^{\mathrm{real},V^?}&=\Tr_{\reg{HE_LE_R}}\proj{\mathcal A J\psi},\\
    \rho_q^{\mathrm{mixed},V^?}&=\Tr_{\reg{HE_0E_LE_R}}\proj{\mathcal A\psi},\\
    \widetilde\rho_q&=\Tr_{\reg{HE_LE_R}}\proj{J\mathcal A\psi}.
  \end{aligned}
\]
We bound the difference between the mixed marginal and $\widetilde\rho_q$, and then between $\widetilde\rho_q$ and the real marginal.
Both bounds apply to the full coherent states, grouping contributions with identical final labels.

Both error estimates use the same property of the mixed purification:
its challenge pair stays in $\reg{E_0}$, which neither the algorithm nor the recording maps touch.
Changing $h$ thus changes only the labels in $\reg H,\reg{E_0}$; the coefficient vectors on the algorithm and oracle records are unchanged.
For the mixed state immediately before query $t$, write
\[
  \ket{\phi'_t}=A_tV_{t-1}A_{t-1}\cdots V_1A_1\ket{\Psi_0^{\mathrm{mixed},V^?}}.
\]
The coefficient vectors at every stage can be grouped as
\[
  \ket{\alpha_{c,r}}=\sum_{\substack{x,L,R\\|L\cup R|\leq q}}\ket x_{\reg{Q}}\ket{\alpha_{c,r,x,L,R}}_{\reg W}\ket L_{\reg{E_L}}\ket R_{\reg{E_R}}.
\]
They are independent of $h$, and the mixed-experiment vector has the form
\begin{equation}
  \ket\phi=\frac1{\sqrt{2^{n+1}}}\sum_{c,r,h}\ket{\alpha_{c,r}}_{\reg{QWE_LE_R}}\ket h_{\reg{H}}\ket{e_0}_{\reg{E_0}}. \label{eq:mixed-experiment-state-family}
\end{equation}
The records start empty and each query adds at most one pair, giving the size bound.
The vector has norm at most one because all algorithm gates and query maps are contractions.
Each coefficient includes the sum of all contributions with the same $x,L,R$; no orthogonality of query histories is assumed.

Different values of $(c,r,h)$ give orthogonal labels in $\reg{E_0}$, even when the corresponding workspace vectors overlap.
The normalization in \cref{eq:mixed-experiment-state-family} therefore gives $\|\phi\|_2^2=2^{|\reg H|-(n+1)}\sum_{c,r}\|\alpha_{c,r}\|_2^2$.
Using $n=|\reg R|+|\reg H|$ and the orthogonality of the query and relation bases, we obtain
\begin{equation}
  \sum_{c,r,x,L,R}\|\alpha_{c,r,x,L,R}\|_2^2=\sum_{c,r}\|\alpha_{c,r}\|_2^2=2^{1+|\reg R|}\|\phi\|_2^2. \label{eq:coefficient-energy}
\end{equation}
For fixed $r$, the labels $(r,h)$ are distinct query labels, so also
\begin{equation}
  \sum_{c,r,h,L,R}\|\alpha_{c,r,(r,h),L,R}\|_2^2\leq\sum_{c,r,x,L,R}\|\alpha_{c,r,x,L,R}\|_2^2=2^{1+|\reg R|}\|\phi\|_2^2. \label{eq:diagonal-energy}
\end{equation}
The first identity will bound the marginal error by counting recorded pairs.
The second will bound deletion of the challenge pair, for which the query must equal $(r,h)$.

\subsubsection{The error due to $J$ not being an isometry}
In the first case, the error comes from insertion ``forgetting'' which pair in the record $\ket{L \cup \{e_0\}}_\reg{E_L}$ was the challenge,
causing more than one query history to contribute to a branch's norm.
For a vector $\ket\phi$ arising in the mixed experiment, define
\[
  K_\phi=\Tr_{\reg{HE_0E_LE_R}}[\Sigma\proj\phi].
\]
Because $J$ acts only on the traced out environment, using \cref{eq:jdj-is-id-plus-sigma} gives
\begin{equation}
  \Tr_{\reg{HE_LE_R}}[J\proj\phi J^\dagger]=\Tr_{\reg{HE_0E_LE_R}}[\proj\phi]+K_\phi. \label{eq:marginal-identity}
\end{equation}
Consequently, the desired comparison splits as
\begin{equation}
  \rho_q^{\mathrm{real},V^?}-\rho_q^{\mathrm{mixed},V^?}=\Tr_{\reg{HE_LE_R}}\!\left[\proj{\Psi_q^{\mathrm{real},V^?}}-J\proj{\Psi_q^{\mathrm{mixed},V^?}}J^\dagger\right]+K_{\Psi_q^{\mathrm{mixed},V^?}}. \label{eq:comparison-decomposition}
\end{equation}
\begin{lemma}[Marginal error from insertion]
\label{lem:insertion-error}
For every vector $\ket\phi$ of the form in \cref{eq:mixed-experiment-state-family}, with coefficient vectors independent of $h$ and supported on $|L\cup R|\leq q$,
\begin{equation}
  \|K_\phi\|_1\leq\frac q{2^{|\reg H|}}\|\phi\|_2^2,\qquad \|J\phi\|_2^2\leq\left(1+\frac q{2^{|\reg H|}}\right)\|\phi\|_2^2. \label{eq:insertion-error-bound}
\end{equation}
\end{lemma}
The proof is in \cref{app:insertion-error-bound-proof}. The bound uses the stronger invariant in \cref{eq:mixed-experiment-state-family}: all mixed-experiment operations leave $\reg{HE_0}$ untouched, so the coefficient vectors are independent of $h$. The proof bounds the crossed coherent terms by counting at most $q$ recorded pairs. Inaccessibility of $\reg H$ alone would not imply the required independence.

This only works because the algorithm has no access to $\ket h_\reg H$,
meaning that its choice of which query paths to put weight on must remain uncorrelated with $h$.

\subsubsection{The error due to $J$ not commuting with the algorithm}
In the second case, our goal is to show that $J$ ``almost commutes'' with the algorithm for this particular initial state,
in the sense that the difference between the algorithm's reduced states under the two orders is small in trace norm.
$J$ commutes exactly with any non-query algorithm operations, as those cannot touch the external registers, so the work is to bound its commutator with the algorithm's queries.
We use a hybrid argument, bounding the commutator of $V^?$ and $J$ for just one query (on relevant intermediate states in the mixed experiment),
and then use the triangle inequality to conclude a bound after all $q$ queries.

\begin{lemma}[Query commutator bound]
\label{lem:query-commutator}
Let $|\reg R|\geq|\reg H|$ and $1\leq q<2^{|\reg H|}$.
For every vector $\ket\phi$ of the form in \cref{eq:mixed-experiment-state-family}, with coefficient vectors independent of $h$ and supported on $|L\cup R|\leq q$,
\begin{equation}
  \|[V_i,J]\phi\|_2\leq\frac{24}{2^{|\reg H|}}\|\phi\|_2,\qquad V_i\in\{V,V^\dagger,\bar V,\bar V^\dagger\}. \label{eq:query-commutator}
\end{equation}
\end{lemma}

The maps $A,\bar A$ create a pair with a fresh output coordinate, while $B,\bar B$ require a fresh input coordinate.
Inserting the challenge first can exclude a coordinate allowed in the other order; it can also change the normalization denominator.
Before insertion, the creation maps spread weight over at least $2^n-q$ available coordinates; insertion can reduce this number by one, and \cref{eq:ratio} controls the normalization change.
For deletion, insertion can change a denominator, block a valid deletion, or make the inserted pair removable.
The proofs count contributions with identical final labels to control coherent interference.
The exceptional terms for $A^\dagger,\bar A^\dagger$ use the hidden-label coefficient bound; those for $B^\dagger,\bar B^\dagger$ use input matching and at most two challenge-bit choices.
\Cref{app:commutators} gives all eight calculations.

Write $\epsilon_X$ for the coefficient in the bound $\|[X,J]\phi\|_2\leq\epsilon_X\|\phi\|_2$: \cref{eq:eight-epsilons} lists these coefficients explicitly.

It remains to combine the creation and deletion estimates for a full query.
Expanding the projection factors gives $V=A-ABB^\dagger+B^\dagger-AA^\dagger B^\dagger$.
The commutator of each product is given by $[XY,J]=X[Y,J]+[X,J]Y$.
Each operator to the right of a commutator is the identity, one or two deletions, or a deletion followed by a creation.
These contractions leave $\reg{HE_0}$ unchanged and do not increase the record size, so they preserve the family in \cref{eq:mixed-experiment-state-family}.
The operators to the left are contractions as well.
Thus, for example,
\begin{equation}
  \|[V,J]\phi\|_2\leq(3\epsilon_A+\epsilon_B+\epsilon_{A^\dagger}+3\epsilon_{B^\dagger})\|\phi\|_2. \label{eq:v-commutator}
\end{equation}
The factors $3$ count the three occurrences of $[A,J]$ and $[B^\dagger,J]$ in the product expansion; $[B,J]$ and $[A^\dagger,J]$ each occur once.
The adjoint and barred queries obey the analogous bounds,
which all give coefficients at most $\frac{24}{2^{|\reg H|}}$.
\Cref{app:commutators,app:commutator-constants} give the full component proofs and coefficient bounds.

\subsection{Completing the state comparison}

\begin{lemma}[Path-recording comparison]
\label{lem:path-bound}
If $|\reg R|\geq|\reg H|$ and $q\geq0$, then
\begin{equation}
  \|\rho_q^{\mathrm{real},V^?}-\rho_q^{\mathrm{mixed},V^?}\|_1\leq\min\left\{2,\ (25+24\sqrt2)\frac q{2^{|\reg H|}}\right\}. \label{eq:all-path-bound}
\end{equation}
\end{lemma}
\begin{proof}
The query estimates now compare insertion before all queries with insertion after all queries.
Assume $1\leq q<2^{|\reg H|}$.

\paragraph{Comparing the final vectors.}
Since $J$ maps the initial mixed state to the real state and commutes with each algorithm gate $A_t$, the difference of the final vectors is
\[
  \ket{\Psi_q^{\mathrm{real},V^?}}-J\ket{\Psi_q^{\mathrm{mixed},V^?}}=\sum_{t=1}^qV_qA_q\cdots V_{t+1}A_{t+1}[V_t,J]\ket{\phi'_t}.
\]
This is the telescoping identity obtained by moving $J$ past one query at a time.
The $t$th summand is the error at query $t$, followed by all later operations; the product is the identity when $t=q$.
Every later operation is a contraction, and the mixed vector $\ket{\phi'_t}$ has norm at most one.
The triangle inequality and \cref{eq:query-commutator} therefore give
\[
  \|\Psi_q^{\mathrm{real},V^?}-J\Psi_q^{\mathrm{mixed},V^?}\|_2\leq\sum_{t=1}^q\|[V_t,J]\phi'_t\|_2\leq\frac{24q}{2^{|\reg H|}}.
\]

\paragraph{Passing to the marginals.}
To convert this vector estimate to a marginal-state bound, we use
\[
  \|\proj v-\proj w\|_1\leq(\|v\|_2+\|w\|_2)\|v-w\|_2.
\]
It follows by writing $\proj v-\proj w=\ketbra{v-w}v+\ketbra w{v-w}$ and bounding the two rank-one terms.
Here the real vector has norm at most one, while \cref{lem:insertion-error} bounds the inserted mixed vector by $\sqrt{1+q/2^{|\reg H|}}$.
The decomposition in \cref{eq:comparison-decomposition}, contractivity under partial trace, and the same lemma give
\begin{equation}
  \|\rho_q^{\mathrm{real},V^?}-\rho_q^{\mathrm{mixed},V^?}\|_1\leq\left[24\left(1+\sqrt{1+\frac q{2^{|\reg H|}}}\right)+1\right]\frac q{2^{|\reg H|}}. \label{eq:all-path-intermediate}
\end{equation}
Since $q<2^{|\reg H|}$, the square root is at most $\sqrt2$.
At $q=0$ the marginals agree.
For $q\geq2^{|\reg H|}$, their traces are at most one, so the trace-norm difference is at most $2$.
\end{proof}

The recording comparison gives the bottom bound in \cref{fig:four-experiments}; the two vertical bounds follow from the Haar approximation.

\paragraph{The Haar approximation.}
For any fixed $t$-query circuit with an oracle-independent input and empty records, let $\rho_t^{U^?}$ and $\rho_t^{V^?}$ denote the outputs of that same circuit using the Haar and recording oracles.
Schuster et al.~\cite[Theorem~10]{SMLBH25} give
\begin{equation}
  \|\rho_t^{U^?}-\rho_t^{V^?}\|_1\leq\frac{18t(t+1)}{2^{n/8}}. \label{eq:all-haar-simulation}
\end{equation}
Their bound is $9t(t+1)/2^{n/8}$ in trace distance, so the trace-norm coefficient is $18$.
They obtain this bound by extending the proof of Ma--Huang~\cite[Theorem~8]{MH24} to the barred query maps.
The Haar output is averaged over $U$, and reference registers may be included in the workspace.
An arbitrary oracle-independent initial state can be purified as in \cref{sec:oracle-algorithms}.
Fixed swaps put the preparation and interaction calls on one query register, without changing their number.

\begin{proof}[Proof of \cref{thm:all-haar}]
If $q\geq2^{n-1}$, then $q\geq2^{|\reg H|}$ because $|\reg R|\geq1$.
The right-hand side of \cref{eq:all-haar-bound} is then $2$, which bounds the trace-norm difference of any two states.
We may therefore assume $q\leq2^{n-1}-1$.
The complete real circuit uses $q+1$ queries, including the query preparing the challenge; the mixed circuit uses $q$.
Both circuits start with oracle-independent inputs and empty records.
The approximation in \cref{eq:all-haar-simulation} applies to each complete circuit, after which we trace out all registers hidden from the algorithm.
The two simulation errors are at most $18(q+1)(q+2)/2^{n/8}$ and $18q(q+1)/2^{n/8}$.
These two upper bounds sum to $36(q+1)^2/2^{n/8}$.
The triangle inequality and \cref{lem:path-bound} give the stated bound.
The distance between two states is also at most $2$ in trace norm.
\end{proof}

The sharper forward-only approximation gives \cref{cor:forward-haar}.

\begin{proof}[Proof of \cref{cor:forward-haar}]
Forward queries keep $\reg{E_R}$ empty and $V=A$, also after inserting $e_0$.
The comparison therefore uses only $[A,J]$, whose coefficient is at most $3/2^{|\reg H|}$.
The proof of \cref{lem:path-bound} then gives $[3(1+\sqrt2)+1]q/2^{|\reg H|}=(4+3\sqrt2)q/2^{|\reg H|}$ for $1\leq q<2^{|\reg H|}$.
Starting from empty records, the forward recording circuits reach only relations with distinct outputs; on these, $A$ agrees with the forward oracle of Ma--Huang~\cite[Definition~10]{MH24}.
Their Theorem~5 gives trace-norm error $4t(t-1)/(2^n+1)$ for a $t$-query circuit.
Using $t=q+1$ for the real circuit and $t=q$ for the mixed circuit, the error bounds sum to $8q^2/(2^n+1)\leq8q/2^{|\reg H|}$, since $q<2^{|\reg H|}$ and $n\geq2|\reg H|$.
At $q=0$ the marginals agree; for $q\geq2^{|\reg H|}$, use the trace-norm bound $2$.
\end{proof}

\section{Applications}
\label{sec:applications}

\subsection{A matching decoder}
\label{sec:decodability-upper}

Utsumi et al.~\cite[Theorems~33 and~34]{UNWT25} already apply their Uhlmann algorithm to decoding and Haar-based state merging.
We specialize their preparation-query algorithm to our source and fixed target, with parameters chosen uniformly over $U$.
A decoder aims to transform $\reg R$ so that one of its qubits is maximally entangled with $\reg B$, while leaving $\reg{BH}$ untouched.
We choose a fixed target purification with this property and approximate it by acting only on $\reg R$.
Decoupling bounds the achievable overlap, and the Uhlmann algorithm implements the transformation through preparation queries.
The source and target marginals on the inaccessible registers are
\[
  \rho_{\reg{BH}}^U=\Tr_{\reg{R}}\proj{\Psi_0^{\mathrm{real},U}},\qquad \tau_{\reg{BH}}=\frac{\id_{\reg{BH}}}{2^{1+|\reg H|}}.
\]
Assume $|\reg R|\geq|\reg H|+1$ and split $\reg R$ into a qubit $\reg A$, a copy $\reg{H'}$ of $\reg H$, and a register $\reg{R'}$ of $|\reg R|-|\reg H|-1$ qubits.
The fixed target is
\begin{align*}
  \ket\Theta_{\reg{BHR}}&=\ket{\Phi^+}_{\reg{BA}}\ket{\Phi_{\reg{H}}}_{\reg{HH'}}\ket0_{\reg{R'}},\\
  \ket{\Phi_{\reg{H}}}_{\reg{HH'}}&=\frac1{\sqrt{2^{|\reg H|}}}\sum_h\ket h_{\reg{H}}\ket h_{\reg{H'}}.
\end{align*}
Its $\reg{BH}$ marginal is $\tau_{\reg{BH}}$, and its qubit $\reg A$ is maximally entangled with $\reg B$.
Uhlmann's theorem~\cite[Theorem~3.22]{Wat18} identifies the optimal overlap:
\[
  \max_{W_{\reg R}\text{ unitary}}\bigl|\bra\Theta(\id_{\reg{BH}}\otimes W_{\reg R})\ket{\Psi_0^{\mathrm{real},U}}\bigr|^2=\mathrm F(\rho_{\reg{BH}}^U,\tau_{\reg{BH}})^2.
\]
Both purifications fit in $\reg{BH}:\reg R$ since $2^{|\reg R|}\geq2^{1+|\reg H|}$.
The remaining task is to approximate this optimum with oracle queries.
The Haar-decoupling estimate~\cite[Theorem~3.3]{DBWR14}, applied with reference $\reg B$ and channel $\Tr_{\reg{R}}$, gives
\begin{equation}
  \mathbb E_U\|\rho_{\reg{BH}}^U-\tau_{\reg{BH}}\|_1\leq2\cdot2^{(|\reg H|-|\reg R|)/2}. \label{eq:decoupling-trace-bound}
\end{equation}
Here $\tau_{\reg{BH}}$ is the product of the maximally mixed states on $\reg B$ and $\reg H$.
\Cref{app:decoupling} gives a direct proof from the moments of the two Haar columns.

Since $\tau_{\reg{BH}}$ is maximally mixed, the fidelity bound follows from a scalar inequality.
Let $a_j$ be the eigenvalues of $\rho_{\reg{BH}}^U$, including zeros.
For nonnegative $a,b$, $(\sqrt a-\sqrt b)^2\leq|a-b|$, since $|\sqrt a-\sqrt b|\leq\sqrt a+\sqrt b$.
Thus
\begin{align*}
  1-\mathrm F(\rho_{\reg{BH}}^U,\tau_{\reg{BH}})&=\frac12\sum_j\left(\sqrt{a_j}-2^{-(1+|\reg H|)/2}\right)^2\\
  &\leq\frac12\sum_j\left|a_j-2^{-(1+|\reg H|)}\right|=\frac12\|\rho_{\reg{BH}}^U-\tau_{\reg{BH}}\|_1.
\end{align*}
Using $1-\mathrm F^2=(1-\mathrm F)(1+\mathrm F)\leq2(1-\mathrm F)$ gives
\begin{equation}
  \mathrm F(\rho_{\reg{BH}}^U,\tau_{\reg{BH}})^2\geq1-\|\rho_{\reg{BH}}^U-\tau_{\reg{BH}}\|_1. \label{eq:fidelity-decoupling}
\end{equation}
For additive squared-fidelity error $\epsilon$, the transformation criterion is
\begin{equation}
  \bra\Theta(\id_{\reg{BH}}\otimes\mathcal T^U)(\proj{\Psi_0^{\mathrm{real},U}})\ket\Theta\geq\mathrm F(\rho_{\reg{BH}}^U,\tau_{\reg{BH}})^2-\epsilon. \label{eq:uhlmann-fidelity-error}
\end{equation}

\begin{lemma}
\label{lem:purification-decoder}
A channel satisfying \cref{eq:uhlmann-fidelity-error} yields a decoder $\mathcal D^U=\Tr_{\reg{H'R'}}\circ\mathcal T^U$ using the same queries, with fidelity at least $1-\|\rho_{\reg{BH}}^U-\tau_{\reg{BH}}\|_1-\epsilon$.
\end{lemma}
\begin{proof}
Since $\proj{\Phi^+}_{\reg{BA}}\otimes\id_{\reg{HH'R'}}\geq\proj\Theta$, the Bell-test probability after discarding $\reg{H'R'}$ is at least the overlap with $\ket\Theta$.
The two fidelity inequalities \cref{eq:uhlmann-fidelity-error,eq:fidelity-decoupling} give the result, and discarding registers uses no queries.
\end{proof}

To implement the transformation, we use source and target preparation circuits.
Preparation queries act on auxiliary registers of the same dimensions as $\reg{BH}:\reg R$ and leave the actual $\reg{BH}$ untouched.
In these preparation circuits, $\reg B,\reg Q$ label auxiliary copies of the register spaces.
Let $C$ apply a Hadamard to this auxiliary $\reg B$, followed by a CNOT from $\reg B$ to the first qubit of $\reg Q$.
Then
\[
  C\ket0_{\reg{BQ}}=\frac1{\sqrt2}\sum_c\ket c_{\reg{B}}\ket{x_{0,c}}_{\reg Q},\qquad P_U=(\id_{\reg{B}}\otimes U)C.
\]
The source preparation $P_U$ costs one query, while preparing $\ket\Theta$ uses fixed gates.

\begin{corollary}[Decoder upper bound]
\label{cor:decoding-upper-bound}
For $|\reg R|\geq|\reg H|+1$ and $0<\epsilon<1/2$, there is a decoder using $O((2^{|\reg H|}/\epsilon)\log(1/\epsilon))$ forward and inverse queries with Haar-averaged fidelity at least
\[
  1-2\cdot2^{(|\reg H|-|\reg R|)/2}-\epsilon.
\]
Its operations between queries are independent of $U$.
At fixed $\epsilon$, its query cost is $O(2^{|\reg H|})$.
\end{corollary}
\begin{proof}[Proof sketch]
The overlap operator $\Tr_{\reg{BH}}\ketbra{\Theta}{\Psi_0^{\mathrm{real},U}}$ on $\reg R$ has at most $2^{1+|\reg H|}$ nonzero singular values.
For this fixed target, the optimal root overlap is the sum of these singular values.
The algorithm of Utsumi et al.~\cite[Theorem~13 and Proposition~10]{UNWT25} uses preparation queries and a polynomial approximation to align the operator's right and left singular vectors.
The uniform approximation error above its cutoff is chosen proportional to $\epsilon$.
The singular values below the cutoff have total weight at most $2^{1+|\reg H|}$ times the cutoff.
A cutoff proportional to $\epsilon/2^{1+|\reg H|}$ therefore bounds the total loss, while keeping the query cost $O(2^{|\reg H|})$ at fixed $\epsilon$.
The same cutoff works for every $U$, so the polynomial and its gates depend only on the dimensions and $\epsilon$.
This gives \cref{eq:uhlmann-fidelity-error}.
The decoder implication in \cref{lem:purification-decoder}, averaged over $U$ with \cref{eq:decoupling-trace-bound}, gives the stated fidelity.
\Cref{app:decoding-upper} gives the full proof and the uniform parameter choice.
\end{proof}
If $|\reg H|\leq n/16$ and $|\reg H|\to\infty$, the decoupling error tends to zero and any fixed fidelity below one is achievable with $O(2^{|\reg H|})$ queries.
For every fixed target fidelity in $(1/4,1)$, \cref{cor:query-lower-bounds} gives the matching lower bound even with all four query types.

\subsection{A rank lower bound for Uhlmann transformation}
\label{sec:uhlmann-rank}

The transformation above gives a decoder without increasing the source-query cost.
We use this reduction to prove a lower bound for Uhlmann transformation with a fixed target.
With $|\reg R|=15|\reg H|$, we have $n=16|\reg H|$.
The two query lower-bound scales then both equal $2^{|\reg H|}$, while the decoupling error is at most $2^{1-7|\reg H|}$.
The rank of the product on this family is almost surely $2^{1+|\reg H|}$, as shown next.

\begin{lemma}
\label{lem:uhlmann-rank}
If $|\reg R|\geq|\reg H|+1$, then for Haar-almost every $U$,
\[
  \operatorname{rank}\bigl(\sqrt{\tau_{\reg{BH}}}\sqrt{\rho_{\reg{BH}}^U}\bigr)=2^{1+|\reg H|}.
\]
\end{lemma}
\begin{proof}[Proof sketch]
Since $\tau_{\reg{BH}}$ is invertible and square roots preserve rank, the product has the rank of $\rho_{\reg{BH}}^U$.
The two relevant Haar columns can be generated by orthonormalizing two independent standard complex Gaussian vectors.
Reshaping these vectors across $\reg{BH}:\reg R$ gives a $2^{1+|\reg H|}\times2^{|\reg R|}$ matrix with independent Gaussian entries and full row rank almost surely.
The hypothesis makes the number of columns at least the number of rows, and each new Gaussian row lies in the span of the preceding rows with probability zero.
Orthonormalization acts by an invertible matrix on its rows and preserves rank.
\Cref{app:uhlmann-rank-proof} gives the coefficient matrices and the full calculation.
\end{proof}

\begin{theorem}[Linear rank lower bound]
\label{thm:uhlmann-rank-lower-bound}
Let $|\reg H|\geq1$, $|\reg R|=15|\reg H|$, and $0\leq\epsilon\leq1/8$.
Suppose one algorithm, with circuit and initial auxiliary state independent of $U$, acts only on $\reg R$ and satisfies
\[
  \mathbb E_U\bra\Theta(\operatorname{id}_{\reg{BH}}\otimes\mathcal T^U)(\proj{\Psi_0^{\mathrm{real},U}})\ket\Theta\geq\mathbb E_U\mathrm F(\rho_{\reg{BH}}^U,\tau_{\reg{BH}})^2-\epsilon.
\]
In particular, this condition holds if \cref{eq:uhlmann-fidelity-error} holds for every $U$.
If it uses at most $q$ source-preparation queries, then $q=\Omega(r)$, where $r=\operatorname{rank}(\sqrt{\tau_{\reg{BH}}}\sqrt{\rho_{\reg{BH}}^U})=2^{1+|\reg H|}$ for Haar-almost every $U$.
This holds even when a query coherently selects among $\id,P_U,P_U^\dagger,P_U^*,P_U^{\mathsf T}$.
\end{theorem}
\begin{proof}[Proof sketch]
Since $C$ is fixed, each of the four source-preparation queries costs one oracle query of the corresponding type.
The same conversion holds for a coherent choice of query type, using fixed controlled gates before and after the query.
The Bell projector tensored with the identity dominates $\proj\Theta$, so discarding $\reg{H'R'}$ gives a $q$-query decoder.
The averaged transformation condition and \cref{eq:fidelity-decoupling,eq:decoupling-trace-bound} bound its Haar-averaged fidelity by
\[
  1-2^{1-7|\reg H|}-\epsilon\geq\frac34.
\]
Here $n=16|\reg H|$, so \cref{cor:query-lower-bounds} requires $q=\Omega(2^{|\reg H|})$.
By \cref{lem:uhlmann-rank}, this is $\Omega(r)$.
\Cref{app:uhlmann-lower-proof} gives the complete query conversion and reduction.
\end{proof}

For every fixed $0<\epsilon\leq1/8$, this matches the $O(r)$ preparation-query upper bound~\cite[Theorem~13]{UNWT25}.
Both purifications have efficient oracle preparations: the source uses one forward query, and the target uses fixed gates.
Decoupling makes their inaccessible marginals nearly equal for most Haar unitaries, so a nearly perfect local transformation exists.
Our bound quantifies the query cost of finding such a transformation on this family.
This connects the example to the hard Uhlmann instances considered by Bostanci et al.~\cite[Theorem~8.10]{BEMPQY23}; the resource measured here is access to the public oracle.

\subsection{EFI and commitments from a public Haar oracle}
\label{sec:efi}

The low rank of the real challenge makes it statistically far from the maximally mixed state in the family below.
The query bound shows that this difference remains hard to detect even with access to the same public oracle.
For $\lambda\geq1$, set $|\reg H|=\lambda$, $|\reg R|=2\lambda$, and $n=3\lambda$, and sample $U_\lambda$ from Haar measure.
The pair is
\[
  \rho_{\reg{BR}}^{U_\lambda},\qquad \sigma_{\reg{BR}}^\lambda=\frac{\id_{\reg{BR}}}{2^{1+2\lambda}}.
\]
The real generator makes one forward query and discards $\reg H$; the mixed generator makes no query.
Both use polynomially many fixed gates, and the distinguisher has the same public oracle.
For a binary-output algorithm $\mathcal A$, write
\[
  \Delta_{\mathcal A}(U)=\Pr[\mathcal A^U(\rho_{\reg{BR}}^U)=1]-\Pr[\mathcal A^U(\sigma_{\reg{BR}}^\lambda)=1].
\]

\begin{corollary}[Oracle-relative EFI]
\label{cor:efi}
For every $U_\lambda$, the trace distance of the pair is at least $1-2^{-\lambda-1}$.
For every algorithm whose circuit family and initial auxiliary state are independent of $U_\lambda$, with polynomial selected-query count $q$,
\[
  \mathbb E_{U_\lambda}|\Delta_{\mathcal A}(U_\lambda)|=O((q+1)^2 2^{-\lambda/16}),
\]
which is negligible in $\lambda$.
\end{corollary}
\begin{proof}[Proof sketch]
The real state has rank at most $2^\lambda$ in dimension $2^{1+2\lambda}$.
Its support projector witnesses the stated trace distance.

We prove the stronger bound on the expected absolute gap, whereas the query bound controls the signed average gap.
Gaps with opposite signs can cancel across oracles.
We augment $\mathcal A$ with a calibration step that leaves the original challenge untouched.
The calibrated algorithm estimates the sign of $\Delta_{\mathcal A}(U)$ using $N$ real trials and $N$ mixed trials on separate registers.
These trials use fresh challenges and separately supplied copies of the auxiliary state, so they are independent conditioned on $U$.
It then runs $\mathcal A$ on the original challenge and reverses the output when the estimated sign is negative.
Conditioned on $U$, the resulting acceptance gap averaged over the trials is at least $|\Delta_{\mathcal A}(U)|-\sqrt{2/N}$.
With $N=\lceil2^{\lambda/8}\rceil$, the procedure uses $O(N(q+1))$ queries.

The calibrated algorithm need not be efficient, since the state bound allows arbitrary computation.
\Cref{thm:all-haar,cor:selected-queries} apply to this query count and give the stated bound on the expected absolute gap.
\Cref{app:efi-proof} gives the full proof, including independence, the sign estimate, and the total query count.
\end{proof}

For an independently sampled oracle sequence, \cref{cor:efi-fixed-oracle} gives almost-sure security against every uniform polynomial-query algorithm with a finite description that accesses the current oracle.
The public oracle is part of the construction: averaging it out gives $\mathbb E_{U_\lambda}\rho_{\reg{BR}}^{U_\lambda}=\sigma_{\reg{BR}}^\lambda$.
The same rank bound shows that the real state's entropy is at most $\lambda$ bits, whereas the mixed state's entropy is $1+2\lambda$ bits.
Thus polynomial-query algorithms cannot detect an entropy gap of at least $\lambda+1$ bits.
A hybrid argument extends the indistinguishability bound to polynomially many copies; \cref{app:efi-copies} gives the query count.

\paragraph{Quantum bit commitment.}
The purified-generator construction of Brakerski, Canetti, and Qian~\cite[Lemma~4.7]{BCQ23} applies directly to this pair.
To commit to zero, the sender prepares $\ket{\Psi_0^{\mathrm{real},U}}_{\reg{BRH}}$, sends $\reg{BR}$, and keeps $\reg H$.
To commit to one, it prepares a maximally entangled state between $\reg{BR}$ and an opening register of $1+2\lambda$ qubits, and again sends $\reg{BR}$.
The opening register for zero is padded with $\lambda+1$ qubits initialized to $\ket0$ so that both openings have the same size.
Opening reveals the bit and this register; the receiver uncomputes the stated preparation and accepts exactly when every qubit is zero.

\begin{corollary}[Quantum bit commitment]
\label{cor:bit-commitment}
This protocol has perfect correctness and expected absolute hiding gap $O((q+1)^2 2^{-\lambda/16})$ against any receiver whose circuit and initial auxiliary state are independent of $U$ and which makes a polynomial number $q$ of selected queries.
$q$ counts all receiver queries before and after receipt of the commitment, through the end of the hiding experiment and before opening.
For every fixed $U$, every common post-commit state and every pair of opening procedures satisfy $p_0+p_1\leq1+2^{-(\lambda+1)/2}$, where $p_b$ is the probability of successfully opening bit $b$.
Honest commitment uses at most one forward query, and verification uses at most one inverse query.
\end{corollary}
\begin{proof}[Proof sketch]
The receiver's commitment marginals are exactly the EFI pair, so \cref{cor:efi} proves hiding.
Let $\Pi_U$ be the support projector of $\rho_{\reg{BR}}^U$.
Its weight in the mixed marginal is
\[
  \Tr[\Pi_U\sigma_{\reg{BR}}^\lambda]=\frac{\operatorname{rank}(\rho_{\reg{BR}}^U)}{2^{1+2\lambda}}\leq2^{-\lambda-1}.
\]
An accepted zero-opening vector is supported on $\Pi_U$ in $\reg{BR}$, whereas an accepted one-opening vector has the mixed marginal.
Cauchy--Schwarz therefore bounds their inner product by $2^{-(\lambda+1)/2}$.
The norm of the sum of the two acceptance projections is at most one plus this overlap, giving the binding bound.
\Cref{app:bit-commitment} proves this for arbitrary malicious commitments and opening procedures.
\end{proof}

\paragraph{A commitment to a qubit.}
The same Haar encoding gives a direct commitment to a quantum message.
Let $\reg M$ be one qubit and initialize the $3\lambda-1$ qubits of $\reg W$ in $\ket0^{\otimes(3\lambda-1)}$.
The sender applies $U$ to $\reg{MW}$, sends $\reg R$, and keeps $\reg H$.
Opening reveals $\reg H$; the receiver applies $U^\dagger$, checks that $\reg W$ is zero, and recovers $\reg M$.
Correctness holds even when $\reg M$ is entangled with an external reference.

We use the hiding definition~\cite[Definition~4.3]{GJMZ23} and the swap-binding definition~\cite[Definition~4.2]{GJMZ23} of Gunn et al.
Hiding compares encoding a receiver-chosen message with encoding zero, while preserving the receiver's workspace.
Swap binding asks whether a malicious sender can distinguish keeping the accepted message from replacing it by zero, after the receiver re-encodes and returns only $\reg H$.

\begin{corollary}[Quantum state commitment]
\label{cor:qubit-commitment}
The one-qubit protocol has perfect correctness and Haar-averaged statistical swap-binding advantage at most $2\cdot2^{-\lambda/2}$.
For polynomial $q$, its expected absolute hiding gap is $O((q+1)^2 2^{-\lambda/16})$.
Hiding holds for chosen messages entangled with a reference and for receivers with oracle-independent circuits and initial auxiliary states; $q$ counts all selected queries before and after commitment, , through the end of the hiding experiment and before opening.
The binding bound holds uniformly over unbounded senders, including strategies chosen with knowledge of $U$.
Honest commitment uses one forward query, and opening verification uses one inverse query.
\end{corollary}
\begin{proof}[Proof sketch]
Replacing $\reg B$ by a maximally mixed qubit fixes the mixed challenge.
Applying \cref{cor:efi} before and after this replacement therefore also bounds the gap between $\rho_{\reg{BR}}^U$ and $\id_{\reg{B}}/2\otimes\rho_{\reg{R}}^U$.
To encode a chosen message using this EPR challenge, measure the message and $\reg B$ in the Bell basis.
The $\Phi^+$ outcome has probability $1/4$ for every $U$ and every message, and produces exactly its encoding on $\reg R$, including its correlations with the receiver's workspace.
For the decorrelated challenge it instead produces $\rho_{\reg{R}}^U$ independently of the message.
Comparing both the chosen message and zero with this same state proves hiding.

For binding, \cref{eq:decoupling-trace-bound} bounds the Choi state of the channel from $\reg M$ to $\reg H$ by its maximally mixed replacement.
Since $\reg M$ is one qubit, this trace-norm bound controls the channel uniformly over all inputs with a reference, losing only a factor of two.
It therefore bounds both states returned to a malicious sender in the swap-binding game, even after conditioning on a successful validity check.
\Cref{app:qubit-commitment} gives the channel estimate, the Bell-measurement identities, and the full security bounds.
\end{proof}

Both protocols use quantum communication and the same public oracle as the decoding task.
The one-qubit protocol sends $2\lambda$ qubits during commitment and another $\lambda$ qubits at opening, for $3\lambda$ qubits in total; its purpose here is to give a direct cryptographic use of the encoding.

\section*{Acknowledgments}
This work is supported by AM's faculty startup grant from Virginia Tech and, in part, by the Commonwealth Cyber Initiative, an investment in the advancement of cyber research, innovation, and workforce development.
The authors thank Charles Cao and Thomas Koutsikos for helpful discussions.

\section*{AI Disclosure}
The authors came up with the initial idea and approach for the forward-access result. Claude Opus 5 then helped turn the authors' proof sketch for this case into a complete proof, and also helped with the calculations for the inverse-access case. For the remaining two access models and for some of the applications, the authors wrote proof sketches and used ChatGPT Sol 6.1 to help draft more rigorous proofs from them. The AI-generated drafts of these proofs were often long and complicated. The authors replaced them with simpler proofs based on elementary identities, and these are the proofs in the paper. Claude Opus 5.5 and ChatGPT Sol 6.1 were also used throughout for revising and editing the manuscript, including finding and fixing errors in the initial proofs. The authors checked the correctness of all proofs and take full responsibility for the content of this paper.

\Needspace{5\baselineskip}
\phantomsection
\addcontentsline{toc}{section}{References}
\bibliographystyle{alpha}
\begingroup
\small
\raggedright
\bibliography{refs_bhrd}
\endgroup

\clearpage
\appendix
\addtocontents{toc}{\protect\setcounter{tocdepth}{1}}
\section*{Deferred proofs}

These appendices give the full case calculations, query-bound constants, and application proofs.

\section{Recording-oracle calculations}
\label{app:deferred-proofs}
\label{app:recording-calculations}

The calculations below establish the relation normalization, swap identity, and commutator bounds used in \cref{sec:trace-distance-proof-sketch}.
The commutators are expanded separately for the four creation maps and their adjoints.

\subsection{Elementary norm and counting bounds}
\label{app:elementary}

\paragraph{Relation states.}
Let $\mathcal K=\operatorname{span}\{\ket e:e\in[2^n]^2\}$.
At relation size $\ell$, each relation register has state space $\operatorname{Sym}^{\ell}(\mathcal K)$, the subspace of $\mathcal K^{\otimes\ell}$ invariant under permutations.
The full relation registers are the direct sums of these spaces over $\ell\geq0$, and $\reg{E_0}$ has state space $\mathcal K$.
For $L=\{(x_1,y_1),\ldots,(x_\ell,y_\ell)\}$, use the normalized relation state of Ma and Huang~\cite[Section~2.1]{MH24}:
\begin{equation}
  \ket L=\frac{\sum_{\pi\in\Sym_\ell}\ket{x_{\pi(1)},y_{\pi(1)},\ldots,x_{\pi(\ell)},y_{\pi(\ell)}}}{\sqrt{\ell!\prod_{e\in[2^n]^2}\operatorname{num}(L,e)!}}. \label{eq:relation-state}
\end{equation}
Here $\Sym_\ell$ is the set of permutations of the $\ell$ positions.
For $\ell=0$, write $\ket\emptyset$.

In \cref{eq:relation-state}, each distinct ordering occurs $\prod_e\operatorname{num}(L,e)!$ times.
There are $\ell!/\prod_e\operatorname{num}(L,e)!$ such orderings.
The numerator therefore has squared norm $\ell!\prod_e\operatorname{num}(L,e)!$, which verifies the normalization.
Distinct multisets have disjoint supports, and a symmetric vector has equal coefficients on all orderings of a given multiset.
Thus these states form an orthonormal basis of the symmetric subspace.

\paragraph{Oracle norms.}
The maps $A,B$ are isometries on their defining domains, as in Ma--Huang~\cite[Claim~14]{MH24}.
For $A$, the final query $y$ identifies the unique added pair $(x,y)$; removing it recovers $x,L,R$.
Thus distinct input basis vectors have orthogonal images, and
\[
  \|A\ket x\ket L\ket R\|_2^2 =\frac{2^n-|\Img(L\cup R)|}{2^n-|\Img(L\cup R)|}=1.
\]
For $B$, the final query $x$ identifies the unique added pair in $R$, and
\[
  \|B\ket y\ket L\ket R\|_2^2 =\frac{2^n-|\Dom(L\cup R)|}{2^n-|\Dom(L\cup R)|}=1.
\]
Extend $A,B$ by zero outside their defining domains.
Then $AA^\dagger,BB^\dagger$ are orthogonal projectors and
\[
  \|A\|=\|A^\dagger\|=\|B\|=\|B^\dagger\|=1.
\]
Ma--Huang~\cite[Claim~15]{MH24} prove that $V$ is a partial isometry.
For the norm bound used here, the first term in $V$ is in the image of $A$ and the second is orthogonal to it.
Hence
\begin{align*}
  \|Vv\|_2^2 &=\|A(\id-BB^\dagger)v\|_2^2 +\|(\id-AA^\dagger)B^\dagger v\|_2^2\\
  &\leq\|(\id-BB^\dagger)v\|_2^2+\|B^\dagger v\|_2^2\\
  &=\|(\id-BB^\dagger)v\|_2^2+\|BB^\dagger v\|_2^2 =\|v\|_2^2.
\end{align*}
Here $\|B^\dagger v\|_2^2=\bra vBB^\dagger\ket v=\|BB^\dagger v\|_2^2$, and the projectors $BB^\dagger$ and $\id-BB^\dagger$ have orthogonal images.
The adjoint obeys the same norm bound:
\begin{align*}
  \|V^\dagger v\|_2 &=\sup_{\|w\|_2=1}|\braket{w}{V^\dagger v}| =\sup_{\|w\|_2=1}|\braket{Vw}{v}|\\
  &\leq\sup_{\|w\|_2=1}\|Vw\|_2\|v\|_2 \leq\|v\|_2.
\end{align*}
Swapping $\reg{E_L},\reg{E_R}$ maps $A,B,V$ to $\bar A,\bar B,\bar V$, respectively.
The swap is unitary, so
\[
  \|\bar A\|=\|\bar A^\dagger\|=\|\bar B\|=\|\bar B^\dagger\|=1,\qquad \|\bar V\|,\|\bar V^\dagger\|\leq1.
\]

\paragraph{Insertion and its norm.}
For $|L|=\ell$, let $P_{\ell+1}$ be the orthogonal projector onto $\operatorname{Sym}^{\ell+1}(\mathcal K)$ and define
\[
  J=\sqrt{\ell+1}\,P_{\ell+1}\big|_{\mathcal K\otimes\operatorname{Sym}^{\ell}(\mathcal K)}.
\]
The first factor is $\reg{E_0}$, and the remaining $\ell$ factors form $\reg{E_L}$.
The map acts trivially on $\reg{E_R}$ and the other registers.
The fraction of orderings of $L\cup\{e\}$ with first pair $e$ is $(\operatorname{num}(L,e)+1)/(\ell+1)$.
By the normalization above, its square root is the overlap of $\ket e\ket L$ with $\ket{L\cup\{e\}}$.
The projection $P_{\ell+1}$ is this overlap times $\ket{L\cup\{e\}}$; multiplying by $\sqrt{\ell+1}$ gives
\[
  J\ket e_\reg{E_0}\ket L_\reg{E_L}=\sqrt{\operatorname{num}(L,e)+1}\ket{L\cup\{e\}}_\reg{E_L}.
\]
The factors are numbered $0,1,\ldots,\ell$, with $0$ corresponding to $\reg{E_0}$.
We write
\[
  \Sigma=\sum_{i=1}^{\ell}\operatorname{SWAP}_{0i}.
\]
Permuting the last $\ell$ factors permutes the summands of $\Sigma$, so $\Sigma$ preserves $\mathcal K\otimes\operatorname{Sym}^{\ell}(\mathcal K)$.
For each $i$, a permutation sending $0$ to $i$ acts as $\operatorname{SWAP}_{0i}$ on this subspace: removing that swap leaves a permutation of the last $\ell$ factors.
There are $\ell!$ such permutations for each $i$, so averaging all $(\ell+1)!$ permutations gives $P_{\ell+1}=(\id+\Sigma)/(\ell+1)$ on this subspace.
Since $P_{\ell+1}^2=P_{\ell+1}$, we obtain
\[
  J^\dagger J=(\ell+1)P_{\ell+1}\big|_{\mathcal K\otimes\operatorname{Sym}^{\ell}(\mathcal K)}=\id+\Sigma.
\]
These definitions are applied separately on each subspace of fixed relation size.

For $|L|=\ell$, the orthogonal projection gives
\begin{align*}
  \|Jv\|_2^2 &=(\ell+1)\|P_{\ell+1}v\|_2^2 \leq(\ell+1)\|v\|_2^2,\\
  \|v\|_2^2-\|P_{\ell+1}v\|_2^2 &=\|(\id-P_{\ell+1})v\|_2^2\geq0.
\end{align*}
If $v=\sum_{\ell=0}^s v_\ell$, where $v_\ell$ has $|L|=\ell$, then
\[
  \|Jv\|_2^2=\sum_{\ell=0}^s\|Jv_\ell\|_2^2 \leq\sum_{\ell=0}^s(\ell+1)\|v_\ell\|_2^2 \leq(s+1)\|v\|_2^2.
\]
Different $\ell$ give orthogonal outputs, since $Jv_\ell$ has $|L|=\ell+1$.

\paragraph{Amplitude and interference bounds.}
We use the following square-root estimate for each denominator change.
For $m\geq2$,
\begin{align}
  f(m):=\sqrt{\frac m{m-1}}-1 &=\frac{m/(m-1)-1}{\sqrt{m/(m-1)}+1}\nonumber\\
  &=\frac1{(m-1)(\sqrt{m/(m-1)}+1)} \leq\frac1{2(m-1)}. \label{eq:ratio}
\end{align}
The last inequality uses $\sqrt{m/(m-1)}\geq1$.

When several terms have the same output label, their workspace vectors can interfere.
For vectors $v_1,\dots,v_m$,
\begin{align}
  \left\|\sum_{i=1}^m v_i\right\|_2^2 &=\sum_{i,j=1}^m\operatorname{Re}\braket{v_i}{v_j}\nonumber\\
  &\leq\sum_{i,j=1}^m\|v_i\|_2\|v_j\|_2\nonumber\\
  &\leq\frac12\sum_{i,j=1}^m(\|v_i\|_2^2+\|v_j\|_2^2) =m\sum_{i=1}^m\|v_i\|_2^2. \label{eq:matching-sums}
\end{align}
The two inequalities use Cauchy--Schwarz and $2ab\leq a^2+b^2$, respectively.
If $\ket\lambda$ are orthonormal labels and $\lambda_\mu$ is the label of the term indexed by $\mu$, then
\begin{align}
  \left\|\sum_\mu\ket{\lambda_\mu}\ket{v_\mu}\right\|_2^2 &=\sum_{\substack{\mu,\widetilde\mu\\ \lambda_\mu=\lambda_{\widetilde\mu}}} \operatorname{Re}\braket{v_\mu}{v_{\widetilde\mu}}\nonumber\\
  &\leq\frac12\sum_{\substack{\mu,\widetilde\mu\\ \lambda_\mu=\lambda_{\widetilde\mu}}} \left(\|v_\mu\|_2^2+\|v_{\widetilde\mu}\|_2^2\right)\nonumber\\
  &=\sum_\mu\|v_\mu\|_2^2 \sum_{\substack{\widetilde\mu\\ \lambda_{\widetilde\mu}=\lambda_\mu}}1. \label{eq:matching-labels}
\end{align}
The last inner sum counts the indices with the same final label.
The two squared-norm sums agree after interchanging $\mu,\widetilde\mu$.

\paragraph{Rank-one trace norm.}
For any vectors $v,w$,
\[
  (\ketbra v w)^\dagger\ketbra v w=\|v\|_2^2\proj w, \qquad \|\ketbra v w\|_1=\|v\|_2\|w\|_2.
\]
The only possible nonzero singular value is $\|v\|_2\|w\|_2$.

\subsection{\texorpdfstring{Proof of \cref{lem:insertion-error}}{Proof of the insertion marginal bound}}
\begin{proof}
\label{app:insertion-error-bound-proof}
We express matrix elements of the swap sum $\Sigma$ as deletion followed by insertion.

\paragraph{Matrix entries of the swap sum.}
For a pair $e$, define $D_e\ket L=\sqrt{\operatorname{num}(L,e)}\ket{L\setminus\{e\}}$, with value zero if $e\notin L$.
Then $D_e^\dagger D_e\ket L=\operatorname{num}(L,e)\ket L$.
The matrix entry $\bra{e'}_{\reg{E_0}}\Sigma\ket e_{\reg{E_0}}$ replaces $e'$ by $e$ in $L$, summed over all positions.
The normalized insertion and deletion weights give
\begin{equation}
  \bra{e'}_{\reg{E_0}}\Sigma\ket e_{\reg{E_0}}=D_e^\dagger D_{e'}. \label{eq:swap-deletion}
\end{equation}
\subsubsection{Verification of \cref{eq:swap-deletion}}
For $|L|=\ell\geq1$, the fraction of orderings whose first pair is $e$ is $\operatorname{num}(L,e)/\ell$.
Grouping by the first pair gives
\[
  \ket L=\sum_{e\in L}\sqrt{\frac{\operatorname{num}(L,e)}\ell} \ket e\ket{L\setminus\{e\}},
\]
where the sum is over distinct pairs.
Also,
\[
  \bra{e'}_\reg{E_0}\Sigma\ket e_\reg{E_0} =\sum_{i=1}^\ell(\ketbra e{e'})_i.
\]
For each swap, taking the matrix entry from $e$ to $e'$ in $\reg{E_0}$ replaces $e'$ by $e$ in the swapped position of $\reg{E_L}$.
The $\ell$ terms have equal matrix entries because the relation states are symmetric.
For $e\neq e'$, $e'\in L$, and $L'=(L\setminus\{e'\})\cup\{e\}$,
\begin{align*}
  \bra{L'}\sum_{i=1}^\ell(\ketbra e{e'})_i\ket L &=\ell\bra{L'}(\ketbra e{e'})_1\ket L\\
  &=\ell\sqrt{\frac{\operatorname{num}(L,e')}\ell} \sqrt{\frac{\operatorname{num}(L',e)}\ell}\\
  &=\sqrt{\operatorname{num}(L,e')\bigl(\operatorname{num}(L,e)+1\bigr)}.
\end{align*}
The second line uses the decomposition by the first pair, and the last uses $\operatorname{num}(L',e)=\operatorname{num}(L,e)+1$.
For $e=e'$,
\[
  \bra L\sum_{i=1}^\ell(\proj e)_i\ket L =\ell\,\frac{\operatorname{num}(L,e)}\ell=\operatorname{num}(L,e).
\]
All other matrix entries are zero.
Applying $D_{e'}$ and then $D_e^\dagger$ gives these same coefficients by their definitions.
For $\ell=0$, both $\Sigma$ and $D_{e'}\ket\emptyset$ vanish.
Thus $\bra{e'}_{\reg{E_0}}\Sigma\ket e_{\reg{E_0}}=D_e^\dagger D_{e'}$ for every relation size.

\paragraph{Counting the deletion weight.}
Let $\widetilde e_0=(x_{0,\widetilde c},(\widetilde r,h))$.
Tracing out $\reg H$ forces the same $h$ in the ket and bra.
Using \cref{eq:swap-deletion} and cyclicity of the partial trace over $\reg{E_L}$ gives
\[
  K_\phi=\frac1{2^{n+1}}\sum_{c,r,\widetilde c,\widetilde r,h}\Tr_{\reg{E_LE_R}}\left[\ketbra{D_{\widetilde e_0}\alpha_{c,r}}{D_{e_0}\alpha_{\widetilde c,\widetilde r}}\right].
\]
For any vectors $v,w$, trace-norm duality and $\|\ketbra vw\|_1=\|v\|_2\|w\|_2$ give $\|\Tr_E\ketbra vw\|_1\leq\|v\|_2\|w\|_2$.
The deletion weights have a common bound.
For fixed $c,r$, the diagonal operator $D_e^\dagger D_e$ counts copies of $e$ in $L$.
Summing over all challenge pairs counts a subset of $L$, so
\[
  \sum_{\widetilde c,\widetilde r,h}\|D_{\widetilde e_0}\alpha_{c,r}\|_2^2\leq q\|\alpha_{c,r}\|_2^2.
\]
\paragraph{Marginal and norm bounds.}
We apply the rank-one trace bound and then $2ab\leq a^2+b^2$.
The two squared sums are equal after interchanging $(c,r)$ with $(\widetilde c,\widetilde r)$, and the deletion bound applies to each.
Thus
\begin{align*}
  \|K_\phi\|_1&\leq\frac1{2^{n+1}}\sum_{c,r,\widetilde c,\widetilde r,h}\|D_{\widetilde e_0}\alpha_{c,r}\|_2\|D_{e_0}\alpha_{\widetilde c,\widetilde r}\|_2\\
  &\leq\frac q{2^{n+1}}\sum_{c,r}\|\alpha_{c,r}\|_2^2=\frac q{2^{|\reg H|}}\|\phi\|_2^2.
\end{align*}
The last equality is \cref{eq:coefficient-energy}.
Taking the trace in \cref{eq:marginal-identity} gives $\|J\phi\|_2^2=\|\phi\|_2^2+\Tr K_\phi$.
The second bound follows from $|\Tr K_\phi|\leq\|K_\phi\|_1$.
\end{proof}

\subsection{Commutator bounds}
\label{app:commutators}

This subsection proves the eight component bounds collected in \cref{eq:eight-epsilons}.
\Cref{lem:commutator-a,lem:commutator-b} treat creation in the two records, and \cref{lem:commutator-a-dagger,lem:commutator-b-dagger} treat the corresponding deletions.
The last four lemmas treat the barred maps, whose destination records are exchanged.
In every case, $[X,J]=XJ-JX$, $0\leq q\leq2^n-2$, and $\ket\phi$ has the form in \cref{eq:mixed-experiment-state-family}, with coefficient vectors independent of $h$ and supported on $|L\cup R|\leq q$.
Here $\ket a,\ket b,\ket c$ denote error vectors; $c$ in summation indices is the challenge bit.
Old labels mean the input challenge, radiation, hidden, query, and relation labels; final labels mean the output basis labels outside the workspace.

Each proof expands the two orders and cancels the terms with unchanged normalization and allowed labels.
The remaining vectors describe the normalization change, the excluded creation or blocked deletion, and, for $A^\dagger,\bar B^\dagger$, deletion of $e_0$ itself.
We bound the normalization term using \cref{eq:ratio} and the norm of $J$.
For the other terms, we group contributions with the same query, hidden, and relation labels, count the input labels yielding each such output, and apply \cref{eq:matching-labels} to their workspace vectors.
\Cref{eq:coefficient-energy,eq:diagonal-energy} then account for the sums over the radiation and hidden labels.

\begingroup
\allowdisplaybreaks[1]

\begin{lemma}[Forward creation]
  \label{lem:commutator-a}
Let $0\leq q\leq2^n-2$ and let $\ket\phi$ have the form in \cref{eq:mixed-experiment-state-family}, with coefficient vectors independent of $h$ and supported on $|L\cup R|\leq q$.
Then
\begin{equation}
  \|[A,J]\phi\|_2^2\leq\left(\frac2{2^n-q}+\frac{q+2}{4(2^n-q-1)^2}\right)\|\phi\|_2^2. \label{eq:creation-summary}
\end{equation}
\end{lemma}
\begin{proof}
Write $k=|\Img(L\cup R)|\leq q$.
Insertion changes the image size as follows:
\[
  |\Img((L\cup\{e_0\})\cup R)| =\begin{cases} k+1,&y_0\notin\Img(L\cup R),\\ k,&y_0\in\Img(L\cup R). \end{cases}
\]
\paragraph{The two orders.}
For the branch $y=y_0\notin\Img(L\cup R)$, the insertion weight is
\[
  w_{c,x}:=\begin{cases}\sqrt2,&x=x_{0,c},\\1,&x\neq x_{0,c},\end{cases}
\]
The common output vector is
\[
  \ket\beta=\ket y_\reg Q\ket{\alpha_{c,r,x,L,R}}_\reg W \ket h_\reg H\ket{L\cup\{e_0,(x,y)\}}_\reg{E_L}\ket R_\reg{E_R}.
\]
By the definitions of $A,J$,
\begin{align*}
  AJ\ket\phi ={}&\frac1{\sqrt{2^{n+1}}} \sum_{\substack{c,r,h,x,y,L,R\\y_0\notin\Img(L\cup R)\\ y\notin\Img(L\cup R),\ y\neq y_0}} \frac{\ket\beta}{\sqrt{2^n-k-1}}\\*
  &+\frac1{\sqrt{2^{n+1}}} \sum_{\substack{c,r,h,x,y,L,R\\y_0\in\Img(L\cup R)\\ y\notin\Img(L\cup R)}} \frac{\sqrt{\operatorname{num}(L,e_0)+1}}{\sqrt{2^n-k}}\ket\beta,\\
  JA\ket\phi ={}&\frac1{\sqrt{2^{n+1}}} \sum_{\substack{c,r,h,x,y,L,R\\y_0\notin\Img(L\cup R)\\ y\notin\Img(L\cup R),\ y\neq y_0}} \frac{\ket\beta}{\sqrt{2^n-k}}\\*
  &+\frac1{\sqrt{2^{n+1}}} \sum_{\substack{c,r,h,x,y,L,R\\y_0\in\Img(L\cup R)\\ y\notin\Img(L\cup R)}} \frac{\sqrt{\operatorname{num}(L,e_0)+1}}{\sqrt{2^n-k}}\ket\beta\\*
  &+\frac1{\sqrt{2^{n+1}}} \sum_{\substack{c,r,h,x,y,L,R\\y_0\notin\Img(L\cup R)\\y=y_0}} \frac{w_{c,x}}{\sqrt{2^n-k}}\ket\beta.
\end{align*}
When $y_0\in\Img(L\cup R)$, the two orders have the same denominator and insertion weight, so their second sums cancel.
When $y_0\notin\Img(L\cup R)$, insertion forbids $y=y_0$ in $AJ$, whereas $JA$ can include it.
In this last case, the inserted pairs coincide exactly when $x=x_{0,c}$, giving the weight $\sqrt2$.
Since
\[
  \frac1{\sqrt{2^n-k-1}}-\frac1{\sqrt{2^n-k}} =\frac{f(2^n-k)}{\sqrt{2^n-k}},
\]
we have
\begin{align*}
  [A,J]\ket\phi&=\ket a-\ket c,\\
  \ket a&=\frac1{\sqrt{2^{n+1}}} \sum_{\substack{c,r,h,x,y,L,R\\y_0\notin\Img(L\cup R)\\ y\notin\Img(L\cup R),\ y\neq y_0}} \frac{f(2^n-k)}{\sqrt{2^n-k}}\ket\beta,\\
  \ket c&=\frac1{\sqrt{2^{n+1}}} \sum_{\substack{c,r,h,x,y,L,R\\y_0\notin\Img(L\cup R)\\y=y_0}} \frac{w_{c,x}}{\sqrt{2^n-k}}\ket\beta.
\end{align*}

\paragraph{The normalization term.}
Let $\ket{a'}$ be the sum defining $\ket a$ before applying $J$:
\[
  \ket{a'}=\frac1{\sqrt{2^{n+1}}} \sum_{\substack{c,r,h,x,y,L,R\\y_0\notin\Img(L\cup R)\\ y\notin\Img(L\cup R),\ y\neq y_0}} \frac{f(2^n-k)}{\sqrt{2^n-k}} \ket y\ket{\alpha_{c,r,x,L,R}}\ket h\ket{e_0} \ket{L\cup\{(x,y)\}}\ket R.
\]
Here $\ket a=J\ket{a'}$, and the added pair is the unique pair with output $y$.
Removing it and reading $e_0$ recovers all old labels, so the terms in $\ket{a'}$ have orthogonal labels.
There are $2^n-k-1$ allowed $y$ values.
The relation in $\reg{E_L}$ has size at most $q+1$ in $\ket{a'}$.
Using $f(2^n-k)\leq1/[2(2^n-q-1)]$, \cref{eq:coefficient-energy}, and the bound on $J$, we get
\begin{align*}
  \|a'\|_2^2 &=\frac1{2^{n+1}} \sum_{\substack{c,r,h,x,L,R\\y_0\notin\Img(L\cup R)}} \frac{2^n-k-1}{2^n-k}f(2^n-k)^2 \|\alpha_{c,r,x,L,R}\|_2^2\\
  &\leq\frac1{4(2^n-q-1)^2}\cdot\frac1{2^{n+1}} \sum_{c,r,h,x,L,R}\|\alpha_{c,r,x,L,R}\|_2^2\\
  &=\frac{\|\phi\|_2^2}{4(2^n-q-1)^2},\\
  \|a\|_2^2 &\leq(q+2)\|a'\|_2^2 \leq\frac{q+2}{4(2^n-q-1)^2}\|\phi\|_2^2.
\end{align*}
\paragraph{The excluded output.}
For $\ket c$, write
\[
  \mu=(c,r,h,x,L,R),\quad \lambda_\mu=(y_0,h,L\cup\{e_0,(x,y_0)\},R),\quad w_\mu=w_{c,x}.
\]
The index $\mu$ ranges over the terms in the sum defining $\ket c$, with $\alpha_\mu=\alpha_{c,r,x,L,R}$ and $k_\mu=|\Img(L\cup R)|$.
The two pairs with output $y_0$ recover the old $L$ after their removal.
Choosing which is $e_0$ fixes $c,r,h,x$.
If the two pairs are distinct, there are at most two choices, both with $w_\mu=1$.
If they are equal, there is one choice with $w_\mu=\sqrt2$.
Hence
\[
  w_\mu^2\sum_{\substack{\widetilde\mu\\ \lambda_{\widetilde\mu}=\lambda_\mu}}1\leq2.
\]
By \cref{eq:matching-labels},
\begin{align*}
  \|c\|_2^2 &=\frac1{2^{n+1}}\sum_\lambda \left\|\sum_{\substack{\mu\\\lambda_\mu=\lambda}} \frac{w_\mu\alpha_\mu}{\sqrt{2^n-k_\mu}}\right\|_2^2\\
  &\leq\frac1{2^{n+1}}\sum_\mu \frac{w_\mu^2\|\alpha_\mu\|_2^2}{2^n-k_\mu} \sum_{\substack{\widetilde\mu\\ \lambda_{\widetilde\mu}=\lambda_\mu}}1\\
  &\leq\frac2{2^{n+1}} \sum_{\substack{c,r,h,x,L,R\\y_0\notin\Img(L\cup R)}} \frac{\|\alpha_{c,r,x,L,R}\|_2^2}{2^n-k}\\
  &\leq\frac2{2^n-q}\|\phi\|_2^2.
\end{align*}
In $\ket a$, the final relation $L$ contains one pair with output equal to the final query label; in $\ket c$, it contains two.
These are different relation labels, so
\[
  \braket a c=0,\qquad \|[A,J]\phi\|_2^2=\|a\|_2^2+\|c\|_2^2 \leq\left(\frac{q+2}{4(2^n-q-1)^2}+\frac2{2^n-q}\right) \|\phi\|_2^2.
\]
\end{proof}

\begin{lemma}[Inverse creation]
  \label{lem:commutator-b}
Let $0\leq q\leq2^n-2$ and let $\ket\phi$ have the form in \cref{eq:mixed-experiment-state-family}, with coefficient vectors independent of $h$ and supported on $|L\cup R|\leq q$.
Then
\[
  \|[B,J]\phi\|_2^2 \leq\left(\frac1{2^n-q} +\frac{q+1}{4(2^n-q-1)^2}\right)\|\phi\|_2^2.
\]
\end{lemma}
\begin{proof}
We relabel the query-basis index $x$ in the expansion of $\alpha_{c,r}$ as $y$.
Write $d=|\Dom(L\cup R)|\leq q$.
Insertion changes the domain size as follows:
\[
  |\Dom((L\cup\{e_0\})\cup R)| =\begin{cases} d+1,&x_{0,c}\notin\Dom(L\cup R),\\ d,&x_{0,c}\in\Dom(L\cup R). \end{cases}
\]
\paragraph{The two orders.}
The common output vector is
\[
  \ket\beta=\ket x_\reg Q\ket{\alpha_{c,r,y,L,R}}_\reg W \ket h_\reg H\ket{L\cup\{e_0\}}_\reg{E_L} \ket{R\cup\{(x,y)\}}_\reg{E_R}.
\]
The definitions of $B,J$ give
\begin{align*}
  BJ\ket\phi ={}&\frac1{\sqrt{2^{n+1}}} \sum_{\substack{c,r,h,x,y,L,R\\x_{0,c}\notin\Dom(L\cup R)\\ x\notin\Dom(L\cup R),\ x\neq x_{0,c}}} \frac{\ket\beta}{\sqrt{2^n-d-1}}\\*
  &+\frac1{\sqrt{2^{n+1}}} \sum_{\substack{c,r,h,x,y,L,R\\x_{0,c}\in\Dom(L\cup R)\\ x\notin\Dom(L\cup R)}} \frac{\sqrt{\operatorname{num}(L,e_0)+1}}{\sqrt{2^n-d}}\ket\beta,\\
  JB\ket\phi ={}&\frac1{\sqrt{2^{n+1}}} \sum_{\substack{c,r,h,x,y,L,R\\x_{0,c}\notin\Dom(L\cup R)\\ x\notin\Dom(L\cup R),\ x\neq x_{0,c}}} \frac{\ket\beta}{\sqrt{2^n-d}}\\*
  &+\frac1{\sqrt{2^{n+1}}} \sum_{\substack{c,r,h,x,y,L,R\\x_{0,c}\in\Dom(L\cup R)\\ x\notin\Dom(L\cup R)}} \frac{\sqrt{\operatorname{num}(L,e_0)+1}}{\sqrt{2^n-d}}\ket\beta\\*
  &+\frac1{\sqrt{2^{n+1}}} \sum_{\substack{c,r,h,x,y,L,R\\x_{0,c}\notin\Dom(L\cup R)\\x=x_{0,c}}} \frac{\ket\beta}{\sqrt{2^n-d}}.
\end{align*}
When $x_{0,c}\in\Dom(L\cup R)$, the second sums cancel.
Otherwise insertion forbids $x=x_{0,c}$ in $BJ$, whereas $JB$ can include it.
Since $B$ adds its pair to $R$, it does not change the insertion weight in $L$.
Since
\[
  \frac1{\sqrt{2^n-d-1}}-\frac1{\sqrt{2^n-d}} =\frac{f(2^n-d)}{\sqrt{2^n-d}},
\]
we have
\begin{align*}
  [B,J]\ket\phi&=\ket a-\ket c,\\
  \ket a&=\frac1{\sqrt{2^{n+1}}} \sum_{\substack{c,r,h,x,y,L,R\\x_{0,c}\notin\Dom(L\cup R)\\ x\notin\Dom(L\cup R),\ x\neq x_{0,c}}} \frac{f(2^n-d)}{\sqrt{2^n-d}}\ket\beta,\\
  \ket c&=\frac1{\sqrt{2^{n+1}}} \sum_{\substack{c,r,h,x,y,L,R\\x_{0,c}\notin\Dom(L\cup R)\\x=x_{0,c}}} \frac{\ket\beta}{\sqrt{2^n-d}}.
\end{align*}

\paragraph{The normalization term.}
Before applying $J$, the added pair in $R$ is identified by $x$.
Let
\[
  \ket{a'}=\frac1{\sqrt{2^{n+1}}} \sum_{\substack{c,r,h,x,y,L,R\\x_{0,c}\notin\Dom(L\cup R)\\ x\notin\Dom(L\cup R),\ x\neq x_{0,c}}} \frac{f(2^n-d)}{\sqrt{2^n-d}} \ket x\ket{\alpha_{c,r,y,L,R}}\ket h\ket{e_0} \ket L\ket{R\cup\{(x,y)\}}.
\]
There are $2^n-d-1$ allowed $x$ values, and $\ket a=J\ket{a'}$.
Removing the unique pair with input $x$ from $R$ and reading $e_0$ recovers all old labels, so the terms in $\ket{a'}$ have orthogonal labels.
Since $B$ leaves $L$ unchanged, $J$ contributes at most $q+1$ to the squared norm.
Using $f(2^n-d)\leq1/[2(2^n-q-1)]$ and \cref{eq:coefficient-energy},
\begin{align*}
  \|a'\|_2^2 &=\frac1{2^{n+1}} \sum_{\substack{c,r,h,y,L,R\\x_{0,c}\notin\Dom(L\cup R)}} \frac{2^n-d-1}{2^n-d}f(2^n-d)^2 \|\alpha_{c,r,y,L,R}\|_2^2\\
  &\leq\frac1{4(2^n-q-1)^2}\cdot\frac1{2^{n+1}} \sum_{c,r,h,y,L,R}\|\alpha_{c,r,y,L,R}\|_2^2\\
  &\leq\frac{\|\phi\|_2^2}{4(2^n-q-1)^2},\\
  \|a\|_2^2 &\leq(q+1)\|a'\|_2^2\\
  &\leq\frac{q+1}{4(2^n-q-1)^2}\|\phi\|_2^2.
\end{align*}
\paragraph{The excluded input.}
In $\ket c$, the final query label $x_{0,c}$ fixes $c$.
The unique pair with this input in the final relation $L$ fixes $e_0$; the unique pair with this input in the final relation $R$ fixes $(x_{0,c},y)$.
Removing these pairs recovers the old records.
Therefore
\begin{align*}
  \|c\|_2^2 &=\frac1{2^{n+1}} \sum_{\substack{c,r,h,y,L,R\\x_{0,c}\notin\Dom(L\cup R)}} \frac{\|\alpha_{c,r,y,L,R}\|_2^2}{2^n-d}\\
  &\leq\frac1{2^n-q}\|\phi\|_2^2.
\end{align*}
In $\ket a$, the final relation $L$ has no pair with input equal to the final query label; in $\ket c$, it has one.
These are different relation labels, so
\[
  \braket a c=0,\qquad \|[B,J]\phi\|_2^2=\|a\|_2^2+\|c\|_2^2 \leq\left(\frac{q+1}{4(2^n-q-1)^2}+\frac1{2^n-q}\right) \|\phi\|_2^2.
\]
\end{proof}
\begin{lemma}[Inverse deletion]
  \label{lem:commutator-a-dagger}
Let $0\leq q\leq2^n-2$ and let $\ket\phi$ have the form in \cref{eq:mixed-experiment-state-family}, with coefficient vectors independent of $h$ and supported on $|L\cup R|\leq q$.
Then
\begin{equation}
  \|[A^\dagger,J]\phi\|_2 \leq\left(\frac{\sqrt q}{2(2^n-q)} +\frac1{2^{|\reg H|}}\sqrt{\frac{2^n}{2^n-q}} +\sqrt{\frac q{2^{|\reg H|}(2^n-q+1)}}\right)\|\phi\|_2. \label{eq:deletion-summary}
\end{equation}
\end{lemma}
\begin{proof}
We relabel the query as $y$ and write $k=|\Img(L\cup R)|$.
For $(x,y)\in L$ with $y\notin\Img((L\setminus\{(x,y)\})\cup R)$,
\begin{equation}
  A^\dagger\ket y\ket L\ket R =\frac{\ket x\ket{L\setminus\{(x,y)\}}\ket R} {\sqrt{2^n-k+1}}; \label{eq:a-adjoint-action}
\end{equation}
otherwise the result is zero.
Removing this pair gives $|\Img((L\setminus\{(x,y)\})\cup R)|=k-1$, so
\[
  A\ket x\ket{L\setminus\{(x,y)\}}\ket R =\frac1{\sqrt{2^n-k+1}} \sum_{z\notin\Img((L\setminus\{(x,y)\})\cup R)} \ket z\ket{(L\setminus\{(x,y)\})\cup\{(x,z)\}}\ket R.
\]
The term $z=y$ gives \cref{eq:a-adjoint-action}.
The query label $y$ fixes the unique pair removed from $L$; no input contributes if the displayed conditions fail.

\paragraph{The two orders.}
Suppose $(x,y)$ satisfies the conditions in \cref{eq:a-adjoint-action} and $y\neq y_0$.
If $y_0\notin\Img(L\cup R)$, then
\begin{align*}
  A^\dagger J\ket y\ket{e_0}\ket L\ket R &=\frac{\ket x\ket{(L\setminus\{(x,y)\})\cup\{e_0\}}\ket R} {\sqrt{2^n-k}},\\
  JA^\dagger\ket y\ket{e_0}\ket L\ket R &=\frac{\ket x\ket{(L\setminus\{(x,y)\})\cup\{e_0\}}\ket R} {\sqrt{2^n-k+1}}.
\end{align*}
If $y_0\in\Img(L\cup R)$, both orders give
\[
  \frac{\sqrt{\operatorname{num}(L,e_0)+1}}{\sqrt{2^n-k+1}} \ket x\ket{(L\setminus\{(x,y)\})\cup\{e_0\}}\ket R.
\]
Removing $(x,y)$ does not change the number of copies of $e_0$, since $y\neq y_0$.
If $y$ cannot be removed, adding output $y_0\neq y$ does not change this, and both orders give zero.

For $y=y_0\notin\Img(L\cup R)$, $A^\dagger J$ removes the added pair:
\[
  A^\dagger J\ket{y_0}\ket{e_0}\ket L\ket R =\frac{\ket{x_{0,c}}\ket L\ket R}{\sqrt{2^n-k}},\qquad JA^\dagger\ket{y_0}\ket{e_0}\ket L\ket R=0.
\]
For $y=y_0$ with exactly one pair $(x,y_0)$ in $L$ and no output $y_0$ in $R$, insertion gives two pairs with output $y_0$.
Thus
\begin{align*}
  A^\dagger J\ket{y_0}\ket{e_0}\ket L\ket R&=0,\\
  JA^\dagger\ket{y_0}\ket{e_0}\ket L\ket R &=\frac{\ket x\ket{(L\setminus\{(x,y_0)\})\cup\{e_0\}}\ket R} {\sqrt{2^n-k+1}}.
\end{align*}
The insertion weight in the second line is one because deletion removes the only output $y_0$.
If $y_0$ occurs in $R$ or at least twice in $L$, both orders give zero.
For the common terms,
\[
  \frac1{\sqrt{2^n-k}}-\frac1{\sqrt{2^n-k+1}} =\frac{f(2^n-k+1)}{\sqrt{2^n-k+1}}.
\]
We define
\[
  \ket\beta=\ket x\ket{\alpha_{c,r,y,L,R}}\ket h \ket{(L\setminus\{(x,y)\})\cup\{e_0\}}\ket R.
\]
The difference of the two orders is
\begin{align*}
  [A^\dagger,J]\ket\phi&=\ket a+\ket b-\ket c,\\
  \ket a&=\frac1{\sqrt{2^{n+1}}} \sum_{\substack{c,r,h,y,x,L,R\\y\neq y_0,\ y_0\notin\Img(L\cup R)\\ (x,y)\in L,\ y\notin\Img((L\setminus\{(x,y)\})\cup R)}} \frac{f(2^n-k+1)}{\sqrt{2^n-k+1}}\ket\beta,\\
  \ket b&=\frac1{\sqrt{2^{n+1}}} \sum_{\substack{c,r,h,L,R\\y_0\notin\Img(L\cup R)}} \frac{\ket{x_{0,c}}\ket{\alpha_{c,r,(r,h),L,R}} \ket h\ket L\ket R}{\sqrt{2^n-k}},\\
  \ket c&=\frac1{\sqrt{2^{n+1}}} \sum_{\substack{c,r,h,x,y,L,R\\y=y_0,\ (x,y)\in L\\ y\notin\Img((L\setminus\{(x,y)\})\cup R)}} \frac{\ket\beta}{\sqrt{2^n-k+1}}.
\end{align*}

\paragraph{The normalization term.}
We define
\[
  \ket\eta=\frac1{\sqrt{2^{n+1}}} \sum_{\substack{c,r,h,y,L,R\\y\neq y_0,\ y_0\notin\Img(L\cup R)}} f(2^n-k+1)\ket y\ket{\alpha_{c,r,y,L,R}}\ket h \ket{e_0}\ket L\ket R.
\]
Then $\ket a=JA^\dagger\ket\eta$.
By \cref{eq:ratio}, $f(2^n-k+1)\leq1/[2(2^n-q)]$.
The relation in $\reg{E_L}$ has size at most $q-1$ in $A^\dagger\ket\eta$, so, for $q\geq1$,
\begin{align*}
  \|a\|_2^2 &\leq q\|A^\dagger\eta\|_2^2\leq q\|\eta\|_2^2\\
  &=\frac q{2^{n+1}} \sum_{\substack{c,r,h,y,L,R\\y\neq y_0,\ y_0\notin\Img(L\cup R)}} f(2^n-k+1)^2\|\alpha_{c,r,y,L,R}\|_2^2\\
  &\leq\frac q{4(2^n-q)^2}\|\phi\|_2^2.
\end{align*}
For $q=0$, $\ket a=0$.

\paragraph{Deletion of the inserted pair.}
For $\ket b$, different $r$ values can give the same labels $c,h,L,R$.
There are at most $2^{|\reg R|}$ such values, so the squared norm of their sum is at most $2^{|\reg R|}$ times the sum of their squared norms.
Using \cref{eq:matching-sums,eq:diagonal-energy},
\begin{align*}
  \|b\|_2^2 &=\frac1{2^{n+1}}\sum_{c,h,L,R} \left\|\sum_{\substack{r\\(r,h)\notin\Img(L\cup R)}} \frac{\alpha_{c,r,(r,h),L,R}}{\sqrt{2^n-k}}\right\|_2^2\\
  &\leq\frac{2^{|\reg R|}}{2^{n+1}(2^n-q)} \sum_{c,r,h,L,R}\|\alpha_{c,r,(r,h),L,R}\|_2^2\\
  &\leq\frac{2^{|\reg R|}\,2^{1+|\reg R|}}{2^{n+1}(2^n-q)} \|\phi\|_2^2 =\frac{2^n}{2^{2|\reg H|}(2^n-q)}\|\phi\|_2^2.
\end{align*}
The last equality uses $2|\reg R|-n=n-2|\reg H|$.

\paragraph{Blocked deletion.}
For $\ket c$, fix the final labels $(x,h,L',R)$.
Deleting one pair and inserting $e_0$ preserves the total size, so $|L'|\leq q$.
Each possible $e_0$ is one of the at most $q$ pairs in $L'$, and
\[
  e_0=(x_{0,c},(r,h)),\qquad L=(L'\setminus\{e_0\})\cup\{(x,(r,h))\}.
\]
Thus each choice of $e_0$ fixes all old labels.
Write
\[
  \mu=(c,r,h,x,L,R),\qquad \lambda_\mu=(x,h,(L\setminus\{(x,y_0)\})\cup\{e_0\},R),
\]
Here $\mu$ ranges over the terms in the sum defining $\ket c$.
Let $\alpha_\mu=\alpha_{c,r,(r,h),L,R}$ and $k_\mu=|\Img(L\cup R)|$.
There are at most $q$ indices with the same $\lambda_\mu$.
For fixed $L,y_0$, at most one $x$ satisfies the deletion condition, so dropping that condition does not count any coefficient twice.
By \cref{eq:matching-labels,eq:diagonal-energy},
\begin{align*}
  \|c\|_2^2 &=\frac1{2^{n+1}}\sum_\lambda \left\|\sum_{\substack{\mu\\\lambda_\mu=\lambda}} \frac{\alpha_\mu}{\sqrt{2^n-k_\mu+1}}\right\|_2^2\\
  &\leq\frac1{2^{n+1}}\sum_\mu \frac{\|\alpha_\mu\|_2^2}{2^n-k_\mu+1} \sum_{\substack{\widetilde\mu\\ \lambda_{\widetilde\mu}=\lambda_\mu}}1\\
  &\leq\frac q{2^{n+1}} \sum_{\substack{c,r,h,x,L,R\\(x,y_0)\in L\\ y_0\notin\Img((L\setminus\{(x,y_0)\})\cup R)}} \frac{\|\alpha_{c,r,(r,h),L,R}\|_2^2}{2^n-k+1}\\
  &\leq\frac q{2^{n+1}(2^n-q+1)} \sum_{c,r,h,L,R}\|\alpha_{c,r,(r,h),L,R}\|_2^2\\
  &\leq\frac q{2^{|\reg H|}(2^n-q+1)}\|\phi\|_2^2.
\end{align*}
The triangle inequality gives
\begin{align*}
  \|[A^\dagger,J]\phi\|_2 &\leq\|a\|_2+\|b\|_2+\|c\|_2\\
  &\leq\left(\frac{\sqrt q}{2(2^n-q)} +\frac1{2^{|\reg H|}}\sqrt{\frac{2^n}{2^n-q}} +\sqrt{\frac q{2^{|\reg H|}(2^n-q+1)}}\right)\|\phi\|_2.
\end{align*}
\end{proof}

\begin{lemma}[Forward deletion]
  \label{lem:commutator-b-dagger}
Let $0\leq q\leq2^n-2$ and let $\ket\phi$ have the form in \cref{eq:mixed-experiment-state-family}, with coefficient vectors independent of $h$ and supported on $|L\cup R|\leq q$.
Then
\[
  \|[B^\dagger,J]\phi\|_2 \leq\left(\frac{\sqrt{q+1}}{2(2^n-q)} +\sqrt{\frac2{2^n-q+1}}\right)\|\phi\|_2.
\]
\end{lemma}
\begin{proof}
Let $d=|\Dom(L\cup R)|$.
For $(x,y)\in R$ with $x\notin\Dom(L\cup(R\setminus\{(x,y)\}))$,
\[
  B^\dagger\ket x\ket L\ket R =\frac{\ket y\ket L\ket{R\setminus\{(x,y)\}}} {\sqrt{2^n-d+1}};
\]
otherwise the result is zero.
Removing this pair gives $|\Dom(L\cup(R\setminus\{(x,y)\}))|=d-1$, so
\[
  B\ket y\ket L\ket{R\setminus\{(x,y)\}} =\frac1{\sqrt{2^n-d+1}} \sum_{z\notin\Dom(L\cup(R\setminus\{(x,y)\}))} \ket z\ket L\ket{(R\setminus\{(x,y)\})\cup\{(z,y)\}}.
\]
The term $z=x$ gives the adjoint action.
The query label $x$ fixes the unique pair removed from $R$; no input contributes if the conditions fail.
\paragraph{The two orders.}
Suppose $(x,y)$ satisfies these conditions and $x\neq x_{0,c}$.
If $x_{0,c}\notin\Dom(L\cup R)$, then
\begin{align*}
  B^\dagger J\ket x\ket{e_0}\ket L\ket R &=\frac{\ket y\ket{L\cup\{e_0\}}\ket{R\setminus\{(x,y)\}}} {\sqrt{2^n-d}},\\
  JB^\dagger\ket x\ket{e_0}\ket L\ket R &=\frac{\ket y\ket{L\cup\{e_0\}}\ket{R\setminus\{(x,y)\}}} {\sqrt{2^n-d+1}}.
\end{align*}
If $x_{0,c}\in\Dom(L\cup R)$, both orders give
\[
  \frac{\sqrt{\operatorname{num}(L,e_0)+1}}{\sqrt{2^n-d+1}} \ket y\ket{L\cup\{e_0\}}\ket{R\setminus\{(x,y)\}}.
\]
For $x=x_{0,c}$ satisfying the deletion condition, adding $e_0$ to $L$ prevents deletion from $R$, so
\begin{align*}
  B^\dagger J\ket{x_{0,c}}\ket{e_0}\ket L\ket R&=0,\\
  JB^\dagger\ket{x_{0,c}}\ket{e_0}\ket L\ket R &=\frac{\ket y\ket{L\cup\{e_0\}}\ket{R\setminus\{(x_{0,c},y)\}}} {\sqrt{2^n-d+1}}.
\end{align*}
Here $x_{0,c}\notin\Dom(L)$, so the insertion weight is one.
If $x$ cannot be removed, adding a pair to $L$ does not change this, and both orders give zero.
For the common terms,
\[
  \frac1{\sqrt{2^n-d}}-\frac1{\sqrt{2^n-d+1}} =\frac{f(2^n-d+1)}{\sqrt{2^n-d+1}}.
\]
We define
\[
  \ket\beta=\ket y\ket{\alpha_{c,r,x,L,R}}\ket h \ket{L\cup\{e_0\}}\ket{R\setminus\{(x,y)\}}.
\]
Then
\begin{align*}
  [B^\dagger,J]\ket\phi&=\ket a-\ket c,\\
  \ket a&=\frac1{\sqrt{2^{n+1}}} \sum_{\substack{c,r,h,x,y,L,R\\x\neq x_{0,c},\ x_{0,c}\notin\Dom(L\cup R)\\ (x,y)\in R,\ x\notin\Dom(L\cup(R\setminus\{(x,y)\}))}} \frac{f(2^n-d+1)}{\sqrt{2^n-d+1}}\ket\beta,\\
  \ket c&=\frac1{\sqrt{2^{n+1}}} \sum_{\substack{c,r,h,x,y,L,R\\x=x_{0,c},\ (x,y)\in R\\ x\notin\Dom(L\cup(R\setminus\{(x,y)\}))}} \frac{\ket\beta}{\sqrt{2^n-d+1}}.
\end{align*}
The insertion weights in these sums are one.

\paragraph{The normalization term.}
We define
\[
  \ket\eta=\frac1{\sqrt{2^{n+1}}} \sum_{\substack{c,r,h,x,L,R\\x\neq x_{0,c},\ x_{0,c}\notin\Dom(L\cup R)}} f(2^n-d+1)\ket x\ket{\alpha_{c,r,x,L,R}}\ket h \ket{e_0}\ket L\ket R.
\]
Since $\ket a=JB^\dagger\ket\eta$ and $B^\dagger$ leaves $L$ unchanged, we use the bound $q+1$ on the squared norm of $J$, $\|B^\dagger\|\leq1$, and $f(2^n-d+1)\leq1/[2(2^n-q)]$ to get
\begin{align*}
  \|a\|_2^2 &\leq(q+1)\|B^\dagger\eta\|_2^2 \leq(q+1)\|\eta\|_2^2\\
  &=\frac{q+1}{2^{n+1}} \sum_{\substack{c,r,h,x,L,R\\x\neq x_{0,c},\ x_{0,c}\notin\Dom(L\cup R)}} f(2^n-d+1)^2\|\alpha_{c,r,x,L,R}\|_2^2\\
  &\leq\frac{q+1}{4(2^n-q)^2}\|\phi\|_2^2.
\end{align*}
\paragraph{Blocked deletion.}
For $\ket c$, fix the final labels $(y,h,L',R')$.
For each $c\in\{0,1\}$, $e_0$ must be the unique pair in $L'$ with input $x_{0,c}$; its output fixes $r,h$.
The old records are
\[
  L=L'\setminus\{e_0\},\qquad R=R'\cup\{(x_{0,c},y)\}.
\]
Thus there are at most two choices of old labels.
Write
\[
  \mu=(c,r,h,y,L,R),\qquad \lambda_\mu=(y,h,L\cup\{e_0\},R\setminus\{(x_{0,c},y)\}),
\]
Here $\mu$ ranges over the terms in the sum defining $\ket c$.
Let $\alpha_\mu=\alpha_{c,r,x_{0,c},L,R}$ and $d_\mu=|\Dom(L\cup R)|$.
For fixed $R,x_{0,c}$, at most one $y$ satisfies the deletion condition, so enlarging the coefficient sum does not count any term twice.
By \cref{eq:matching-labels},
\begin{align*}
  \|c\|_2^2 &=\frac1{2^{n+1}}\sum_\lambda \left\|\sum_{\substack{\mu\\\lambda_\mu=\lambda}} \frac{\alpha_\mu}{\sqrt{2^n-d_\mu+1}}\right\|_2^2\\
  &\leq\frac1{2^{n+1}}\sum_\mu \frac{\|\alpha_\mu\|_2^2}{2^n-d_\mu+1} \sum_{\substack{\widetilde\mu\\ \lambda_{\widetilde\mu}=\lambda_\mu}}1\\
  &\leq\frac2{2^{n+1}} \sum_{\substack{c,r,h,y,L,R\\(x_{0,c},y)\in R\\ x_{0,c}\notin\Dom(L\cup(R\setminus\{(x_{0,c},y)\}))}} \frac{\|\alpha_{c,r,x_{0,c},L,R}\|_2^2}{2^n-d+1}\\
  &\leq\frac2{2^n-q+1}\cdot\frac1{2^{n+1}} \sum_{c,r,h,x,L,R}\|\alpha_{c,r,x,L,R}\|_2^2\\
  &=\frac2{2^n-q+1}\|\phi\|_2^2,\\
  \|[B^\dagger,J]\phi\|_2 &\leq\|a\|_2+\|c\|_2 \leq\left(\frac{\sqrt{q+1}}{2(2^n-q)} +\sqrt{\frac2{2^n-q+1}}\right)\|\phi\|_2.
\end{align*}
\end{proof}

\begin{lemma}[Conjugate creation]
\label{lem:commutator-bar-a}
Let $0\leq q\leq2^n-2$ and let $\ket\phi$ have the form in \cref{eq:mixed-experiment-state-family}, with coefficient vectors independent of $h$ and supported on $|L\cup R|\leq q$.
Then
\[
  \|[\bar A,J]\phi\|_2^2\leq\left(\frac1{2^n-q}+\frac{q+1}{4(2^n-q-1)^2}\right)\|\phi\|_2^2.
\]
\end{lemma}
\begin{proof}
Let $k=|\Img(L\cup R)|\leq q$ and
\[
  \ket\beta=\ket y_\reg Q\ket{\alpha_{c,r,x,L,R}}_\reg W\ket h_\reg H\ket{L\cup\{e_0\}}_\reg{E_L}\ket{R\cup\{(x,y)\}}_\reg{E_R}.
\]
Inserting $e_0$ increases $k$ by one exactly when $y_0\notin\Img(L\cup R)$.

\paragraph{The two orders.}
The definitions give
\begin{align*}
  \bar AJ\ket\phi={}&\frac1{\sqrt{2^{n+1}}}\sum_{\substack{c,r,h,x,y,L,R\\y_0\notin\Img(L\cup R)\\y\notin\Img(L\cup R),\ y\neq y_0}}\frac{\ket\beta}{\sqrt{2^n-k-1}}\\*
  &+\frac1{\sqrt{2^{n+1}}}\sum_{\substack{c,r,h,x,y,L,R\\y_0\in\Img(L\cup R)\\y\notin\Img(L\cup R)}}\frac{\sqrt{\operatorname{num}(L,e_0)+1}}{\sqrt{2^n-k}}\ket\beta,\\
  J\bar A\ket\phi={}&\frac1{\sqrt{2^{n+1}}}\sum_{\substack{c,r,h,x,y,L,R\\y_0\notin\Img(L\cup R)\\y\notin\Img(L\cup R),\ y\neq y_0}}\frac{\ket\beta}{\sqrt{2^n-k}}\\*
  &+\frac1{\sqrt{2^{n+1}}}\sum_{\substack{c,r,h,x,y,L,R\\y_0\in\Img(L\cup R)\\y\notin\Img(L\cup R)}}\frac{\sqrt{\operatorname{num}(L,e_0)+1}}{\sqrt{2^n-k}}\ket\beta\\*
  &+\frac1{\sqrt{2^{n+1}}}\sum_{\substack{c,r,h,x,y,L,R\\y_0\notin\Img(L\cup R)\\y=y_0}}\frac{\ket\beta}{\sqrt{2^n-k}}.
\end{align*}
The last insertion weight is one because the new pair is in $R$ and $e_0\notin L$.
The second sums cancel, and the first sums differ only in their denominators.
Since
\[
  \frac1{\sqrt{2^n-k-1}}-\frac1{\sqrt{2^n-k}}=\frac{f(2^n-k)}{\sqrt{2^n-k}},
\]
we obtain
\begin{align*}
  [\bar A,J]\ket\phi&=\ket a-\ket c,\\
  \ket a&=\frac1{\sqrt{2^{n+1}}}\sum_{\substack{c,r,h,x,y,L,R\\y_0\notin\Img(L\cup R)\\y\notin\Img(L\cup R),\ y\neq y_0}}\frac{f(2^n-k)}{\sqrt{2^n-k}}\ket\beta,\\
  \ket c&=\frac1{\sqrt{2^{n+1}}}\sum_{\substack{c,r,h,x,y,L,R\\y_0\notin\Img(L\cup R)\\y=y_0}}\frac{\ket\beta}{\sqrt{2^n-k}}.
\end{align*}

\paragraph{The normalization term.}
Before applying $J$, let
\[
  \ket{a'}=\frac1{\sqrt{2^{n+1}}}\sum_{\substack{c,r,h,x,y,L,R\\y_0\notin\Img(L\cup R)\\y\notin\Img(L\cup R),\ y\neq y_0}}\frac{f(2^n-k)}{\sqrt{2^n-k}}\ket y\ket{\alpha_{c,r,x,L,R}}\ket h\ket{e_0}\ket L\ket{R\cup\{(x,y)\}}.
\]
The added pair is the unique pair in $R$ with output $y$.
Removing it and reading $e_0$ recovers all old labels, so distinct input labels give orthogonal terms in $\ket{a'}$.
There are $2^n-k-1$ allowed values of $y$, and $\ket a=J\ket{a'}$.
Since $\bar A$ leaves $L$ unchanged, we use the bound $q+1$ on the squared norm of $J$, $f(2^n-k)\leq1/[2(2^n-q-1)]$, and \cref{eq:coefficient-energy} to get
\begin{align*}
  \|a'\|_2^2&=\frac1{2^{n+1}}\sum_{\substack{c,r,h,x,L,R\\y_0\notin\Img(L\cup R)}}\frac{2^n-k-1}{2^n-k}f(2^n-k)^2\|\alpha_{c,r,x,L,R}\|_2^2\\
  &\leq\frac{\|\phi\|_2^2}{4(2^n-q-1)^2},\\
  \|a\|_2^2&\leq(q+1)\|a'\|_2^2\leq\frac{q+1}{4(2^n-q-1)^2}\|\phi\|_2^2.
\end{align*}
\paragraph{The excluded output.}
In $\ket c$, the final query label $y_0=(r,h)$ fixes $r,h$.
The unique pair with output $y_0$ in the final relation $L$ fixes $e_0$ and $c$; the unique pair with this output in the final relation $R$ fixes $x$.
Removing both pairs recovers the input labels, so distinct input labels give distinct output labels.
Therefore
\begin{align*}
  \|c\|_2^2&=\frac1{2^{n+1}}\sum_{\substack{c,r,h,x,L,R\\y_0\notin\Img(L\cup R)}}\frac{\|\alpha_{c,r,x,L,R}\|_2^2}{2^n-k}\\
  &\leq\frac1{2^n-q}\|\phi\|_2^2.
\end{align*}
In $\ket a$, the final relation $L$ has no pair with output equal to the final query label; in $\ket c$, it has one.
These are different relation labels, so
\[
  \braket a c=0,\qquad \|[\bar A,J]\phi\|_2^2=\|a\|_2^2+\|c\|_2^2\leq\left(\frac{q+1}{4(2^n-q-1)^2}+\frac1{2^n-q}\right)\|\phi\|_2^2.
\]
\end{proof}

\begin{lemma}[Transpose creation]
\label{lem:commutator-bar-b}
Let $0\leq q\leq2^n-2$ and let $\ket\phi$ have the form in \cref{eq:mixed-experiment-state-family}, with coefficient vectors independent of $h$ and supported on $|L\cup R|\leq q$.
Then
\[
  \|[\bar B,J]\phi\|_2^2\leq\left(\frac2{2^n-q}+\frac{q+2}{4(2^n-q-1)^2}\right)\|\phi\|_2^2.
\]
\end{lemma}
\begin{proof}
We relabel the query as $y$ and write $d=|\Dom(L\cup R)|\leq q$.
Inserting $e_0$ increases $d$ by one exactly when $x_{0,c}\notin\Dom(L\cup R)$.

\paragraph{The two orders.}
We use $w_{r,h,y}$ for the insertion weight on the branch $x=x_{0,c}\notin\Dom(L\cup R)$ and write the common output vector below:
\[
  w_{r,h,y}:=\begin{cases}\sqrt2,&y=y_0,\\1,&y\neq y_0,\end{cases}\qquad \ket\beta=\ket x\ket{\alpha_{c,r,y,L,R}}\ket h\ket{L\cup\{e_0,(x,y)\}}\ket R.
\]
Then
\begin{align*}
  \bar BJ\ket\phi={}&\frac1{\sqrt{2^{n+1}}}\sum_{\substack{c,r,h,x,y,L,R\\x_{0,c}\notin\Dom(L\cup R)\\x\notin\Dom(L\cup R),\ x\neq x_{0,c}}}\frac{\ket\beta}{\sqrt{2^n-d-1}}\\*
  &+\frac1{\sqrt{2^{n+1}}}\sum_{\substack{c,r,h,x,y,L,R\\x_{0,c}\in\Dom(L\cup R)\\x\notin\Dom(L\cup R)}}\frac{\sqrt{\operatorname{num}(L,e_0)+1}}{\sqrt{2^n-d}}\ket\beta,\\
  J\bar B\ket\phi={}&\frac1{\sqrt{2^{n+1}}}\sum_{\substack{c,r,h,x,y,L,R\\x_{0,c}\notin\Dom(L\cup R)\\x\notin\Dom(L\cup R),\ x\neq x_{0,c}}}\frac{\ket\beta}{\sqrt{2^n-d}}\\*
  &+\frac1{\sqrt{2^{n+1}}}\sum_{\substack{c,r,h,x,y,L,R\\x_{0,c}\in\Dom(L\cup R)\\x\notin\Dom(L\cup R)}}\frac{\sqrt{\operatorname{num}(L,e_0)+1}}{\sqrt{2^n-d}}\ket\beta\\*
  &+\frac1{\sqrt{2^{n+1}}}\sum_{\substack{c,r,h,x,y,L,R\\x_{0,c}\notin\Dom(L\cup R)\\x=x_{0,c}}}\frac{w_{r,h,y}}{\sqrt{2^n-d}}\ket\beta.
\end{align*}
For the last sum, $L$ has no pair with input $x_{0,c}$ before either insertion.
The two inserted pairs coincide exactly when $y=y_0$, giving the weight $\sqrt2$.
The second sums cancel, and the first sums differ only in their denominators.
Using
\[
  \frac1{\sqrt{2^n-d-1}}-\frac1{\sqrt{2^n-d}}=\frac{f(2^n-d)}{\sqrt{2^n-d}},
\]
we have
\begin{align*}
  [\bar B,J]\ket\phi&=\ket a-\ket c,\\
  \ket a&=\frac1{\sqrt{2^{n+1}}}\sum_{\substack{c,r,h,x,y,L,R\\x_{0,c}\notin\Dom(L\cup R)\\x\notin\Dom(L\cup R),\ x\neq x_{0,c}}}\frac{f(2^n-d)}{\sqrt{2^n-d}}\ket\beta,\\
  \ket c&=\frac1{\sqrt{2^{n+1}}}\sum_{\substack{c,r,h,x,y,L,R\\x_{0,c}\notin\Dom(L\cup R)\\x=x_{0,c}}}\frac{w_{r,h,y}}{\sqrt{2^n-d}}\ket\beta.
\end{align*}

\paragraph{The normalization term.}
Let
\[
  \ket{a'}=\frac1{\sqrt{2^{n+1}}}\sum_{\substack{c,r,h,x,y,L,R\\x_{0,c}\notin\Dom(L\cup R)\\x\notin\Dom(L\cup R),\ x\neq x_{0,c}}}\frac{f(2^n-d)}{\sqrt{2^n-d}}\ket x\ket{\alpha_{c,r,y,L,R}}\ket h\ket{e_0}\ket{L\cup\{(x,y)\}}\ket R.
\]
The added pair is identified by $x$.
Removing it and reading $e_0$ recovers all old labels, so distinct input labels give orthogonal terms in $\ket{a'}$.
There are $2^n-d-1$ allowed values of $x$.
Since $\ket a=J\ket{a'}$ and the final relation $L$ has size at most $q+1$ before $J$, we use the bound $q+2$ on the squared norm of $J$, $f(2^n-d)\leq1/[2(2^n-q-1)]$, and \cref{eq:coefficient-energy} to get
\begin{align*}
  \|a'\|_2^2&=\frac1{2^{n+1}}\sum_{\substack{c,r,h,y,L,R\\x_{0,c}\notin\Dom(L\cup R)}}\frac{2^n-d-1}{2^n-d}f(2^n-d)^2\|\alpha_{c,r,y,L,R}\|_2^2\\
  &\leq\frac{\|\phi\|_2^2}{4(2^n-q-1)^2},\\
  \|a\|_2^2&\leq(q+2)\|a'\|_2^2\leq\frac{q+2}{4(2^n-q-1)^2}\|\phi\|_2^2.
\end{align*}
\paragraph{The excluded input.}
For $\ket c$, write
\[
  \mu=(c,r,h,y,L,R),\qquad \lambda_\mu=(x_{0,c},h,L\cup\{e_0,(x_{0,c},y)\},R),\qquad w_\mu=w_{r,h,y}.
\]
The indices satisfy the conditions in the sum defining $\ket c$.
Let $\alpha_\mu=\alpha_{c,r,y,L,R}$ and $d_\mu=|\Dom(L\cup R)|$.
The final query label $x_{0,c}$ fixes $c$.
The two pairs with this input recover the old $L$ after removal.
Choosing which pair is $e_0$ fixes $r,h$ and the old query $y$.
If the pairs are distinct, there are at most two choices, each with weight one.
If they coincide, there is one choice with weight $\sqrt2$.
Thus
\[
  w_\mu^2\sum_{\substack{\widetilde\mu\\\lambda_{\widetilde\mu}=\lambda_\mu}}1\leq2.
\]
By \cref{eq:matching-labels},
\begin{align*}
  \|c\|_2^2&=\frac1{2^{n+1}}\sum_\lambda\left\|\sum_{\substack{\mu\\\lambda_\mu=\lambda}}\frac{w_\mu\alpha_\mu}{\sqrt{2^n-d_\mu}}\right\|_2^2\\
  &\leq\frac1{2^{n+1}}\sum_\mu\frac{w_\mu^2\|\alpha_\mu\|_2^2}{2^n-d_\mu}\sum_{\substack{\widetilde\mu\\\lambda_{\widetilde\mu}=\lambda_\mu}}1\\
  &\leq\frac2{2^{n+1}}\sum_{\substack{c,r,h,y,L,R\\x_{0,c}\notin\Dom(L\cup R)}}\frac{\|\alpha_{c,r,y,L,R}\|_2^2}{2^n-d}\\
  &\leq\frac2{2^n-q}\|\phi\|_2^2.
\end{align*}
In $\ket a$, the final relation $L$ contains one pair with input equal to the final query label; in $\ket c$, it contains two.
These are different relation labels, so
\[
  \braket a c=0,\qquad \|[\bar B,J]\phi\|_2^2=\|a\|_2^2+\|c\|_2^2\leq\left(\frac{q+2}{4(2^n-q-1)^2}+\frac2{2^n-q}\right)\|\phi\|_2^2.
\]
\end{proof}

\begin{lemma}[Transpose deletion]
\label{lem:commutator-bar-a-dagger}
Let $0\leq q\leq2^n-2$ and let $\ket\phi$ have the form in \cref{eq:mixed-experiment-state-family}, with coefficient vectors independent of $h$ and supported on $|L\cup R|\leq q$.
Then
\[
  \|[\bar A^\dagger,J]\phi\|_2\leq\left(\frac{\sqrt q}{2(2^n-q)}+\sqrt{\frac q{2^{|\reg H|}(2^n-q+1)}}\right)\|\phi\|_2.
\]
\end{lemma}
\begin{proof}
We relabel the query as $y$ and write $k=|\Img(L\cup R)|$.
For $(x,y)\in R$ with $y\notin\Img(L\cup(R\setminus\{(x,y)\}))$,
\begin{equation}
  \bar A^\dagger\ket y\ket L\ket R=\frac{\ket x\ket L\ket{R\setminus\{(x,y)\}}}{\sqrt{2^n-k+1}}; \label{eq:bar-a-adjoint-action}
\end{equation}
otherwise the result is zero.
Removing this pair reduces the image size to $k-1$, so
\[
  \bar A\ket x\ket L\ket{R\setminus\{(x,y)\}}=\frac1{\sqrt{2^n-k+1}}\sum_{z\notin\Img(L\cup(R\setminus\{(x,y)\}))}\ket z\ket L\ket{(R\setminus\{(x,y)\})\cup\{(x,z)\}}.
\]
The term $z=y$ gives \cref{eq:bar-a-adjoint-action}.
The query label $y$ fixes the unique pair removed from $R$; no input contributes if the conditions fail.

\paragraph{The two orders.}
Suppose $(x,y)$ satisfies these conditions and $y\neq y_0$.
If $y_0\notin\Img(L\cup R)$, then
\begin{align*}
  \bar A^\dagger J\ket y\ket{e_0}\ket L\ket R&=\frac{\ket x\ket{L\cup\{e_0\}}\ket{R\setminus\{(x,y)\}}}{\sqrt{2^n-k}},\\
  J\bar A^\dagger\ket y\ket{e_0}\ket L\ket R&=\frac{\ket x\ket{L\cup\{e_0\}}\ket{R\setminus\{(x,y)\}}}{\sqrt{2^n-k+1}}.
\end{align*}
If $y_0\in\Img(L\cup R)$, both orders give
\[
  \frac{\sqrt{\operatorname{num}(L,e_0)+1}}{\sqrt{2^n-k+1}}\ket x\ket{L\cup\{e_0\}}\ket{R\setminus\{(x,y)\}}.
\]
Deletion leaves $L$ unchanged, and insertion of output $y_0\neq y$ does not affect the deletion condition.
If $y$ cannot be removed, both orders give zero.

For $y=y_0$ satisfying the deletion condition, $R$ has exactly one pair with output $y_0$ and $L$ has none.
Inserting $e_0$ into $L$ prevents deletion from $R$, so
\begin{align*}
  \bar A^\dagger J\ket{y_0}\ket{e_0}\ket L\ket R&=0,\\
  J\bar A^\dagger\ket{y_0}\ket{e_0}\ket L\ket R&=\frac{\ket x\ket{L\cup\{e_0\}}\ket{R\setminus\{(x,y_0)\}}}{\sqrt{2^n-k+1}}.
\end{align*}
The insertion weight is one because $y_0\notin\Img(L)$.
If $y_0$ cannot be removed, adding a pair to $L$ cannot make it removable, so both orders give zero.
For the common terms,
\[
  \frac1{\sqrt{2^n-k}}-\frac1{\sqrt{2^n-k+1}}=\frac{f(2^n-k+1)}{\sqrt{2^n-k+1}}.
\]
We define
\[
  \ket\beta=\ket x\ket{\alpha_{c,r,y,L,R}}\ket h\ket{L\cup\{e_0\}}\ket{R\setminus\{(x,y)\}}.
\]
Then
\begin{align*}
  [\bar A^\dagger,J]\ket\phi&=\ket a-\ket c,\\
  \ket a&=\frac1{\sqrt{2^{n+1}}}\sum_{\substack{c,r,h,y,x,L,R\\y\neq y_0,\ y_0\notin\Img(L\cup R)\\(x,y)\in R,\ y\notin\Img(L\cup(R\setminus\{(x,y)\}))}}\frac{f(2^n-k+1)}{\sqrt{2^n-k+1}}\ket\beta,\\
  \ket c&=\frac1{\sqrt{2^{n+1}}}\sum_{\substack{c,r,h,x,y,L,R\\y=y_0,\ (x,y)\in R\\y\notin\Img(L\cup(R\setminus\{(x,y)\}))}}\frac{\ket\beta}{\sqrt{2^n-k+1}}.
\end{align*}

\paragraph{The normalization term.}
Let
\[
  \ket\eta=\frac1{\sqrt{2^{n+1}}}\sum_{\substack{c,r,h,y,L,R\\y\neq y_0,\ y_0\notin\Img(L\cup R)}}f(2^n-k+1)\ket y\ket{\alpha_{c,r,y,L,R}}\ket h\ket{e_0}\ket L\ket R.
\]
The vector satisfies $\ket a=J\bar A^\dagger\ket\eta$.
By \cref{eq:ratio}, $f(2^n-k+1)\leq1/[2(2^n-q)]$.
For every nonzero deletion, $R$ contains at least one pair before deletion, so $|L|\leq q-1$.
For $q\geq1$, this gives
\begin{align*}
  \|a\|_2^2&\leq q\|\bar A^\dagger\eta\|_2^2\leq q\|\eta\|_2^2\\
  &=\frac q{2^{n+1}}\sum_{\substack{c,r,h,y,L,R\\y\neq y_0,\ y_0\notin\Img(L\cup R)}}f(2^n-k+1)^2\|\alpha_{c,r,y,L,R}\|_2^2\\
  &\leq\frac q{4(2^n-q)^2}\|\phi\|_2^2.
\end{align*}
For $q=0$, $\ket a=0$.

\paragraph{Blocked deletion.}
For $\ket c$, fix the final labels $(x,h,L',R')$.
Deletion followed by insertion preserves total record size, so $|L'|\leq q$.
Each possible $e_0$ is one of the pairs in $L'$, and
\[
  e_0=(x_{0,c},(r,h)),\qquad L=L'\setminus\{e_0\},\qquad R=R'\cup\{(x,(r,h))\}.
\]
Thus choosing $e_0$ recovers $c,r$, the old query $y_0$, and the old records.
Write
\[
  \mu=(c,r,h,x,L,R),\qquad \lambda_\mu=(x,h,L\cup\{e_0\},R\setminus\{(x,y_0)\}),
\]
Here $\mu$ ranges over the terms in the sum defining $\ket c$.
Let $\alpha_\mu=\alpha_{c,r,(r,h),L,R}$ and $k_\mu=|\Img(L\cup R)|$.
There are at most $q$ indices with the same final label.
For fixed $R,y_0$, at most one $x$ satisfies the deletion condition, so enlarging the sum does not count any coefficient twice.
By \cref{eq:matching-labels,eq:diagonal-energy},
\begin{align*}
  \|c\|_2^2&=\frac1{2^{n+1}}\sum_\lambda\left\|\sum_{\substack{\mu\\\lambda_\mu=\lambda}}\frac{\alpha_\mu}{\sqrt{2^n-k_\mu+1}}\right\|_2^2\\
  &\leq\frac1{2^{n+1}}\sum_\mu\frac{\|\alpha_\mu\|_2^2}{2^n-k_\mu+1}\sum_{\substack{\widetilde\mu\\\lambda_{\widetilde\mu}=\lambda_\mu}}1\\
  &\leq\frac q{2^{n+1}}\sum_{\substack{c,r,h,x,L,R\\(x,y_0)\in R\\y_0\notin\Img(L\cup(R\setminus\{(x,y_0)\}))}}\frac{\|\alpha_{c,r,(r,h),L,R}\|_2^2}{2^n-k+1}\\
  &\leq\frac q{2^{n+1}(2^n-q+1)}\sum_{c,r,h,L,R}\|\alpha_{c,r,(r,h),L,R}\|_2^2\\
  &\leq\frac q{2^{|\reg H|}(2^n-q+1)}\|\phi\|_2^2.
\end{align*}
The triangle inequality gives
\[
  \|[\bar A^\dagger,J]\phi\|_2\leq\|a\|_2+\|c\|_2\leq\left(\frac{\sqrt q}{2(2^n-q)}+\sqrt{\frac q{2^{|\reg H|}(2^n-q+1)}}\right)\|\phi\|_2.
\]
\end{proof}

\begin{lemma}[Conjugate deletion]
\label{lem:commutator-bar-b-dagger}
Let $0\leq q\leq2^n-2$ and let $\ket\phi$ have the form in \cref{eq:mixed-experiment-state-family}, with coefficient vectors independent of $h$ and supported on $|L\cup R|\leq q$.
Then
\[
  \|[\bar B^\dagger,J]\phi\|_2\leq\left(\frac{\sqrt q}{2(2^n-q)}+\sqrt{\frac2{2^n-q}}+\sqrt{\frac2{2^n-q+1}}\right)\|\phi\|_2.
\]
\end{lemma}
\begin{proof}
Let $d=|\Dom(L\cup R)|$.
For $(x,y)\in L$ with $x\notin\Dom((L\setminus\{(x,y)\})\cup R)$,
\begin{equation}
  \bar B^\dagger\ket x\ket L\ket R=\frac{\ket y\ket{L\setminus\{(x,y)\}}\ket R}{\sqrt{2^n-d+1}}; \label{eq:bar-b-adjoint-action}
\end{equation}
otherwise the result is zero.
Removing this pair reduces the domain size to $d-1$, so
\[
  \bar B\ket y\ket{L\setminus\{(x,y)\}}\ket R=\frac1{\sqrt{2^n-d+1}}\sum_{z\notin\Dom((L\setminus\{(x,y)\})\cup R)}\ket z\ket{(L\setminus\{(x,y)\})\cup\{(z,y)\}}\ket R.
\]
The term $z=x$ gives \cref{eq:bar-b-adjoint-action}.
The query label $x$ fixes the unique pair removed from $L$; no input contributes if the conditions fail.

\paragraph{The two orders.}
Suppose $(x,y)$ satisfies these conditions and $x\neq x_{0,c}$.
If $x_{0,c}\notin\Dom(L\cup R)$, then
\begin{align*}
  \bar B^\dagger J\ket x\ket{e_0}\ket L\ket R&=\frac{\ket y\ket{(L\setminus\{(x,y)\})\cup\{e_0\}}\ket R}{\sqrt{2^n-d}},\\
  J\bar B^\dagger\ket x\ket{e_0}\ket L\ket R&=\frac{\ket y\ket{(L\setminus\{(x,y)\})\cup\{e_0\}}\ket R}{\sqrt{2^n-d+1}}.
\end{align*}
If $x_{0,c}\in\Dom(L\cup R)$, both orders give
\[
  \frac{\sqrt{\operatorname{num}(L,e_0)+1}}{\sqrt{2^n-d+1}}\ket y\ket{(L\setminus\{(x,y)\})\cup\{e_0\}}\ket R.
\]
Since $x\neq x_{0,c}$, deleting $(x,y)$ does not change the number of copies of $e_0$.
If $x$ cannot be removed, adding a pair with different input does not change this, and both orders give zero.

For $x=x_{0,c}\notin\Dom(L\cup R)$, deletion after insertion removes $e_0$:
\[
  \bar B^\dagger J\ket{x_{0,c}}\ket{e_0}\ket L\ket R=\frac{\ket{y_0}\ket L\ket R}{\sqrt{2^n-d}},\qquad J\bar B^\dagger\ket{x_{0,c}}\ket{e_0}\ket L\ket R=0.
\]
For $x=x_{0,c}$ with exactly one pair $(x_{0,c},y)$ in $L$ and no input $x_{0,c}$ in $R$, insertion gives two pairs with this input.
Hence
\begin{align*}
  \bar B^\dagger J\ket{x_{0,c}}\ket{e_0}\ket L\ket R&=0,\\
  J\bar B^\dagger\ket{x_{0,c}}\ket{e_0}\ket L\ket R&=\frac{\ket y\ket{(L\setminus\{(x_{0,c},y)\})\cup\{e_0\}}\ket R}{\sqrt{2^n-d+1}}.
\end{align*}
The insertion weight in the second line is one because deletion removes the only input $x_{0,c}$.
If $x_{0,c}$ occurs in $R$ or at least twice in $L$, both orders give zero.
For the common terms,
\[
  \frac1{\sqrt{2^n-d}}-\frac1{\sqrt{2^n-d+1}}=\frac{f(2^n-d+1)}{\sqrt{2^n-d+1}}.
\]
We define
\[
  \ket\beta=\ket y\ket{\alpha_{c,r,x,L,R}}\ket h\ket{(L\setminus\{(x,y)\})\cup\{e_0\}}\ket R.
\]
Then
\begin{align*}
  [\bar B^\dagger,J]\ket\phi&=\ket a+\ket b-\ket c,\\
  \ket a&=\frac1{\sqrt{2^{n+1}}}\sum_{\substack{c,r,h,x,y,L,R\\x\neq x_{0,c},\ x_{0,c}\notin\Dom(L\cup R)\\(x,y)\in L,\ x\notin\Dom((L\setminus\{(x,y)\})\cup R)}}\frac{f(2^n-d+1)}{\sqrt{2^n-d+1}}\ket\beta,\\
  \ket b&=\frac1{\sqrt{2^{n+1}}}\sum_{\substack{c,r,h,L,R\\x_{0,c}\notin\Dom(L\cup R)}}\frac{\ket{(r,h)}\ket{\alpha_{c,r,x_{0,c},L,R}}\ket h\ket L\ket R}{\sqrt{2^n-d}},\\
  \ket c&=\frac1{\sqrt{2^{n+1}}}\sum_{\substack{c,r,h,x,y,L,R\\x=x_{0,c},\ (x,y)\in L\\x\notin\Dom((L\setminus\{(x,y)\})\cup R)}}\frac{\ket\beta}{\sqrt{2^n-d+1}}.
\end{align*}
The insertion weights in these sums are one.

\paragraph{The normalization term.}
Let
\[
  \ket\eta=\frac1{\sqrt{2^{n+1}}}\sum_{\substack{c,r,h,x,L,R\\x\neq x_{0,c},\ x_{0,c}\notin\Dom(L\cup R)}}f(2^n-d+1)\ket x\ket{\alpha_{c,r,x,L,R}}\ket h\ket{e_0}\ket L\ket R.
\]
The vector satisfies $\ket a=J\bar B^\dagger\ket\eta$.
By \cref{eq:ratio}, $f(2^n-d+1)\leq1/[2(2^n-q)]$.
After deletion, $L$ has size at most $q-1$, so, for $q\geq1$,
\begin{align*}
  \|a\|_2^2&\leq q\|\bar B^\dagger\eta\|_2^2\leq q\|\eta\|_2^2\\
  &=\frac q{2^{n+1}}\sum_{\substack{c,r,h,x,L,R\\x\neq x_{0,c},\ x_{0,c}\notin\Dom(L\cup R)}}f(2^n-d+1)^2\|\alpha_{c,r,x,L,R}\|_2^2\\
  &\leq\frac q{4(2^n-q)^2}\|\phi\|_2^2.
\end{align*}
For $q=0$, $\ket a=0$.

\paragraph{Deletion of the inserted pair.}
In $\ket b$, the final query label $(r,h)$ fixes $r,h$, while the two values of $c$ can give the same output labels.
By \cref{eq:matching-sums,eq:coefficient-energy},
\begin{align*}
  \|b\|_2^2&=\frac1{2^{n+1}}\sum_{r,h,L,R}\left\|\sum_{\substack{c\in\{0,1\}\\x_{0,c}\notin\Dom(L\cup R)}}\frac{\alpha_{c,r,x_{0,c},L,R}}{\sqrt{2^n-d}}\right\|_2^2\\
  &\leq\frac2{2^{n+1}(2^n-q)}\sum_{c,r,h,L,R}\|\alpha_{c,r,x_{0,c},L,R}\|_2^2\\
  &\leq\frac2{2^n-q}\cdot\frac1{2^{n+1}}\sum_{c,r,h,x,L,R}\|\alpha_{c,r,x,L,R}\|_2^2\\
  &=\frac2{2^n-q}\|\phi\|_2^2.
\end{align*}

\paragraph{Blocked deletion.}
For $\ket c$, fix the final labels $(y,h,L',R)$.
For each $c\in\{0,1\}$, $e_0$ is the unique pair in $L'$ with input $x_{0,c}$.
Its output fixes $r,h$, and
\[
  L=(L'\setminus\{e_0\})\cup\{(x_{0,c},y)\},\qquad \text{input query label}=x_{0,c}.
\]
Thus there are at most two choices of old labels.
Write
\[
  \mu=(c,r,h,y,L,R),\qquad \lambda_\mu=(y,h,(L\setminus\{(x_{0,c},y)\})\cup\{e_0\},R),
\]
Here $\mu$ ranges over the terms in the sum defining $\ket c$.
Let $\alpha_\mu=\alpha_{c,r,x_{0,c},L,R}$ and $d_\mu=|\Dom(L\cup R)|$.
For fixed $L,x_{0,c}$, at most one $y$ satisfies the deletion condition, so enlarging the sum does not count any coefficient twice.
By \cref{eq:matching-labels},
\begin{align*}
  \|c\|_2^2&=\frac1{2^{n+1}}\sum_\lambda\left\|\sum_{\substack{\mu\\\lambda_\mu=\lambda}}\frac{\alpha_\mu}{\sqrt{2^n-d_\mu+1}}\right\|_2^2\\
  &\leq\frac1{2^{n+1}}\sum_\mu\frac{\|\alpha_\mu\|_2^2}{2^n-d_\mu+1}\sum_{\substack{\widetilde\mu\\\lambda_{\widetilde\mu}=\lambda_\mu}}1\\
  &\leq\frac2{2^{n+1}}\sum_{\substack{c,r,h,y,L,R\\(x_{0,c},y)\in L\\x_{0,c}\notin\Dom((L\setminus\{(x_{0,c},y)\})\cup R)}}\frac{\|\alpha_{c,r,x_{0,c},L,R}\|_2^2}{2^n-d+1}\\
  &\leq\frac2{2^n-q+1}\cdot\frac1{2^{n+1}}\sum_{c,r,h,x,L,R}\|\alpha_{c,r,x,L,R}\|_2^2\\
  &=\frac2{2^n-q+1}\|\phi\|_2^2.
\end{align*}
The triangle inequality gives
\begin{align*}
  \|[\bar B^\dagger,J]\phi\|_2&\leq\|a\|_2+\|b\|_2+\|c\|_2\\
  &\leq\left(\frac{\sqrt q}{2(2^n-q)}+\sqrt{\frac2{2^n-q}}+\sqrt{\frac2{2^n-q+1}}\right)\|\phi\|_2.
\end{align*}
\end{proof}

\endgroup

\subsection{Constants in the recording comparison}
\label{app:comparison}

\Cref{lem:commutator-a,lem:commutator-b,lem:commutator-a-dagger,lem:commutator-b-dagger,lem:commutator-bar-a,lem:commutator-bar-b,lem:commutator-bar-a-dagger,lem:commutator-bar-b-dagger} give $\|[X,J]\phi\|_2\leq\epsilon_X\|\phi\|_2$.
Taking nonnegative square roots where needed, the coefficients are
\begin{align}
  \epsilon_A^2&=\frac2{2^n-q}+\frac{q+2}{4(2^n-q-1)^2},\nonumber\\
  \epsilon_B^2&=\frac1{2^n-q}+\frac{q+1}{4(2^n-q-1)^2},\nonumber\\
  \epsilon_{A^\dagger}&=\frac{\sqrt q}{2(2^n-q)}+\frac1{2^{|\reg H|}}\sqrt{\frac{2^n}{2^n-q}}+\sqrt{\frac q{2^{|\reg H|}(2^n-q+1)}},\nonumber\\
  \epsilon_{B^\dagger}&=\frac{\sqrt{q+1}}{2(2^n-q)}+\sqrt{\frac2{2^n-q+1}},\nonumber\\
  \epsilon_{\bar A}&=\epsilon_B,\qquad \epsilon_{\bar B}=\epsilon_A,\nonumber\\
  \epsilon_{\bar A^\dagger}&=\frac{\sqrt q}{2(2^n-q)}+\sqrt{\frac q{2^{|\reg H|}(2^n-q+1)}},\nonumber\\
  \epsilon_{\bar B^\dagger}&=\frac{\sqrt q}{2(2^n-q)}+\sqrt{\frac2{2^n-q}}+\sqrt{\frac2{2^n-q+1}}. \label{eq:eight-epsilons}
\end{align}
The following expansions give the four commutator coefficients used in \cref{eq:v-commutator,eq:query-commutator}.
Expanding the definitions and taking their adjoints gives
\begin{align*}
  V&=A-ABB^\dagger+B^\dagger-AA^\dagger B^\dagger,\\
  V^\dagger&=A^\dagger-BB^\dagger A^\dagger+B-BAA^\dagger,\\
  \bar V&=\bar A-\bar A\bar B\bar B^\dagger+\bar B^\dagger-\bar A\bar A^\dagger\bar B^\dagger,\\
  \bar V^\dagger&=\bar A^\dagger-\bar B\bar B^\dagger\bar A^\dagger+\bar B-\bar B\bar A\bar A^\dagger.
\end{align*}
Since
\[
  [XY,J]=XYJ-JXY=X(YJ-JY)+(XJ-JX)Y=X[Y,J]+[X,J]Y,
\]
we obtain
\begin{align*}
  [V,J]={}&[A,J]-AB[B^\dagger,J]-A[B,J]B^\dagger-[A,J]BB^\dagger\\
  &+[B^\dagger,J]-AA^\dagger[B^\dagger,J]-A[A^\dagger,J]B^\dagger-[A,J]A^\dagger B^\dagger,\\
  [V^\dagger,J]={}&[A^\dagger,J]-BB^\dagger[A^\dagger,J]-B[B^\dagger,J]A^\dagger-[B,J]B^\dagger A^\dagger\\
  &+[B,J]-BA[A^\dagger,J]-B[A,J]A^\dagger-[B,J]AA^\dagger,\\
  [\bar V,J]={}&[\bar A,J]-\bar A\bar B[\bar B^\dagger,J]-\bar A[\bar B,J]\bar B^\dagger-[\bar A,J]\bar B\bar B^\dagger\\
  &+[\bar B^\dagger,J]-\bar A\bar A^\dagger[\bar B^\dagger,J]-\bar A[\bar A^\dagger,J]\bar B^\dagger-[\bar A,J]\bar A^\dagger\bar B^\dagger,\\
  [\bar V^\dagger,J]={}&[\bar A^\dagger,J]-\bar B\bar B^\dagger[\bar A^\dagger,J]-\bar B[\bar B^\dagger,J]\bar A^\dagger-[\bar B,J]\bar B^\dagger\bar A^\dagger\\
  &+[\bar B,J]-\bar B\bar A[\bar A^\dagger,J]-\bar B[\bar A,J]\bar A^\dagger-[\bar B,J]\bar A\bar A^\dagger.
\end{align*}
The operators to the right of these commutators belong to
\begin{align*}
  &\{\id,B^\dagger,BB^\dagger,A^\dagger B^\dagger,A^\dagger,B^\dagger A^\dagger,AA^\dagger\},\\
  &\{\id,\bar B^\dagger,\bar B\bar B^\dagger,\bar A^\dagger\bar B^\dagger,\bar A^\dagger,\bar B^\dagger\bar A^\dagger,\bar A\bar A^\dagger\}.
\end{align*}
Their norms are at most one, they act trivially on $\reg{HE_0}$, and they do not increase $|L\cup R|$.
Indeed, a deletion reduces the total size by one, and a deletion followed by a creation preserves it.
Thus each resulting state still has the form in \cref{eq:mixed-experiment-state-family}.
The operators to the left have norm at most one as well.
The triangle inequality and the eight commutator bounds give
\begin{align*}
  \|[V,J]\phi\|_2&\leq(3\epsilon_A+\epsilon_B+\epsilon_{A^\dagger}+3\epsilon_{B^\dagger})\|\phi\|_2,\\
  \|[V^\dagger,J]\phi\|_2&\leq(\epsilon_A+3\epsilon_B+3\epsilon_{A^\dagger}+\epsilon_{B^\dagger})\|\phi\|_2,\\
  \|[\bar V,J]\phi\|_2&\leq(3\epsilon_{\bar A}+\epsilon_{\bar B}+\epsilon_{\bar A^\dagger}+3\epsilon_{\bar B^\dagger})\|\phi\|_2,\\
  \|[\bar V^\dagger,J]\phi\|_2&\leq(\epsilon_{\bar A}+3\epsilon_{\bar B}+3\epsilon_{\bar A^\dagger}+\epsilon_{\bar B^\dagger})\|\phi\|_2.
\end{align*}
For example, the expansion of $[V,J]$ has three terms containing $[A,J]$, one containing $[B,J]$, one containing $[A^\dagger,J]$, and three containing $[B^\dagger,J]$.
Since $\|\phi'_t\|_2\leq1$, these coefficients bound each query error in the telescoping identity.

\subsubsection{Bounds on the commutator constants}
\label{app:commutator-constants}

\begin{proof}[Proof of \cref{lem:query-commutator}]
For $1\leq q<2^{|\reg H|}$ and $|\reg R|\geq|\reg H|$, we have $|\reg H|\geq1$, $n\geq2$, and
\[
  2^n\geq2^{2|\reg H|},\qquad q\leq2^{|\reg H|}-1\leq\frac{2^{2|\reg H|}}4\leq2^{n-2}.
\]
The middle inequality follows from $(2^{|\reg H|}-2)^2\geq0$.
Consequently,
\begin{align*}
  2^n-q&\geq\frac34\,2^n\geq2^{n-1},\\
  2^n-q-1&\geq\frac34\,2^n-1\geq2^{n-1},\\
  2^n-q+1&\geq\frac34\,2^n,\\
  q+1&\leq2^{|\reg H|}\leq2^{n-1},\qquad q+2\leq2^{|\reg H|}+1\leq2^n.
\end{align*}
Here $n\geq2$, so $2^n/4\geq1$.
Using the weaker denominator bounds $2^n-q,2^n-q-1\geq2^{n-1}$ first,
\begin{align*}
  \epsilon_A^2&\leq\frac4{2^n}+\frac{q+2}{2^{2n}}\leq\frac5{2^n},\\
  \epsilon_B^2&\leq\frac2{2^n}+\frac{q+1}{2^{2n}}\leq\frac3{2^n},\\
  \epsilon_A&\leq\frac{\sqrt5}{2^{n/2}}\leq\frac3{2^{|\reg H|}},\qquad \epsilon_B\leq\frac{\sqrt3}{2^{n/2}}\leq\frac2{2^{|\reg H|}}.
\end{align*}
Also,
\begin{align*}
  \frac{\sqrt q}{2(2^n-q)}&\leq\frac{\sqrt{2^{n-1}}}{2^n}=\frac1{\sqrt{2\cdot2^n}},\\
  \frac{\sqrt{q+1}}{2(2^n-q)}&\leq\frac1{\sqrt{2\cdot2^n}},\\
  \sqrt{\frac q{2^{|\reg H|}(2^n-q+1)}}&\leq\sqrt{\frac{2^{|\reg H|}}{2^{|\reg H|}\,2^{n-1}}}=\sqrt{\frac2{2^n}},\\
  \sqrt{\frac{2^n}{2^n-q}}&\leq\sqrt2,\qquad \sqrt{\frac2{2^n-q+1}}\leq\frac2{2^{n/2}}.
\end{align*}
Therefore
\begin{align*}
  \epsilon_{A^\dagger}&\leq\frac1{\sqrt{2\cdot2^n}}+\frac{\sqrt2}{2^{|\reg H|}}+\sqrt{\frac2{2^n}}\leq\frac5{\sqrt2\,2^{|\reg H|}}\leq\frac4{2^{|\reg H|}},\\
  \epsilon_{B^\dagger}&\leq\frac1{\sqrt{2\cdot2^n}}+\frac2{2^{n/2}}\leq\frac3{2^{|\reg H|}}.
\end{align*}
For the conjugate terms, use $2^n-q,2^n-q+1\geq3\cdot2^n/4$:
\begin{align*}
  \frac{\sqrt q}{2(2^n-q)}&\leq\frac{2\sqrt{2^{|\reg H|}}}{3\cdot2^n}\leq\frac2{3\sqrt{2^{|\reg H|}}}\frac1{2^{|\reg H|}}\leq\frac{\sqrt2}{3\cdot2^{|\reg H|}},\\
  \sqrt{\frac q{2^{|\reg H|}(2^n-q+1)}}&\leq\frac2{\sqrt3\,2^{n/2}}\leq\frac2{\sqrt3\,2^{|\reg H|}},\\
  \sqrt{\frac2{2^n-q}}&\leq\frac{\sqrt{8/3}}{2^{n/2}}\leq\frac{\sqrt{8/3}}{2^{|\reg H|}},\\
  \sqrt{\frac2{2^n-q+1}}&\leq\frac{\sqrt{8/3}}{2^{n/2}}\leq\frac{\sqrt{8/3}}{2^{|\reg H|}}.
\end{align*}
Since $\sqrt2\leq2$, $2/\sqrt3\leq4/3$, and $\sqrt{8/3}\leq5/3$, we obtain
\begin{align*}
  \epsilon_{\bar A}&=\epsilon_B\leq\frac2{2^{|\reg H|}},\qquad \epsilon_{\bar B}=\epsilon_A\leq\frac3{2^{|\reg H|}},\\
  \epsilon_{\bar A^\dagger}&\leq\frac{\sqrt2/3+2/\sqrt3}{2^{|\reg H|}}\leq\frac{2/3+4/3}{2^{|\reg H|}}=\frac2{2^{|\reg H|}},\\
  \epsilon_{\bar B^\dagger}&\leq\frac{\sqrt2/3+2\sqrt{8/3}}{2^{|\reg H|}}\leq\frac{2/3+10/3}{2^{|\reg H|}}=\frac4{2^{|\reg H|}}.
\end{align*}
Squaring the three positive inequalities gives $2\leq4$, $12\leq16$, and $24\leq25$, respectively.
Substituting these bounds in the four query estimates gives
\begin{align*}
  3\epsilon_A+\epsilon_B+\epsilon_{A^\dagger}+3\epsilon_{B^\dagger}&\leq\frac{3\cdot3+2+4+3\cdot3}{2^{|\reg H|}}=\frac{24}{2^{|\reg H|}},\\
  \epsilon_A+3\epsilon_B+3\epsilon_{A^\dagger}+\epsilon_{B^\dagger}&\leq\frac{3+3\cdot2+3\cdot4+3}{2^{|\reg H|}}=\frac{24}{2^{|\reg H|}},\\
  3\epsilon_{\bar A}+\epsilon_{\bar B}+\epsilon_{\bar A^\dagger}+3\epsilon_{\bar B^\dagger}&\leq\frac{3\cdot2+3+2+3\cdot4}{2^{|\reg H|}}=\frac{23}{2^{|\reg H|}},\\
  \epsilon_{\bar A}+3\epsilon_{\bar B}+3\epsilon_{\bar A^\dagger}+\epsilon_{\bar B^\dagger}&\leq\frac{2+3\cdot3+3\cdot2+4}{2^{|\reg H|}}=\frac{21}{2^{|\reg H|}}.
\end{align*}
Thus all four query errors have coefficient at most $24/2^{|\reg H|}$, as claimed in \cref{eq:query-commutator}.
\end{proof}

\section{Query-bound reductions}
\label{app:query-reductions}

The controlled-query replacement preserves both Haar-averaged experiment states with the same query count.
We prove this equality here and verify the constants in \cref{cor:query-lower-bounds}.

\subsection{Controlled queries}
\label{app:controlled-proof}

\begin{proof}[Proof of \cref{cor:controlled-queries}]
At $q=0$ there are no queries to replace.
For $q\geq1$, fix the query types $U_1,\ldots,U_q$ and put
\[
  \sigma_i=\begin{cases}1,&U_i\in\{U,U^{\mathsf T}\},\\-1,&U_i\in\{U^\dagger,U^*\}.\end{cases}
\]
Fix either experiment and purify its input, including all reference registers.

\paragraph{Grouping the control branches.}
Expand each controlled query as the sum of its two control branches.
A choice $b=(b_1,\ldots,b_q)\in\{0,1\}^q$ has signed query count $\sum_i\sigma_i b_i$.
Grouping these terms by signed count gives the original output vector
\[
  \sum_{k=-q}^q\ket{v_k(U)}.
\]
Each term is linear in the initial input vector.

\paragraph{The uncontrolled circuit.}
We use the construction of Tang--Wright~\cite[Theorem~1.1 and Claim~2.6]{TW25}, with a counter taken modulo $2q+1$ and a fixed maximally entangled auxiliary pair.
On an active branch, the query acts on $\reg Q$ and the counter changes by $\sigma_i$.
On an inactive branch, controlled swaps route the same single query to the auxiliary pair.
Forward and inverse queries use its first half; conjugate and transpose queries use its second half.
Their Claim~2.6 gives the replacement output
\[
  \sum_{k=-q}^q\ket{v_k(U)}\ket k\ket{w_k(U)},\qquad \|w_k(U)\|_2=1.
\]
For a fixed query schedule, the auxiliary vector depends only on $U,k$, not on the individual control choices with that count.
Indeed,
\[
  \ket{w_k(U)}=(U^{\sum_{i=1}^q\sigma_i-k}\otimes\id)\frac1{\sqrt{2^n}}\sum_{x\in[2^n]}\ket{x,x}.
\]
Forward and inverse queries contribute their signed powers directly; conjugate and transpose queries contribute the same powers by $(\id\otimes U^{\mathsf T})\sum_x\ket{x,x}=(U\otimes\id)\sum_x\ket{x,x}$ and $U^*=(U^\dagger)^{\mathsf T}$.
Thus the auxiliary vector has unit norm.
The construction uses one uncontrolled query of the same type per original call.
Its added gates and initial registers are independent of $U$ and the input.
The induction in Claim~2.6 is linear in the initial vector, so the same construction applies pointwise to the oracle-dependent real input and preserves arbitrary reference entanglement.
Their Remark~2.7 also permits the real experiment's uncontrolled preparation query to remain unchanged.

The counter labels in $[-q,q]$ are distinct modulo $2q+1$.
Tracing out the counter and auxiliary pair therefore leaves
\[
  \sum_{k=-q}^q\proj{v_k(U)}.
\]
\paragraph{Equality after Haar averaging.}
It remains to compare this with the original Haar-averaged output.
Replacing $U$ by $e^{i\theta}U$ multiplies query type $i$ by $e^{i\sigma_i\theta}$.
The mixed preparation is independent of $U$, and the real preparation acquires one common global phase.
In either experiment, these preparation phases cancel in the outer product, so
\[
  \ketbra{v_k(e^{i\theta}U)}{v_{k'}(e^{i\theta}U)}=e^{i(k-k')\theta}\ketbra{v_k(U)}{v_{k'}(U)}.
\]
Haar invariance now gives, for every $\theta$,
\[
  \mathbb E_U\ketbra{v_k(U)}{v_{k'}(U)}=e^{i(k-k')\theta}\mathbb E_U\ketbra{v_k(U)}{v_{k'}(U)}.
\]
For $k\neq k'$, choose $\theta$ with $e^{i(k-k')\theta}\neq1$; the expectation is then zero.
Consequently,
\[
  \mathbb E_U\sum_{k,k'}\ketbra{v_k(U)}{v_{k'}(U)}=\mathbb E_U\sum_k\proj{v_k(U)}.
\]
Tracing out the experiment's hidden registers preserves this equality.
Thus the same replacement circuit preserves both averaged algorithm states with exactly $q$ queries.
Forward-only circuits remain forward-only.
The bounds now follow from \cref{thm:all-haar,cor:forward-haar}.
\end{proof}

\subsection{Constants in the query lower bounds}
\label{app:query-lower-bounds}

We give explicit constants for \cref{cor:query-lower-bounds}.
For forward-only access, \cref{eq:forward-haar-bound} implies
\[
  \Pr[a=b]-\frac12\leq\frac{12+3\sqrt2}{4}\frac q{2^{|\reg H|}}.
\]
Thus advantage at least $\eta>0$ requires $q\geq4\eta\,2^{|\reg H|}/(12+3\sqrt2)$ for every query count.

For all four query types, \cref{eq:all-haar-bound} gives
\[
  \Pr[a=b]-\frac12\leq\frac{25+24\sqrt2}{4}\frac q{2^{|\reg H|}}+\frac{9(q+1)^2}{2^{n/8}}.
\]
Fix $0<\eta\leq1/2$ and let
\[
  \kappa=\min\left\{\frac{\eta}{25+24\sqrt2},\sqrt{\frac\eta{144}}\right\}.
\]
For $1\leq q\leq\kappa\min\{2^{|\reg H|},2^{n/16}\}$, we have $q+1\leq2q$, so
\begin{align*}
  \Pr[a=b]-\frac12&\leq\frac{(25+24\sqrt2)\kappa}4+\frac{36q^2}{2^{n/8}}\\
  &\leq\frac{(25+24\sqrt2)\kappa}4+36\kappa^2\leq\frac\eta4+\frac\eta4=\frac\eta2<\eta.
\end{align*}
At $q=0$, the experiments agree by \cref{eq:forward-haar-bound}, since a circuit with no queries is the same in every access model.
This proves the stated asymptotic lower bound.
Every bidirectional algorithm is also an algorithm with access to all four types, so it obeys the same bound.
Controlled queries obey the same inequalities by \cref{cor:controlled-queries}.
For selected queries, \cref{cor:selected-queries} replaces $q$ by $4q$; use $\kappa/4$ in the preceding argument, which changes only the lower-bound constant.

\subsection{Decoding on a constant fraction of oracles}
\label{app:typical-decoding}

\begin{proof}[Proof of \cref{cor:typical-decoding}]
To calibrate the fixed decoder $\mathcal D$, we use $N=\lceil4/(p\delta)^2\rceil$ trials.

\paragraph{The trials and sign estimate.}
The auxiliary distinguisher holds its supplied challenge aside while preparing $N$ fresh real challenges, each with its own reference qubit.
It runs $\mathcal D$ on each radiation register and measures the corresponding reference and output with the Bell projector.
At initialization, it is supplied with $N+1$ tensor-product copies of the decoder's oracle-independent auxiliary state; no unknown state is cloned.
The trials use disjoint registers, so conditioned on $U$ they are independent Bernoulli variables with mean $f_{\mathcal D}(U)$.
Their empirical mean $\widehat f$ has conditional variance at most $1/(4N)$, and therefore
\[
  \mathbb E[|\widehat f-f_{\mathcal D}(U)|\mid U]\leq\frac1{2\sqrt N}.
\]
Let $s=1$ when $\widehat f\geq1/4$ and $s=-1$ otherwise.
A wrong sign requires $|\widehat f-f_{\mathcal D}(U)|\geq|f_{\mathcal D}(U)-1/4|$.
For either sign,
\[
  s\bigl(f_{\mathcal D}(U)-1/4\bigr)\geq|f_{\mathcal D}(U)-1/4|-2|\widehat f-f_{\mathcal D}(U)|.
\]
\paragraph{Testing the original challenge.}
The distinguisher now runs $\mathcal D$ on the original challenge and performs the Bell test, reversing its outcome when $s=-1$.
Conditioned on $U$, this run is independent of the trials.
In the mixed experiment, its Bell-test probability is $1/4$ for every $U$, even conditioned on the trial outcomes.
Its Haar-averaged acceptance gap is consequently at least
\[
  \mathbb E_U|f_{\mathcal D}(U)-1/4|-\frac1{\sqrt N}\geq p\delta-\frac1{\sqrt N}\geq\frac{p\delta}2.
\]
\paragraph{The query count.}
Each trial uses one forward preparation query and at most $q$ decoding queries, while the final run uses at most $q$ queries.
The total is at most $q+N(q+1)\leq(N+1)(q+1)$.
The auxiliary circuit and its initial state are independent of $U$, and its distinguishing advantage is at least $p\delta/4$.
It uses the same access model as the decoder, with an additional forward query per trial.
\Cref{cor:query-lower-bounds} therefore gives the asserted lower bounds on $q$ as the corresponding lower-bound scale grows.
\end{proof}

\section{Decoding and Uhlmann transformation}
\label{app:decoding-proofs}

This appendix gives the full proofs of the decoding and Uhlmann results in \cref{sec:decodability-upper,sec:uhlmann-rank}.
The last subsection gives the fidelity-estimation comparison from \cref{sec:related-work}.

\subsection{The decoupling estimate}
\label{app:decoupling}

We prove \cref{eq:decoupling-trace-bound} by computing the expected purity of $\rho_{\reg{BH}}^U$.
The calculation uses only the two Haar columns
\[
  \ket v=U\ket{x_{0,0}},\qquad \ket w=U\ket{x_{0,1}}.
\]
Then $\ket{\Psi_0^{\mathrm{real},U}}=(\ket0_{\reg B}\ket v+\ket1_{\reg B}\ket w)/\sqrt2$.

\paragraph{One Haar column.}
The vector $v$ is uniform on the unit sphere in $\mathbb C^{2^n}$.
For $i,j,k,\ell\in[2^n]$, its fourth moments are
\[
  \mathbb E[v_i\overline{v_j}v_k\overline{v_\ell}]=\frac{\delta_{i,j}\delta_{k,\ell}+\delta_{i,\ell}\delta_{k,j}}{2^n(2^n+1)}.
\]
To verify the formula, invariance under separate coordinate phases forces the ket indices $\{i,k\}$ to equal the bra indices $\{j,\ell\}$ as multisets.
For distinct $i,j$, the coordinate $(v_i+v_j)/\sqrt2$ has the same distribution as $v_i$.
Expanding $\mathbb E|(v_i+v_j)/\sqrt2|^4$ and using phase invariance gives
\[
  \mathbb E|v_i|^4=\frac12\mathbb E|v_i|^4+\mathbb E[|v_i|^2|v_j|^2],
\]
so $\mathbb E|v_i|^4=2\mathbb E[|v_i|^2|v_j|^2]$.
Coordinate permutations make these expectations independent of the distinct indices $i,j$.
Expanding $1=\mathbb E(\sum_i|v_i|^2)^2$ now gives $\mathbb E[|v_i|^2|v_j|^2]=1/[2^n(2^n+1)]$ for $i\neq j$, proving the formula.

Applying it to the coefficients indexed by $(r,h)$ gives
\begin{align*}
  \mathbb E\Tr\left[(\Tr_{\reg R}\proj v)^2\right]&=\sum_{r,s,h,k}\mathbb E[v_{r,h}\overline{v_{r,k}}v_{s,k}\overline{v_{s,h}}]\\
  &=\frac{2^{2|\reg R|+|\reg H|}+2^{|\reg R|+2|\reg H|}}{2^n(2^n+1)}=\frac{2^{|\reg R|}+2^{|\reg H|}}{2^n+1}.
\end{align*}
The two moment terms require $h=k$ and $r=s$, respectively, giving the two counts in the numerator.

\paragraph{The second column.}
Conditioned on $v$, the vector $w$ is uniform on the unit sphere in its orthogonal complement.
Its expected projector is therefore
\[
  \mathbb E[\proj w\mid v]=\frac{\id-\proj v}{2^n-1}.
\]
Indeed, invariance under unitaries on this complement makes the expectation a scalar multiple of its projection, and trace one fixes the scalar.
Expanding the cross term and using this conditional expectation gives
\begin{align*}
  &\mathbb E\left[\Tr\left[(\Tr_{\reg R}\ketbra vw)(\Tr_{\reg R}\ketbra wv)\right]\middle|v\right]\\
  &\qquad=\frac1{2^n-1}\sum_{r,s,h,k}v_{r,h}\overline{v_{s,h}}\left(\delta_{r,s}-v_{s,k}\overline{v_{r,k}}\right)\\
  &\qquad=\frac{2^{|\reg H|}-\Tr[(\Tr_{\reg R}\proj v)^2]}{2^n-1}.
\end{align*}
In the first term, $r=s$ and the sum over $k$ gives $2^{|\reg H|}$; the remaining squared coefficients sum to one.
The second term is the purity computed above.
Averaging over $v$ yields
\[
  \mathbb E\Tr\left[(\Tr_{\reg R}\ketbra vw)(\Tr_{\reg R}\ketbra wv)\right]=\frac{2^{|\reg R|}(2^{2|\reg H|}-1)}{2^{2n}-1}.
\]

\paragraph{The marginal purity.}
In the basis $\ket0,\ket1$ of $\reg B$, the marginal is
\[
  \rho_{\reg{BH}}^U=\frac12\begin{pmatrix}\Tr_{\reg R}\proj v&\Tr_{\reg R}\ketbra vw\\\Tr_{\reg R}\ketbra wv&\Tr_{\reg R}\proj w\end{pmatrix}.
\]
The two columns have the same marginal distribution.
Squaring this block matrix therefore gives
\begin{align*}
  \mathbb E\Tr[(\rho_{\reg{BH}}^U)^2]&=\frac12\mathbb E\Tr[(\Tr_{\reg R}\proj v)^2]+\frac12\mathbb E\Tr[(\Tr_{\reg R}\ketbra vw)(\Tr_{\reg R}\ketbra wv)]\\
  &=\frac{2^{|\reg H|}(2^{2|\reg R|}-1)+2^{|\reg R|+1}(2^{2|\reg H|}-1)}{2(2^{2n}-1)}.
\end{align*}
Since $\tau_{\reg{BH}}=\id_{\reg{BH}}/2^{1+|\reg H|}$ and $\Tr\rho_{\reg{BH}}^U=1$, expanding the square gives
\begin{align*}
  \mathbb E\Tr[(\rho_{\reg{BH}}^U-\tau_{\reg{BH}})^2]&=\mathbb E\Tr[(\rho_{\reg{BH}}^U)^2]-2^{-(1+|\reg H|)}\\
  &=\frac{(2^{2|\reg H|}-1)(2^{n+1}-1)}{2^{1+|\reg H|}(2^{2n}-1)}.
\end{align*}

\paragraph{From purity to trace norm.}
Cauchy--Schwarz on the $2^{1+|\reg H|}$ eigenvalues of a Hermitian operator gives $\|X\|_1^2\leq2^{1+|\reg H|}\Tr(X^2)$.
Cauchy--Schwarz for the Haar expectation then gives
\begin{align*}
  \left(\mathbb E_U\|\rho_{\reg{BH}}^U-\tau_{\reg{BH}}\|_1\right)^2&\leq\mathbb E_U\|\rho_{\reg{BH}}^U-\tau_{\reg{BH}}\|_1^2\\
  &\leq2^{1+|\reg H|}\mathbb E_U\Tr[(\rho_{\reg{BH}}^U-\tau_{\reg{BH}})^2]\\
  &=\frac{(2^{2|\reg H|}-1)(2^{n+1}-1)}{2^{2n}-1}\leq2\cdot2^{|\reg H|-|\reg R|}.
\end{align*}
For the last inequality, use $2^{2|\reg H|}-1\leq2^{2|\reg H|}$ and $(2^{n+1}-1)/(2^{2n}-1)\leq2^{1-n}$.
The latter follows by multiplying through by the positive denominator and using $2^n\geq2$.
Taking square roots and weakening $\sqrt2$ to $2$ proves \cref{eq:decoupling-trace-bound}.

\subsection{The decoding upper bound}
\label{app:decoding-upper}

The cutoff below depends only on the dimensions and $\epsilon$, so the same decoding circuit works for every $U$.
\begin{proof}[Proof of \cref{cor:decoding-upper-bound}]
We use the preparation unitary $P_U$ and fixed target $\ket\Theta$ from \cref{sec:decodability-upper}.
We apply the construction in Utsumi et al.~\cite[Theorem~13 and Proposition~10]{UNWT25} with a uniform singular-value cutoff.
Their fidelity $F$ is squared fidelity, while $\mathrm F$ here denotes root fidelity.

\paragraph{The overlap operator.}
The construction aligns the left and right singular vectors of the overlap operator on the accessible register $\reg R$:
\[
  M_U=\Tr_{\reg{BH}}\ketbra{\Theta}{\Psi_0^{\mathrm{real},U}}.
\]
The Schmidt vectors of $\ket\Theta$ on $\reg R$ are orthonormal, and each has coefficient $2^{-(1+|\reg H|)/2}$.
Multiplying the partial-trace expansion by its adjoint therefore gives
\[
  M_U^\dagger M_U=\frac{\Tr_{\reg{BH}}\proj{\Psi_0^{\mathrm{real},U}}}{2^{1+|\reg H|}}.
\]
The two marginals of the pure source state have the same nonzero eigenvalues.
Also, $\tau_{\reg{BH}}=\id_{\reg{BH}}/2^{1+|\reg H|}$.
Thus the nonzero singular values $s_j$ of $M_U$ are those of $\sqrt{\tau_{\reg{BH}}}\sqrt{\rho_{\reg{BH}}^U}$, and
\[
  r=\operatorname{rank}\bigl(\sqrt{\tau_{\reg{BH}}}\sqrt{\rho_{\reg{BH}}^U}\bigr)\leq2^{1+|\reg H|},\qquad \sum_{j=1}^r s_j=\mathrm F(\rho_{\reg{BH}}^U,\tau_{\reg{BH}})\leq1.
\]

\paragraph{A uniform cutoff.}
Utsumi et al.'s Eqs.~(36)--(42) implement a block encoding of $M_U$ using preparation calls on auxiliary copies of the register spaces and a swap involving the actual $\reg R$.
The actual $\reg{BH}$ is untouched.
Their Proposition~10 gives a real odd sign polynomial $p$ satisfying
\[
  |p(s)-1|\leq\frac\epsilon4\quad\text{for }\frac{\epsilon}{8\cdot2^{1+|\reg H|}}\leq s\leq1,\qquad |p(s)|\leq1\quad\text{for }0\leq s\leq1.
\]
The degree is $O((2^{|\reg H|}/\epsilon)\log(1/\epsilon))$, and the polynomial and its gates depend only on the dimensions and $\epsilon$.
The construction applies $p$ to the singular-value diagonal of $M_U$.
Their Eqs.~(50)--(52) express the overlap, and Eqs.~(55)--(57) give its singular-value expansion.
Thus its absolute overlap with the target and zero auxiliary registers is
\[
  a_U=\left|\sum_{j=1}^r p(s_j)s_j\right|.
\]

\paragraph{The overlap error.}
Above the cutoff, the approximation error is at most $\epsilon/4$ and the singular values sum to at most one.
Below it, $|1-p(s_j)|\leq2$, and at most $r$ singular values contribute.
As in their Eqs.~(59)--(62), we bound the two ranges separately:
\begin{align*}
  \sum_{s_j\geq\epsilon/(8\cdot2^{1+|\reg H|})}|1-p(s_j)|s_j&\leq\frac\epsilon4\sum_j s_j\leq\frac\epsilon4,\\
  \sum_{s_j<\epsilon/(8\cdot2^{1+|\reg H|})}|1-p(s_j)|s_j&\leq\frac{2\epsilon r}{8\cdot2^{1+|\reg H|}}\leq\frac\epsilon4.
\end{align*}
The reverse triangle inequality now gives
\[
  \bigl|\mathrm F(\rho_{\reg{BH}}^U,\tau_{\reg{BH}})-a_U\bigr|\leq\sum_j|1-p(s_j)|s_j\leq\frac\epsilon2.
\]
Both quantities lie in $[0,1]$, so the difference of their squares is at most twice their difference:
\[
  \bigl|\mathrm F(\rho_{\reg{BH}}^U,\tau_{\reg{BH}})^2-a_U^2\bigr|\leq2\bigl|\mathrm F(\rho_{\reg{BH}}^U,\tau_{\reg{BH}})-a_U\bigr|\leq\epsilon.
\]
The squared overlap with zero auxiliary registers is one nonnegative summand of the target overlap after tracing out those registers.
Thus discarding them gives \cref{eq:uhlmann-fidelity-error} for every $U$.

\paragraph{Query cost and decoding fidelity.}
The singular-value transformation uses a number of block-encoding calls proportional to the degree of $p$.
Each call to $P_U$ or $P_U^\dagger=C^\dagger(\id_{\reg{B}}\otimes U^\dagger)$ costs one query to $U$ or $U^\dagger$; target-preparation calls use only fixed gates.
The decoder therefore uses $O((2^{|\reg H|}/\epsilon)\log(1/\epsilon))$ queries, which is $O(2^{|\reg H|})$ at fixed $\epsilon$.
By \cref{lem:purification-decoder}, its fidelity for each $U$ is at least $1-\|\rho_{\reg{BH}}^U-\tau_{\reg{BH}}\|_1-\epsilon$.
Averaging and applying \cref{eq:decoupling-trace-bound} give
\[
  \mathbb E_U f_{\mathcal D}(U)\geq1-2\cdot2^{(|\reg H|-|\reg R|)/2}-\epsilon,
\]
as required.
\end{proof}

\subsection{Rank and the transformation lower bound}
\label{app:uhlmann-rank-proof}

\begin{proof}[Proof of \cref{lem:uhlmann-rank}]
Since $\tau_{\reg{BH}}$ is invertible and square roots preserve rank, it suffices to show that $\rho_{\reg{BH}}^U$ has full rank.
Let $X_U$ be the coefficient matrix of $\ket{\Psi_0^{\mathrm{real},U}}$ across $\reg{BH}:\reg R$:
\[
  (X_U)_{(c,h),r}=\frac1{\sqrt2}\bra{r,h}U\ket{x_{0,c}},\qquad \rho_{\reg{BH}}^U=X_UX_U^\dagger.
\]
It suffices to show that $X_U$ has full row rank.

\paragraph{Generating the two columns.}
We generate the two Haar columns by orthogonalizing two independent standard complex Gaussian vectors $g_0,g_1$ on $\reg{RH}$.
Their joint distribution is invariant under every unitary on $\reg{RH}$, and orthogonalization preserves this invariance.
The unitary group acts transitively on orthonormal pairs, so the resulting invariant distribution is the Haar distribution on such pairs.
Before orthogonalizing, arrange their coefficients in matrices $(G_c)_{h,r}=(g_c)_{r,h}$.
The stacked matrix is
\[
  G=\begin{pmatrix}G_0\\G_1\end{pmatrix}.
\]
It has $2^{1+|\reg H|}$ rows and $2^{|\reg R|}$ columns, with independent Gaussian entries.
The hypothesis $|\reg R|\geq|\reg H|+1$ makes the number of columns at least the number of rows.

\paragraph{Full rank before and after orthogonalization.}
Conditioned on the preceding rows, each next row has a density on the ambient space of row vectors.
The preceding rows span a proper linear subspace, since there are fewer of them than columns.
A proper linear subspace has Lebesgue measure zero, so the rows are linearly independent with probability one.
For orthogonalization, let
\[
  s=\frac{\braket{g_0}{g_1}}{\|g_0\|_2^2},\qquad t_0=\|g_0\|_2,\qquad t_1=\|g_1-sg_0\|_2.
\]
Also, $t_0,t_1>0$ with probability one.
The two normalized columns are $g_0/t_0$ and $(g_1-sg_0)/t_1$.
Consequently,
\[
  X_U=\frac1{\sqrt2}\begin{pmatrix}t_0^{-1}\id&0\\-s\,t_1^{-1}\id&t_1^{-1}\id\end{pmatrix}G
\]
in distribution, where each identity block has size $2^{|\reg H|}$.
The block matrix is invertible, so it preserves row rank.
Thus $X_UX_U^\dagger$ has full rank $2^{1+|\reg H|}$ with probability one.
\end{proof}

\label{app:uhlmann-lower-proof}
\begin{proof}[Proof of \cref{thm:uhlmann-rank-lower-bound}]
We reduce the transformation algorithm to a decoder with the same query count.

\paragraph{Converting preparation queries.}
Taking adjoints, conjugates, and transposes of $P_U=(\id_{\reg{B}}\otimes U)C$ gives
\begin{align*}
  P_U&=(\id_{\reg{B}}\otimes U)C,& P_U^\dagger&=C^\dagger(\id_{\reg{B}}\otimes U^\dagger),\\
  P_U^*&=(\id_{\reg{B}}\otimes U^*)C^*,& P_U^{\mathsf T}&=C^{\mathsf T}(\id_{\reg{B}}\otimes U^{\mathsf T}).
\end{align*}
The factors involving $C$ are fixed.
For a coherent choice, apply $C,C^*$ on the forward and conjugate branches before the selected query, and $C^\dagger,C^{\mathsf T}$ on the inverse and transpose branches afterwards.
The identity branch applies neither factor.
Thus each source-preparation query costs one selected query to $U^?$, and queries to the target preparation circuit are fixed gates.

\paragraph{From transformation to decoding.}
Let $\mathcal D^U=\Tr_{\reg{H'R'}}\circ\mathcal T^U$.
The target $\ket\Theta$ lies in the Bell-test acceptance space, so
\[
  \proj{\Phi^+}_{\reg{BA}}\otimes\id_{\reg{HH'R'}}\geq\proj\Theta.
\]
Taking expectations in the transformed state proves, for every $U$,
\[
  f_{\mathcal D}(U)\geq\bra\Theta(\operatorname{id}_{\reg{BH}}\otimes\mathcal T^U)(\proj{\Psi_0^{\mathrm{real},U}})\ket\Theta.
\]
The theorem's averaged transformation condition therefore gives
\[
  \mathbb E_U f_{\mathcal D}(U)\geq\mathbb E_U\mathrm F(\rho_{\reg{BH}}^U,\tau_{\reg{BH}})^2-\epsilon.
\]
By \cref{eq:fidelity-decoupling,eq:decoupling-trace-bound}, the right-hand side is at least $1-2\cdot2^{(|\reg H|-|\reg R|)/2}-\epsilon$.
Substituting $|\reg R|=15|\reg H|$ and $\epsilon\leq1/8$ gives
\[
  \mathbb E_U f_{\mathcal D}(U)\geq1-2^{1-7|\reg H|}-\frac18\geq\frac34.
\]
The last inequality uses $|\reg H|\geq1$.

\paragraph{Applying the query lower bound.}
Since $n=16|\reg H|$, the lower-bound scale is
\[
  \min\{2^{|\reg H|},2^{n/16}\}=2^{|\reg H|}.
\]
\Cref{cor:query-lower-bounds} gives $q=\Omega(2^{|\reg H|})$, including selected queries.
By \cref{lem:uhlmann-rank}, $r=2^{1+|\reg H|}$ almost surely, giving $q=\Omega(r)$.
\end{proof}

\subsection{Comparison with fidelity-estimation lower bounds}
\label{app:fidelity-comparison}

Wang~\cite[Theorem~1.7]{Wan26} already states the nearly linear query lower bound for fidelity estimation.
We derive its Uhlmann-transformation consequence from the reduction in Utsumi et al.~\cite[Proof of Corollary~5]{UNWT25}, and record the hypotheses needed for the query-to-sample comparison.

\paragraph{The transformation hypothesis.}
Fix a $d$-dimensional marginal register, the known target marginal $\sigma=\id/d$, and a purifying register of dimension at least $2d$.
Consider one transformation algorithm whose circuit and initial auxiliary state are independent of the unknown marginal $\rho$ and its preparation unitary.
Suppose that for every $\rho$ and every valid source-preparation unitary, it uses at most $q$ forward and inverse preparation queries and produces a channel acting only on the purifying register, with squared target overlap at least $\mathrm F(\rho,\sigma)^2-\epsilon$.

\paragraph{Estimating fidelity from the transformation.}
We prepare a fresh source, apply the output channel, and measure the projector onto the known pure target.
For this estimation task, the final measurement can act on both the marginal register and its purifying register.
The approximation guarantee gives the lower bound below, and fidelity monotonicity under partial trace~\cite[Theorem~3.27]{Wat18} gives the upper bound:
\[
  \mathrm F(\rho,\sigma)^2-\epsilon\leq\Pr[\text{the target test accepts}]\leq\mathrm F(\rho,\sigma)^2.
\]
Repeating the complete test with fresh registers $O(\epsilon^{-2})$ times estimates its acceptance probability to error $\epsilon$ with probability at least $2/3$.
Clipping the estimate to $[0,1]$ cannot increase its error.
The estimate therefore differs from $\mathrm F(\rho,\sigma)^2$ by at most $2\epsilon$.

Taking its square root gives an estimate of root fidelity with error at most $\sqrt{2\epsilon}$, since $|\sqrt a-\sqrt b|\leq\sqrt{|a-b|}$ for $a,b\geq0$.
At fixed $0<\epsilon\leq2^{-23}$, this error is at most $2^{-11}$.
The proof of Wang's Theorem~5.18 sets its error threshold to $2^{-10}$, so this choice meets the hypothesis.
Constant repetition and taking the median increase the estimator's success probability to at least $0.9$.

Each test uses at most $q+1$ source-preparation queries; the target is known.
Thus a $q$-query transformation would give an $O(q+1)$-query fidelity estimator.

\paragraph{Comparing query and sample complexity.}
Tang--Wright--Zhandry~\cite[Theorem~1.5]{TWZ25} simulate this estimator using $O((q+1)^2)$ copies at trace-distance error $1/20$.
The estimator works for every valid source preparation, and the purifying dimension is at least $2d$, as their theorem requires.
The simulation preserves success above $2/3$, since $0.9-1/20>2/3$.
For root-fidelity error $\eta$ and success probability at least $2/3$, Wang~\cite[Theorem~5.18]{Wan26} gives the sample lower bound
\[
  \Omega\left(\frac{d^2}{\eta^2\log^2d\max\{1,\eta^2\log^2d\}}\right).
\]
At fixed sufficiently small $\eta$, this is $\Omega(d^2/\log^4d)$.
The hard states in that proof have a positive maximally mixed block of dimension at least $d/2$, so their rank is $\Theta(d)$.
This follows from Wang's Definition~3.3 and the direct-sum construction in the proof of Theorem~5.18.
Since $\sigma=\id/d$ is invertible, $\operatorname{rank}(\sqrt\sigma\sqrt\rho)=\operatorname{rank}(\rho)=\Theta(d)$.
Comparing the sample counts gives
\[
  q+1=\Omega(d/\log^2d)=\Omega(r/\log^2r),
\]
and hence $q=\Omega(r/\log^2r)$ as $r\to\infty$.
This argument uses forward and inverse preparation queries; it does not invoke a simulation for conjugate or transpose access.

\section{Security of the cryptographic constructions}
\label{app:efi-proofs}

We prove the EFI estimate, the bit-commitment and quantum-state-commitment bounds, and security for a fixed oracle sequence.

\subsection{The oracle-relative EFI bound}
\label{app:efi-proof}
\label{app:efi-bound}

\begin{proof}[Proof of \cref{cor:efi}]
The rank bound gives statistical separation, and a sign estimate gives the expected absolute gap.

\paragraph{Statistical distance.}
The real state has rank at most $2^\lambda$, since it is obtained by tracing $\reg H$ out of a pure state.
Its support projector $\Pi_U$ accepts it with probability one, whereas $\Tr[\Pi_U\sigma_{\reg{BR}}^\lambda]\leq2^\lambda/2^{1+2\lambda}=2^{-\lambda-1}$.
The variational formula for trace distance therefore gives
\[
  \frac12\|\rho_{\reg{BR}}^U-\sigma_{\reg{BR}}^\lambda\|_1\geq\Tr[\Pi_U(\rho_{\reg{BR}}^U-\sigma_{\reg{BR}}^\lambda)]\geq1-2^{-\lambda-1}.
\]
This projector witnesses statistical distance; the query-limited distinguisher is not given its implementation.

\paragraph{Estimating the sign.}
The query bound controls $|\mathbb E_U\Delta_{\mathcal A}(U)|$, but positive and negative gaps could cancel across oracles.
To bound $\mathbb E_U|\Delta_{\mathcal A}(U)|$, we estimate the sign using $N$ real and $N$ mixed trials while holding the challenge aside.
This auxiliary algorithm may use more than polynomially many queries: the state bound covers its query count and allows arbitrary computation.
The auxiliary algorithm receives $2N+1$ tensor-product copies of the oracle-independent auxiliary state at initialization and uses separate registers for each run.
Conditioned on $U$, the trials are independent of each other, the challenge bit, and the final run of $\mathcal A$.

Let $\widehat\Delta$ be the empirical real mean minus the empirical mixed mean.
The conditional mean is $\Delta_{\mathcal A}(U)$.
Each binary trial has variance at most $1/4$, so independence gives
\[
  \operatorname{Var}(\widehat\Delta\mid U)\leq\frac{N}{4N^2}+\frac{N}{4N^2}=\frac1{2N}.
\]
Let $s=1$ if $\widehat\Delta\geq0$ and $s=-1$ otherwise.
If the estimated sign is correct, $s\Delta_{\mathcal A}(U)=|\Delta_{\mathcal A}(U)|$.
If it is wrong, the left-hand side is $-|\Delta_{\mathcal A}(U)|$, and $|\widehat\Delta-\Delta_{\mathcal A}(U)|\geq|\Delta_{\mathcal A}(U)|$.
In both cases,
\[
  s\Delta_{\mathcal A}(U)\geq|\Delta_{\mathcal A}(U)|-2|\widehat\Delta-\Delta_{\mathcal A}(U)|.
\]
By Cauchy--Schwarz,
\[
  \mathbb E[|\widehat\Delta-\Delta_{\mathcal A}(U)|\mid U]\leq\sqrt{\operatorname{Var}(\widehat\Delta\mid U)}\leq\frac1{\sqrt{2N}}.
\]
Taking the conditional expectation in the sign inequality gives
\[
  \mathbb E[s\Delta_{\mathcal A}(U)\mid U]\geq|\Delta_{\mathcal A}(U)|-\sqrt{2/N}.
\]

\paragraph{Applying the query bound.}
The auxiliary algorithm then runs $\mathcal A$ on the actual challenge and reverses its output when $s=-1$.
The calibration is independent of the challenge bit conditioned on $U$, so its Haar-averaged acceptance gap is at least $\mathbb E_U|\Delta_{\mathcal A}(U)|-\sqrt{2/N}$.
Each real trial costs $q+1$ queries, including one for preparation; each mixed trial costs $q$.
The final challenge run costs $q$, giving
\[
  q_{\mathrm{tot}}=q+N(q+1)+Nq=q+N(2q+1).
\]
An acceptance gap is at most half the trace-norm difference, so \cref{thm:all-haar,cor:selected-queries} give
\begin{equation}
  \mathbb E_U|\Delta_{\mathcal A}(U)|\leq2(25+24\sqrt2)\frac{q_{\mathrm{tot}}}{2^\lambda}+\frac{18(4q_{\mathrm{tot}}+1)^2}{2^{3\lambda/8}}+\sqrt{\frac2N}. \label{eq:efi-calibration}
\end{equation}
For $N=\lceil2^{\lambda/8}\rceil$, the query count is $q_{\mathrm{tot}}=O((q+1)2^{\lambda/8})$.
The first two terms in \cref{eq:efi-calibration} are therefore $O((q+1)2^{-7\lambda/8})$ and $O((q+1)^2 2^{-\lambda/8})$.
The sign-estimation error is $O(2^{-\lambda/16})$.
Each term is at most $O((q+1)^2 2^{-\lambda/16})$, proving the stated bound.
\end{proof}

\paragraph{Polynomially many copies.}
\label{app:efi-copies}
Consider a receiver given $m$ independent copies of either EFI state, with the same sampled $U$, and making $q$ selected queries.
A hybrid replaces the copies one at a time.
Each adjacent pair is distinguished by an algorithm using one supplied challenge and preparing the other $m-1$ copies itself.
Preparing each real copy uses one forward query, and preparing a mixed copy uses none, so this algorithm uses at most $q+m-1$ queries.
For every fixed $U$, these fresh preparations are in tensor product with the supplied challenge.
The triangle inequality and \cref{cor:efi} therefore bound the expected absolute gap by
\[
  m\,O((q+m)^2 2^{-\lambda/16}).
\]
This is negligible for polynomial $m,q$.

\subsection{Quantum bit commitment}
\label{app:bit-commitment}

\begin{proof}[Proof of \cref{cor:bit-commitment}]
The purified generators define the protocol, and the support of the real marginal bounds the two opening probabilities.

\paragraph{The protocol and hiding.}
Let $\reg S$ be an opening register of $1+2\lambda$ qubits.
The two honest commitment states are
\begin{align*}
  \ket{\phi_0^U}_{\reg{BR,S}}&=\ket{\Psi_0^{\mathrm{real},U}}_{\reg{BRH}}\ket{0^{\lambda+1}},\\
  \ket{\phi_1}_{\reg{BR,S}}&=\frac1{\sqrt{2^{1+2\lambda}}}\sum_z\ket z_{\reg{BR}}\ket z_{\reg{S}}.
\end{align*}
In the first state, $\reg S$ consists of $\reg H$ and its padding.
Let $C_b^U$ prepare $\ket{\phi_b^U}$ from zero, where $C_1^U$ is independent of $U$.
The sender sends $\reg{BR}$ and later reveals $b,\reg S$.
The receiver applies $(C_b^U)^\dagger$ and checks for the all-zero state, equivalently measuring $\proj{\phi_b^U}$ on $\reg{BR},\reg S$.
This gives perfect correctness and the stated honest query counts.
It is the canonical construction of~\cite[Lemma~4.7]{BCQ23}.

The two commitment marginals are $\rho_{\reg{BR}}^U$ and $\sigma_{\reg{BR}}^\lambda$.
Hiding follows from \cref{cor:efi}, including oracle-independent auxiliary states.
Receiver preprocessing is also covered: for each fixed $U$, it acts on registers disjoint from the honest preparation and can be placed after that preparation without changing the state.
All its queries remain counted.

\paragraph{The two acceptance spaces.}
Fix any $U$, any common post-commit state on $\reg{BR},\reg Z$, and any two opening channels acting on the sender's workspace $\reg Z$.
Dilate both channels to unitaries $W_0,W_1$ on one sufficiently large workspace with zero ancillas.
The two acceptance projectors pulled back to this common space are
\[
  P_b=(\id_{\reg{BR}}\otimes W_b^\dagger)(\proj{\phi_b^U}_{\reg{BR,S}}\otimes\id_{Z'})(\id_{\reg{BR}}\otimes W_b),\qquad b\in\{0,1\},
\]
where $Z'$ contains the sender's remaining output registers.
Every unit vector in the range of $P_b$ has $\reg{BR}$ marginal $\rho_{\reg{BR}}^U$ for $b=0$ and $\sigma_{\reg{BR}}^\lambda$ for $b=1$.
Indeed, after $W_b$ the vector factors as $\ket{\phi_b^U}\ket\eta$, and $W_b$ acts only on the sender's registers.
Let $\Pi_U$ be the support projector of $\rho_{\reg{BR}}^U$.
For unit vectors $u\in\operatorname{ran}(P_0)$ and $v\in\operatorname{ran}(P_1)$, the marginal identities imply $(\Pi_U\otimes\id)u=u$ and
\[
  \|(\Pi_U\otimes\id)v\|_2^2=\Tr[\Pi_U\sigma_{\reg{BR}}^\lambda]=\frac{\operatorname{rank}(\rho_{\reg{BR}}^U)}{2^{1+2\lambda}}\leq2^{-\lambda-1}.
\]
Cauchy--Schwarz therefore gives
\[
  |\braket uv|=|\bra u(\Pi_U\otimes\id)\ket v|\leq\|u\|_2\|(\Pi_U\otimes\id)v\|_2\leq2^{-(\lambda+1)/2}.
\]
The norm $\|P_0P_1\|$ is the supremum of these inner products, so it obeys the same bound.

\paragraph{Sum binding.}
Let $p_b$ be the probability that opening $b$ is accepted from the common post-commit state.
Their sum is the expectation of $P_0+P_1$, so it is at most $\|P_0+P_1\|$.
For completeness, define $T:v\mapsto(P_0v,P_1v)$ with codomain $\operatorname{ran}(P_0)\oplus\operatorname{ran}(P_1)$.
Then $T^\dagger T=P_0+P_1$, while $TT^\dagger$ has identity diagonal blocks and off-diagonal blocks $P_0P_1$ and $P_1P_0$.
For $u\in\operatorname{ran}(P_0)$ and $v\in\operatorname{ran}(P_1)$, its quadratic form satisfies
\begin{align*}
  \langle(u,v),TT^\dagger(u,v)\rangle&=\|u\|^2+\|v\|^2+2\operatorname{Re}\langle u,P_0P_1v\rangle\\
  &\leq\|u\|^2+\|v\|^2+2\|P_0P_1\|\|u\|\|v\|\\
  &\leq(1+\|P_0P_1\|)(\|u\|^2+\|v\|^2).
\end{align*}

The first inequality is Cauchy--Schwarz; the second uses $2\|u\|\|v\|\leq\|u\|^2+\|v\|^2$.
Thus
\[
  \|P_0+P_1\|=\|T^\dagger T\|=\|TT^\dagger\|\leq1+\|P_0P_1\|.
\]
Combining this norm bound with the support bound proves
\[
  p_0+p_1\leq1+\|P_0P_1\|\leq1+2^{-(\lambda+1)/2}.
\]
This holds for all opening procedures and for every $U$, even when the sender's state and procedures depend on $U$.
\end{proof}

\subsection{Quantum state commitment}
\label{app:qubit-commitment}

\begin{proof}[Proof of \cref{cor:qubit-commitment}]
The channel to $\reg H$ will give swap binding, and the query bound will give hiding.

\paragraph{Correctness.}
Fix $|\reg H|=\lambda$, $|\reg R|=2\lambda$, and a Haar-random $U$ on $3\lambda$ qubits.
The encoding maps $\reg M,\reg W$ to $\reg R,\reg H$, where $\reg M$ is one qubit and $\reg W$ starts in $\ket{0^{3\lambda-1}}$.
The receiver's verification applies $U^\dagger$ and measures $\proj0_{\reg{W}}$.
Honest verification accepts with probability one and leaves the original message and its reference unchanged.
This proves correctness and the honest query counts.

\paragraph{A channel bound for the opening register.}
We bound the channel from the message to $\reg H$ uniformly over inputs entangled with a reference.
Define the channel to $\reg H$ and the replacement channel by
\begin{align*}
  \mathcal E^U(X)&=\Tr_{\reg{R}}[U(X_{\reg M}\otimes\proj0_{\reg{W}})U^\dagger],\\
  \mathcal T(X)&=\Tr(X)\frac{\id_{\reg{H}}}{2^\lambda}.
\end{align*}
With normalized Bell input, their Choi-state difference is
\[
  \epsilon_U=\left\|(\operatorname{id}_{\reg{B}}\otimes\mathcal E^U)(\proj{\Phi^+}_{\reg{BM}})-\frac{\id_{\reg{B}}}{2}\otimes\frac{\id_{\reg{H}}}{2^\lambda}\right\|_1.
\]
The first state is exactly $\rho_{\reg{BH}}^U$, so \cref{eq:decoupling-trace-bound} gives
\begin{equation}
  \mathbb E_U\epsilon_U\leq2\cdot2^{-\lambda/2}. \label{eq:commitment-decoupling}
\end{equation}
Since $\reg M$ is one qubit, we claim that for every state $\omega_{\reg{MZ}}$ and every reference register $\reg Z$,
\begin{equation}
  \|((\mathcal E^U-\mathcal T)\otimes\operatorname{id}_{\reg Z})(\omega_{\reg{MZ}})\|_1\leq2\epsilon_U. \label{eq:commitment-diamond}
\end{equation}
For any pure input $\ket\omega_{\reg{MZ}}$, there is a linear map $A:\reg B\to\reg Z$ such that $\ket\omega=(\id_{\reg{M}}\otimes A)\ket{\Phi^+}_{\reg{MB}}$.
Normalization gives $\Tr(A^\dagger A)=2$ and hence $\|A\|^2\leq2$.
The output difference is the Choi-state difference multiplied on its reference register by $A$ and $A^\dagger$.
The inequality $\|AXB\|_1\leq\|A\|\|X\|_1\|B\|$ therefore gives
\[
  \|((\mathcal E^U-\mathcal T)\otimes\operatorname{id}_{\reg{Z}})(\proj\omega)\|_1\leq\|A\|^2\epsilon_U\leq2\epsilon_U.
\]
Convexity extends this to mixed inputs, with any reference dimension.
This proves \cref{eq:commitment-diamond}, which is the uniform input bound needed for swap binding.

\paragraph{Statistical swap binding.}
In the swap-binding game of~\cite[Definition~4.2]{GJMZ23}, a malicious sender supplies arbitrary $\reg{RHZ}$.
The challenger verifies the commitment.
Rejection has the same probability in both worlds and gives a random guess.
If acceptance has positive probability, the conditional decoded state is $\omega_{\reg{MZ}}\otimes\proj0_{\reg{W}}$, with arbitrary $\omega$ that may depend on $U$.
The challenger either keeps $\reg M$ or replaces it by zero, re-encodes using $U$, and returns only $\reg H$.
The two states available to the sender are respectively
\begin{gather*}
  (\mathcal E^U\otimes\operatorname{id}_{\reg{Z}})(\omega_{\reg{MZ}}),\\
  \mathcal E^U(\proj0_{\reg{M}})\otimes\omega_{\reg Z}.
\end{gather*}
Each is within $2\epsilon_U$ in trace norm of $\id_{\reg{H}}/2^\lambda\otimes\omega_{\reg Z}$, by \cref{eq:commitment-diamond}.
The triangle inequality therefore gives
\[
  \|(\mathcal E^U\otimes\operatorname{id}_{\reg{Z}})(\omega_{\reg{MZ}})-\mathcal E^U(\proj0_{\reg{M}})\otimes\omega_{\reg Z}\|_1\leq4\epsilon_U.
\]
For a uniform challenge bit, distinguishing advantage is at most one quarter of this difference.
Multiplying by the acceptance probability can only decrease the advantage.
Every subsequent sender operation, including further oracle queries, contracts this trace distance for each fixed $U$.
Thus, writing $s_{\mathcal A}(U)$ for a sender's probability of guessing the swap bit correctly, we have the pointwise bound
\[
  \sup_{\mathcal A}\bigl(s_{\mathcal A}(U)-1/2\bigr)\leq\epsilon_U.
\]
It applies to every sender and every accepted state, including those chosen with knowledge of $U$.
Averaging this pointwise bound over the Haar oracle and using \cref{eq:commitment-decoupling} give
\begin{equation}
  \mathbb E_U\sup_{\mathcal A}\bigl(s_{\mathcal A}(U)-1/2\bigr)\leq\mathbb E_U\epsilon_U\leq2\cdot2^{-\lambda/2}. \label{eq:commitment-swap-binding}
\end{equation}
The supremum is inside the expectation and includes unbounded senders with oracle-dependent states and strategies.

\paragraph{Hiding for an EPR challenge.}
We first bound the effect of removing the correlation between $\reg B$ and $\reg R$, and then reduce chosen-message hiding to this bound.
Fix a receiver with an oracle-independent circuit and initial auxiliary state, making at most polynomially many selected queries $q$ in total.
Choose a universal constant $C$ for which the bound in \cref{cor:efi} is at most
\[
  h=C(q+1)^2 2^{-\lambda/16}
\]
for all sufficiently large $\lambda$.
We write $\rho_{\reg{R}}^U=\Tr_{\reg{B}}\rho_{\reg{BR}}^U$.
For any $q$-query algorithm $\mathcal A$, let $\mathcal A'$ first replace $\reg B$ by $\id_{\reg{B}}/2$ and then run $\mathcal A$.
This replacement fixes $\sigma_{\reg{BR}}^\lambda$ and sends $\rho_{\reg{BR}}^U$ to $\id_{\reg{B}}/2\otimes\rho_{\reg{R}}^U$.
For each fixed $U$, the acceptance gap below is $\Delta_{\mathcal A}(U)-\Delta_{\mathcal A'}(U)$.
By \cref{cor:efi}, each of these gaps has expected absolute value at most $h$.
The triangle inequality therefore gives
\begin{equation}
  \mathbb E_U\left|\Pr[\mathcal A^U(\rho_{\reg{BR}}^U)=1]-\Pr[\mathcal A^U(\id_{\reg{B}}/2\otimes\rho_{\reg{R}}^U)=1]\right|\leq2h. \label{eq:commitment-epr-hiding}
\end{equation}
No oracle query is needed for the replacement.

\paragraph{From an EPR challenge to a chosen message.}
In the hiding game of~\cite[Definition~4.3]{GJMZ23}, the receiver prepares a message on $\reg M$ and keeps a workspace $\reg Z$.
The challenger encodes either this message or zero, returns $\reg R$, and keeps $\reg H$.
Message preparation may itself use the oracle; all of these queries count toward $q$.
To reduce this game to \cref{eq:commitment-epr-hiding}, supply the reduction with an EPR challenge on $\reg{BR}$ and leave it isolated during message preparation.
For each fixed $U$, the resulting state $\omega_{\reg{MZ}}^U$ is in tensor product with this challenge, because the preparation acts on disjoint registers.
This conditional product statement is sufficient even though both states depend on $U$.

Define the radiation channel by
\[
  \mathcal N^U(X)=\Tr_{\reg{H}}[U(X_{\reg M}\otimes\proj0_{\reg{W}})U^\dagger].
\]
Measure $\reg{MB}$ in the Bell basis.
To compute this outcome, express the real challenge using its input qubit $\reg{M'}$:
\[
  \rho_{\reg{BR}}^U=(\operatorname{id}_{\reg B}\otimes\mathcal N^U)(\proj{\Phi^+}_{\reg{BM'}}).
\]
Before applying $\mathcal N^U$, the $\Phi^+$ contraction for a pure message is
\[
  (\bra{\Phi^+}_{\reg{MB}}\otimes\id)(\ket\psi_{\reg{MZ}}\otimes\ket{\Phi^+}_{\reg{BM'}})=\frac12\ket\psi_{\reg{M'Z}}.
\]
It follows by expanding $\ket\psi=\sum_{c=0}^1\ket c_{\reg{M}}\ket{\psi_c}_{\reg Z}$ and using $\braket{c}{c'}=\delta_{c,c'}$.
The corresponding identity for density operators extends to mixed inputs by linearity.
Encoding acts on $\reg{M'W}$, whereas the Bell projection acts on $\reg{MB}$, so these operations commute.
Applying $\mathcal N^U$ to $\reg{M'}$ therefore gives the unnormalized real state after this outcome:
\[
  \frac14(\mathcal N^U\otimes\operatorname{id}_{\reg{Z}})(\omega_{\reg{MZ}}^U).
\]
For the decorrelated challenge $\id_{\reg{B}}/2\otimes\rho_{\reg{R}}^U$, the same contraction instead gives
\[
  \frac14\rho_{\reg{R}}^U\otimes\omega_{\reg Z}^U.
\]
Both states have trace $1/4$ for every fixed $U$ and every message.
Thus the successful outcome does not bias the oracle distribution.

\paragraph{Completing the hiding reduction.}
On this outcome, the reduction runs the receiver's remaining computation on $\reg{RZ}$.
On the other outcomes, it outputs a fixed bit.
Its acceptance gap is exactly one quarter of the gap between encoding the chosen message and returning $\rho_{\reg{R}}^U$ independently of that message.
The reduction uses at most $q$ queries, including message preparation.
\Cref{eq:commitment-epr-hiding} therefore bounds the expected absolute gap in this comparison by $8h$.

For encoding zero, the reduction first discards the prepared $\reg M$ and substitutes $\proj0_{\reg{M}}$, preserving $\omega_{\reg Z}^U$.
The same calculation bounds its gap against $\rho_{\reg{R}}^U\otimes\omega_{\reg Z}^U$ by $8h$.
The triangle inequality gives
\begin{equation}
  \mathbb E_U|p_{\mathrm{message}}(U)-p_{\mathrm{zero}}(U)|\leq16h, \label{eq:commitment-message-hiding}
\end{equation}
where the two probabilities are the receiver's acceptance probabilities in the chosen-message and zero-message games.
This proves hiding for oracle-dependent messages generated through counted queries, including arbitrary entanglement with $\reg Z$.
The reduction's circuit and initial auxiliary state are independent of the sampled $U$; intermediate states may depend on $U$ through the counted queries.
\end{proof}

\subsection{Security for a fixed oracle sequence}
\label{app:efi-security}

\begin{corollary}
\label{cor:efi-fixed-oracle}
Independently sample $U_\lambda$ from Haar measure for every $\lambda\geq1$.
With probability one over this sequence, every uniform polynomial-query oracle algorithm whose circuit and initial-state preparation have finite descriptions satisfies
\[
  |\Delta_{\mathcal A_\lambda}(U_\lambda)|\leq2^{-\lambda/32}
\]
for all sufficiently large $\lambda$.
At security parameter $\lambda$, its oracle access is to $U_\lambda$ alone.
\end{corollary}
\begin{proof}
We use Markov's inequality for each algorithm and then countability to obtain the simultaneous statement.

\paragraph{One fixed algorithm.}
Fix one such algorithm and a polynomial bound $q_\lambda$ on its query count.
By Markov's inequality and \cref{cor:efi},
\[
  \Pr_{U_\lambda}\bigl[|\Delta_{\mathcal A_\lambda}(U_\lambda)|>2^{-\lambda/32}\bigr]=O((q_\lambda+1)^2 2^{-\lambda/32}).
\]
This is summable in $\lambda$, since $q_\lambda$ is polynomial.
The first Borel--Cantelli lemma implies that only finitely many of these events occur, with probability one.

\paragraph{All uniform algorithms.}
There are countably many such algorithms and initial-state preparation procedures.
The intersection of their probability-one events still has probability one.
The threshold in $\lambda$ may depend on the algorithm.
This simultaneous statement is restricted to uniform algorithms; the averaged statement of \cref{cor:efi} applies to each oracle-independent circuit family separately.
\end{proof}

The same argument, using the hiding bounds of \cref{cor:bit-commitment,cor:qubit-commitment}, gives almost-sure computational hiding of both protocols against every such uniform receiver.

\paragraph{Statistical binding.}
For statistical swap binding, \cref{eq:commitment-decoupling} and Markov's inequality give
\[
  \Pr_U[\epsilon_U>2^{-\lambda/4}]\leq2\cdot2^{-\lambda/4}.
\]
These probabilities are summable.
Almost every oracle sequence therefore has swap-binding advantage at most $2^{-\lambda/4}$ for all sufficiently large $\lambda$, uniformly over all sender strategies by the pointwise bound preceding \cref{eq:commitment-swap-binding}.
Bit-commitment binding already holds for every oracle by \cref{cor:bit-commitment}.

\end{document}